\documentclass[acmlarge]{acmart}
\usepackage{amsmath,amssymb,amsthm}
\usepackage{mathpartir}
\usepackage{stmaryrd}
\usepackage{soul}
\usepackage{listings}
\usepackage{proof}
\usepackage{color}
\usepackage{hyperref}
\usepackage{float}
\usepackage{rotating}
\usepackage{mdframed}
\usepackage{enumitem}
\usepackage{xifthen}

\newcommand{\rname}[1]{[\textsf{#1}]}
\newcommand{\jrhol}[7]{#1 \mid #2 \vdash #3 : #4 \sim #5 : #6 \mid #7}
\newcommand{\jrholu}[5]{#1 \mid #2 \vdash #3 \sim #4 \mid #5}
\newcommand{\sjrhol}[7]{#2 \vdash #3 \sim #5 \mid #7}
\newcommand{\juhol}[5]{#1 \mid #2 \vdash #3 : #4 \mid #5}

\newcommand{\jhol}[3]{#1 \mid #2 \vdash #3}

\newcommand{\jholn}[3]{#1 \mid #2 \vdash_{\sf HOL} #3}
\newcommand{\jlc}[3]{#1 \vdash #2 : #3}

\newcommand{\subst}[2]{[#2/ #1]}
\newcommand{\ltag}{_1}
\newcommand{\rtag}{_2}
\newcommand{\res}{\mathbf{r}}
\newcommand{\sem}[1]{\llbracket #1 \rrbracket}
\newcommand{\semv}[2]{\llbracket #1 \rrbracket_{#2} }

\newcommand{\semexpv}[2]{\llparenthesis #1 \rrparenthesis_{#2} }

\newcommand{\pair}[2]{\langle #1, #2 \rangle}
\newcommand{\nat}{\mathbb{N}}

\newcommand{\listt}[1]{\mathsf{list}_{#1}}
\newcommand{\bool}{\mathbb{B}}
\newcommand{\tbool}{\mathsf{tt}}
\newcommand{\fbool}{\mathsf{ff}}
\newcommand{\nil}{[]}
\newcommand{\cons}[2]{#1 :: #2}
\newcommand{\casebool}[3]{{\rm case}\ #1\ {\sf of}\ \tbool \mapsto #2 ; \fbool \mapsto #3}
\newcommand{\casenat}[3]{{\rm case}\ #1\ {\sf of}\ 0 \mapsto #2 ; S \mapsto #3}
\newcommand{\caselist}[3]{{\rm case}\ #1\ {\rm of}\ \nil \mapsto #2 ; \cons{\_}{\_} \mapsto #3}  
\newcommand{\defsubst}[2]{\subst{\res\ltag}{#1}\subst{\res\rtag}{#2}}

\newcommand{\wdef}[3]{\mathcal{D}\hspace{-0.2mm}\mathit{ef}(#1,#2,#3)}
\newcommand{\letrec}[3]{{\rm letrec}\ {\mathit #1}\ #2 = #3}

\newcommand{\cost}[1]{\llparenthesis #1 \rrparenthesis}
\newcommand{\costt}[1]{\lfloor #1 \rfloor}
\newcommand{\rcost}[1]{\| #1 \|}
\newcommand{\rcostt}[1]{\llfloor #1 \rrfloor}
\newcommand{\defeq}{\triangleq}
\newcommand{\erase}[1]{| #1 |}

\newcommand{\reftype}[3]{\ifthenelse{\isempty{#1}}{\{\res : #2 \mid #3\}}{\{#1 : #2 \mid #3\}}}
\newcommand{\ertype}[2]{\ifthenelse{\isempty{#2}}{\lfloor #1 \rfloor}{\lfloor #1 \rfloor(#2)}}
\newcommand{\errtype}[2]{\ifthenelse{\isempty{#2}}{\llfloor #1 \rrfloor}{\llfloor #1 \rrfloor(#2)}}
\newcommand{\erfor}[1]{| #1 |}

\newcommand{\emfor}[2]{\ifthenelse{\isempty{#2}}{#1}{#1\subst{\res}{#2}}}
\newcommand{\emterm}[2]{\ifthenelse{\isempty{#2}}{#1}{#1\subst{\res}{#2}}}
\newcommand{\emtyj}[3]{\juhol{\erfor{#1}}{\ertype{#1}{}}{#2}{\erfor{#3}}{\ertype{#3}{\res}}}
\newcommand{\jreft}[3]{#1 \vdash #2 : #3}
\newcommand{\pitype}[3]{\Pi(#1:#2).#3}

\newcommand{\ttmonad}[1]{\mathbb{T}_{#1}}
\newcommand{\eret}[1]{\eta_{#1}}
\newcommand{\ebind}{\mathsf{bind}}
\newcommand{\dccform}[2]{\lfloor #2 \rfloor_{#1}}
\newcommand{\adversary}{a}

\newcommand{\wexec}{\mbox{\scriptsize exec}}
\newcommand{\wdiff}{\mbox{\scriptsize diff}}
\newcommand{\ldiff}{\lesssim}
\newcommand{\uarr}[2]{\mathrel{\xrightarrow[]{\wexec(#1,#2)}}} 
\newcommand{\tforall}[3]{\forall#1\overset{\wexec(#2,#3)}{::}S.\,}

\newcommand{\tarrd}[1]{\mathrel{\xrightarrow{\wdiff(#1)}}}
\newcommand{\tforalld}[2]{\forall#1\overset{\wdiff(#2)}{::}S.\,}
\newcommand{\tbox}[1]{\square\,#1}
\newcommand{\jtype}[4]{\mathrel{\vdash_{#1}^{#2} {#3} : #4}}
\newcommand{\jtypediff}[4]{\mathrel{\vdash
    {#2} \ominus {#3} \ldiff #1 : {#4}}}
\newcommand{\econs}{\mbox{cons}}

\newcommand{\ctx}{\Delta; \Phi_a; \Gamma}
\newcommand{\kw}[1]{\mathtt{#1}}
\newcommand{\econst}{\kw{n}}
\newcommand{\trint}{\mbox{int}_r}
\newcommand{\tlist}[3]{\mbox{list}[#1]^{#2}\,#3}
\newcommand{\ecase}{\mbox{\;case\;}} 
\newcommand{\eof}{\mbox{\;of\;}}

\newcommand{\enil}{\mbox{nil\;}} 
\newcommand{\fiv}[1]{\text{FIV}(#1)}
\newcommand{\eLam}{ \Lambda}
 
\newcommand{\eApp}{ [\,]\,}
\newcommand{\sat}[1]{\models#1}
\newcommand{\jsubtype}[2]{\sat#1\sqsubseteq#2}

\newcommand{\grt}{A}
\newcommand{\tcho}[1]{U\,#1} 
\newcommand{\octx}{\Delta; \Phi_a; \Omega}
\newcommand{\tint}{\mbox{int}}
\newcommand{\kapp}{c_{app}}

\def\tyle{\preceq}

\theoremstyle{definition}
\newtheorem{thm}{Theorem}
\newtheorem{lem}[thm]{Lemma}
\newtheorem{cor}[thm]{Corollary}

\newcommand{\vbar}[0]{\mathrel{|}}
\newcommand{\rtprod}[3]{\Pi ({#1} :: {#2}) .\, {#3}}
\newcommand{\rtref}[3]{\{ {#1} :: {#2} \vbar {#3} \}}

\newcommand{\rinterp}[2]{\llparenthesis {#2} \rrparenthesis_{#1}}

\begin{document}

\title{A Relational Logic for Higher-Order Programs}
\author{Alejandro Aguirre}
\affiliation{IMDEA Software Institute}
\author{Gilles Barthe}
\affiliation{IMDEA Software Institute}
\author{Marco Gaboardi}
\affiliation{University at Buffalo, SUNY}
\author{Deepak Garg}
\affiliation{MPI-SWS}
\author{Pierre-Yves Strub}
\affiliation{\=Ecole Polytechnique}
\acmYear{2017}
\acmMonth{2}
\setcopyright{none}


\begin{abstract}
Relational program verification is a variant of program verification
where one can reason about two programs and as a special case about
two executions of a single program on different inputs. Relational
program verification can be used for reasoning about a broad range of
properties, including equivalence and refinement, and specialized
notions such as continuity, information flow security or relative
cost. In a higher-order setting, relational program verification can
be achieved using relational refinement type systems, a form of
refinement types where assertions have a relational interpretation.
Relational refinement type systems excel at relating structurally
equivalent terms but provide limited support for relating terms with
very different structures.

We present a logic, called Relational Higher Order Logic (RHOL), for
proving relational properties of a simply typed $\lambda$-calculus
with inductive types and recursive definitions. RHOL retains the
type-directed flavour of relational refinement type systems but
achieves greater expressivity through rules which simultaneously
reason about the two terms as well as rules which only contemplate one
of the two terms. We show that RHOL has strong foundations, by proving
an equivalence with higher-order logic (HOL), and leverage this
equivalence to derive key meta-theoretical properties: subject
reduction, admissibility of a transitivity rule and set-theoretical
soundness.  Moreover, we define sound embeddings for several existing
relational type systems such as relational refinement types and type
systems for dependency analysis and relative cost, and we verify
examples that were out of reach of prior work.
\end{abstract}

\maketitle

\section{Introduction}
Many important aspects of program behavior go beyond the traditional
characterization of program properties as sets of
traces~\cite{AlpernS85}. Hyperproperties~\citep{ClarksonS08}
generalize properties and capture a larger class of program behaviors,
by focusing on sets of sets of traces. As an intermediate point in
this space, relational properties are sets of pairs of
traces. Relational properties encompass many properties of interest,
including program equivalence and refinement, as well as more specific
notions such as non-interference and continuity.

Relational verification is an instance of program verification that
targets relational properties. Expectedly, standard verification
methods such as type systems, program logics, and program analyses can
be lifted to a relational setting. However, it remains a challenge to
devise sufficiently powerful methods that can be used to verify a
broad range of examples. In effect, most existing relational
verification methods are limited in the examples that they can
naturally verify, due to the fundamental tension between the
syntax-directed nature of program verification, and the need to relate
structurally different programs. Moreover, approaches to resolve this
tension highly depend on the programming paradigm, on the class of
program properties considered, and on the verification method. In the
(arguably simplest) case of deductive verification of general
properties of imperative programs, one approach to reduce this tension
is to use self-composition~\cite{BartheDR04}, which reduces relational
verification to standard verification. However, reasoning about
self-composed programs might be cumbersome. Alternatively, there exist
expressive relational program logics that rely on an intricate set of
rules to reason about a pair of programs. These logics combine
two-sided rules, in which the two programs have the same top-level
structure, and one-sided rules, which operate on a single
program. Rules for loops are further divided into synchronous, in
which both programs perform the same number of iterations, and
asynchronous rules, that do not have this restriction but introduce
more complexity~\cite{Benton04,BartheGHS17}.

In contrast, deductive verification of general properties of (pure)
higher-order programs is less developed. One potential approach to
solve the tension between the syntax-directedness, and the need to
relate structurally different programs, is to reduce relational
verification of pure higher-order programs to proofs in higher-order
logic. There are strong similarities between this approach and
self-composition: it reduces relational verification to standard
verification, but this approach is very difficult to use in
practice. A better alternative is to use relational refinement types
such as $\mathrm{rF}^*$~\citep{BFGSSZ14},
$\mathrm{HOARe}^2$~\citep{BGGHRS15}, DFuzz~\citep{GaboardiHHNP13} or
RelCost~\citep{CBGGH17}. Informally, relational refinement type
systems use assertions to capture relationships between inputs and
outputs of two higher-order programs. They are appealing for two
reasons:
\begin{itemize}
\item They capture many important properties of programs in a direct
  and intuitive manner. For instance, the type $\rtref{x}{\nat}{x\ltag
    \leq x\rtag}\rightarrow \rtref{y}{\nat}{y\ltag \leq y\rtag}$
  captures the set of pairs of functions that preserve the natural
  order on natural numbers, i.e.\, pairs of functions $f_1, f_2: \nat
  \rightarrow \nat$ such that for every $x_1,x_2\in\nat$, $x_1\leq
  x_2$ implies $f_1(x_1)\leq f_2(x_2)$. (The subscripts $1$ and $2$ on
  a variable refer to its values in the two runs.)

\item They can potentially benefit from a long and successful line of
  foundational~\citep{FreemanP91,XiP99,DunfieldP04,MelliesZ15} and
  practical~\citep{VazouSJVJ14,Swamy+16} research on refinement types.
\end{itemize}
Unfortunately, existing relational refinement type systems fail to
support the verification of several examples. Broadly speaking, the
two programs in a relational judgment may be required to have the same
type and the same control flow; moreover, this requirement must be
satisfied by their subprograms: if the two programs are applications,
then the two sub-programs in argument position (resp. in function
position) must have the same type and the same control flow; if the
two programs are case expressions, they must enter the same branch,
and their branches must themselves have the same control flow; if the
two programs are recursive definitions, then their bodies must perform
the same sequence of recursive calls; etc. This restriction, which can
be found in more or less strict forms in the different relational type
systems, limits the ability to carry fine-grained reasoning about
terms that are structurally different.  This raises the question
whether the type-directed form of reasoning purported by refinement
types can be reconciled with an expressive relational verification of
higher-order programs. We provide a positive answer for pure
higher-order programs; extending our results to effectful programs is
an important goal, but we leave it for future work.

Our starting point is the observation that relational refinement type
systems are inherently restricted to reasoning about two structurally
similar programs, because relational assertions are embedded into
types. In order to provide broad support for one-sided rules (i.e.,
rules that contemplate only one of the two expressions), it is
therefore necessary to consider relational assertions at the
top-level, since one-sided rules have a natural formulation in this
setting.  Considering relational assertions at the top-level can be
done in two different ways: either by supporting a rich theory of
subtyping for relational refinement types, in such a way that each
type admits a normal form where refinements only arise at the
top-level, or simply by adapting the definitions and rules of
refinement type systems so that only the top-level refinements are
considered. Although both approaches are feasible, we believe that the
second approach is more streamlined and leads to friendlier
verification environments.

\subsubsection*{Contributions}
We present a new logic, called Relational Higher Order Logic (RHOL,
\S~\ref{sec:rhol}), for reasoning about relational properties of
higher-order programs written in a variant of Plotkin's PCF
(\S~\ref{sec:pcf}). The logic manipulates judgments of the form:
\[\jrhol{\Gamma}{\Psi}{t_1}{\sigma_1}{t_2}{\sigma_2}{\phi}\]
where $\Gamma$ is a simply typed context, $\sigma_1$ and $\sigma_2$
are (possibly different) simple types, $t_1$ and $t_2$ are terms,
$\Psi$ is a set of assertions, and $\phi$ is an assertion. Our logic
retains the type-directed nature of (relational) refinement type
systems, and features typing rules for reasoning about structurally
similar terms. However, disentangling types from assertions also makes
it possible to define type-directed rules operating on a single term
(left or right) of the judgment. This confers great expressivity to
the logic, without significantly affecting its type-directed nature,
and opens the possibility to alternate freely between two-sided and
one-sided reasoning, as done in first-order imperative languages.

The validity of judgments is expressed relative to a set-theoretical
semantics---our variant of PCF is restricted to terms which admit a
set-theoretical semantics, including strongly normalizing terms. More
precisely, a judgment
$\jrhol{\Gamma}{\Psi}{t_1}{\sigma_1}{t_2}{\sigma_2}{\phi}$ is valid if
for every valuation $\rho$ (mapping variables in the context $\Gamma$
to elements in the interpretation of their types), the interpretation
of $\phi$ is true whenever the interpretation of (all the assertions
in) $\Psi$ is true. Soundness of the logic can be proved through a
standard model-theoretic argument; however, we provide an alternative
proof based on a sound and complete embedding into Higher-Order Logic
(HOL, \S~\ref{sec:hol}). We leverage this equivalence to establish
several meta-theoretical properties of the logic, notably subject
reduction.
 
Moreover, we demonstrate that RHOL can be used as a general framework,
by defining sound embedding for several relational type systems:
relational refinement types (\S~\ref{ssec:emb-rrt}), the Dependency
Core Calculus (DCC) for many dependency analyses, including those for
information flow security (\S~\ref{sec:dcc}), and the RelCost
(\S~\ref{ssec:emb-rc}) type system for relative cost. The embedding of
RelCost is particularly interesting, since it exercises the ability of
our logic to alternate between synchronous and asynchronous
reasoning. Finally, we verify several examples that go beyond the
capabilities of previous systems (\S~\ref{sec:ex}).

\subsubsection*{Related work}
While dependent type theory is the prevailing approach to reason about
(pure) higher-order programs, several authors have explored another
approach, which is crisply summarized by~\citet{JacobsCLTT}: \lq\lq A
logic is always a logic over a type theory\rq\rq. Formalisms following
this approach are defined in two stages; the first stage introduces a
(dependent) type theory for writing programs, and the second stage
introduces a predicate logic to reason about programs. This approach
has been pursued independently in a series of works on logic-enriched
type theories~\cite{Dybjer85,AczelG00,AczelG06,Belo07,AdamsL10}, and
on refinement types~\citet{Pfenning08andrews,Zeilberger16}. In the
latter line of work, programs are written in an intrinsically typed
$\lambda$-calculus \`a la Church; then, a system of sorts (a.k.a.\,
refinements) is used to establish properties of programs typable in
the first system. Our approach is similar; however, these works are
developed in a unary setting, and do not consider the problem of
relational verification.

Moreover, there is a large body of work on relational verification; we
focus on type-based methods and deductive methods. Relational Hoare
Logic~\citep{Benton04} and Relational Separation Logic~\citep{Yang07}
are two program logics, respectively based on Hoare Logic and
Separation Logic, for reasoning about relational properties of
(first-order) imperative programs. These logics have been used for a
broad range of examples and applications, ranging from program
equivalence to compiler verification and information flow
analysis. Moreover, they have been extended in several directions. For
example, Probabilistic Relational Hoare Logic~\citep{BartheGB09} and
approximate probabilistic Relational Hoare Logic~\citep{BartheKOB12}
are generalizations of Relational Hoare logic for reasoning about
relational properties of (first-order) probabilistic programs. These
logics have been used for a broad range of applications, including
probabilistic information flow, side-channel security, proofs of
cryptographic strength (reductionist security) and differential
privacy. Cartesian Hoare Logic~\citep{SousaD16} is also a recent
generalization of Relational Hoare Logic for reasoning about bounded
safety (i.e.\, $k$-safety for arbitrary but fixed $k$) properties of
(first-order) imperative programs. This logic has been used for
analyzing standard libraries. Experiments have demonstrated that such
logics can be very effective in practice. Our formalism can be seen as
a proposal to adapt their flexibility to pure higher-order programs.

Product programs~\citep{BartheDR04,TerauchiA05,ZaksP08,BartheCK11} are
a general class of constructions that emulate the behavior of two
programs and can be used for reducing relational verification to
standard verification.  While product programs naturally achieve
(relative) completeness, they are often difficult to use since they
require global reasoning on the obtained program---however recent
works~\citep{BlatterKGP16} show how this approach can be automated in
specific settings. Building product programs for (pure) higher-order
languages is an intriguing possibility, and it might be possible to
instrument RHOL using ideas from~\citet{BartheGHS17} to this effect;
however, the product programs constructed in~\citep{BartheGHS17} are a
consequence, rather than a means, of relational verification.

Several type systems have been designed to support formal reasoning
about relational properties for functional programs. Some of the
earlier works in this direction have focused on the semantics
foundations of parametricity, like the work by \citet{AbadiCC93} on
System R, a relational version of System F. The recent work by
\citet{GhaniFS16} has further extended this approach to give better
foundations to a combination of relational parametricity and
impredicative polymorphism.  Interestingly, similarly to RHOL, System
R also supports relations between expressions at different types,
although, since System R does not support refinement types, the only
relations that System R can support are the parametric ones on
polymorphic terms. In RHOL, we do not support parametric polymorphism
\`a la System F currently but the relations that we support are more
general. Adding parametric polymorphism will require foregoing the
set-theoretical semantics, but it should still be possible to prove
equivalence with a polymorphic variant of higher-order logic.

Several type systems have been proposed to reason about information
flow security, a prime example of a relational property. Some examples
include SLAM~\citep{HeintzeR98}, the type system underlying Flow
Caml~\citep{PottierS02} and DCC~\citep{AbadiBHR99}.  Most of these
type systems consider only one expression but they allow the use of
information flow labels to specify relations between two different
executions of the expression. As we show in this paper, this approach
can also be implemented in RHOL. We show how to translate DCC since it
is one of the most general type systems; however, similar translations
can also be provided for the other type systems.

Relational Hoare Type Theory (RHTT)~\citep{NanevskiBG13,StewartBN13}
is a formalism for relational reasoning about stateful higher-order
programs. RHTT was designed to verify security properties like
authorization and information flow policies but was used for the
verification of hetergenous pointer data structures as well. RHTT uses
a monad to separate stateful computations and relational refinements
on the monadic type express relational pre- and post-conditions. RHTT
supports reasoning about two different programs but the programs must
have the same types at the top-level. RHTT's rules support both two-
and one-sided reasoning similar to RHOL, but the focus of RHTT is on
verifying properties of the program state. In particular, examples
such as those in \S\ref{sec:ex} or embeddings such as those in
\S\ref{sec:emb} were not considered in RHTT. RHTT is proved sound over
a domain-theoretic model and continuity must be proven explicitily
during the verification of recursive functions (rules are provided to
prove continuity in many cases). In contrast, RHOL's set-theoretic
model is simpler, but admits only those recursive functions that have
a unique interpretation in set-theory.

Logical relations~\citep{plotkin1973lambda,statman1985logical,Tait67}
provide a fundamental tool for reasoning about programs. They have
been used for a broad range of purposes, including proving unary
properties (for instance strong normalization or complexity) and
relational properties (for instance equivalence or information flow
security). Our work can be understood as an attempt to internalize the
versatility of relational logical relations in a syntactic framework. 
There is a large body of works on logic for logical relations, from
the early works by~\citet{PlotkinA93} to more recent works on logics
for reasoning about states and concurrency by Ahmed, Birkedal, Dreyer,
and collaborators among
others~\citep{DreyerAB11,DreyerNRB10,JungSSSTBD15,Krogh-Jespersen17}. In
particular, the IRIS logic~\citep{JungSSSTBD15} can be seen as a
powerful reasoning framework for logical relations, as shown recently
by~\citet{Krogh-Jespersen17} . Even if we also aim at internalize the
logical relations, the goal of RHOL differ from the one of IRIS in the
fact that we aim for syntax-driven relational verification.

We have already mentioned the works on relational refinement type
systems for verifying cryptographic  constructions~\citep{BFGSSZ14}, for
differential privacy~\citep{BGGHRS15,GaboardiHHNP13} and for
relational cost analysis~\citep{CBGGH17}. This line of works is probably the
most related to our work, however RHOL improves over all of them, as
also shown by some of the embedding we give in Section~\ref{sec:emb}. 
Another work related to this direction is the one
by~\citet{asada2016verifying}. This work proposes a technique to
reduce relational refinement to standard first order
refinements. Their technique is incomplete but it works well on some
concrete examples. As we discussed before, we believe that some
technique of this kind can be applied also to RHOL however this is
orthogonal to our goal and we leave it for future investigations.

\section{(A variant of) PCF}
\label{sec:pcf}

We consider a variant of PCF~\citep{plotkin1977lcf} with booleans,
natural numbers, lists and recursion, and recursive definitions. For
the latter, we require that all recursive calls are performed on
stricly smaller elements---as a consequence, the fixpoint equation
derived from the definition has a unique set-theoretical solution. The
precise method to enforce this requirement is orthogonal to our
purposes, and could for instance be based on a syntactic guard
predicate, or on sized types.

Types are defined by the grammar:
\[\tau, \sigma ::= \bool \mid \nat \mid \listt{\tau} \mid \tau\times \tau \mid
\tau \rightarrow \tau\]
We let $I$ range over inductive types.

Terms of the language are defined by the grammar:
\begin{align*}
  t &::= x \mid \pair{t}{t} \mid \pi_1\ t \mid \pi_2\ t \mid t\ t\mid
  \lambda x:\tau.t \mid c \mid S~t \mid \cons{t}{t} \mid \casenat{t}{t}{t} \mid \casebool{t}{t}{t} \\
  & \mid  \caselist{t}{t}{t} \mid \letrec{f}{x}{t}
\end{align*}
where $x$ ranges over a set $V$ of variables, $c$ ranges over the set
$\{ \tbool, \fbool, 0, []\}$ of constants, and $\lambda$-abstractions are \emph{\`a la} Church. The operational behavior of
terms is captured by $\beta\iota\mu$-reduction $\rightarrow_{\beta\iota\mu}
=\rightarrow_{\beta}\cup \rightarrow_{\iota} \cup \rightarrow_{\mu}$, where $\beta$-reduction,
$\iota$-reduction and $\mu$-reduction are defined as the contextual closure of:
\[\begin{array}{rcllcl}
  (\lambda x.t)\ u & \rightarrow_\beta & t\subst{x}{u}\quad\:   &\casebool{\fbool}{u}{v} & \rightarrow_\iota & v \\
  \pi_i \pair{t_1}{t_2} & \rightarrow_\beta & t_i \quad\:        &\caselist{\nil}{u}{v} & \rightarrow_\iota & u \\
  \casenat{0}{u}{v} & \rightarrow_\iota & u  \quad\:            & \caselist{\cons{h}{t}}{u}{v} & \rightarrow_\iota & (v\ h\ t)\\
  \casenat{S t}{u}{v} & \rightarrow_\iota & (v\ t)  \quad\:     & (\letrec{f}{x}{t})\ (C\ \vec{t}) & \rightarrow_\mu & t\subst{x}{C\ \vec{t}}\subst{f}{\letrec{f}{x}{t}} \\
  \casebool{\tbool}{u}{v} & \rightarrow_\iota & u \\
  \end{array}\]
where $t\subst{x}{u}$ denotes the usual (capture-free) notion of
substitution on terms (replace $x$ by $u$ in $t$).  As usual, we let
$=_{\beta\iota\mu}$ denote the reflexive-symmetric-transitive closure
of $\rightarrow_{\beta\iota\mu}$. In particular, we only
allow reduction of ${\rm letrec}$ when the argument has a constructor
$C\in\{\tbool, \fbool, 0, S, \nil, ::\}$ in head position.

Judgments are of the form $\Gamma\vdash t:\tau$, where $\Gamma$ is a
set of typing declarations of the form $x:\sigma$, such that variables
are declared at most once. The typing rules are standard, except for
recursive functions. In this case, the typing rule requires that the
domain of the recursive function is an inductive type (booleans,
naturals, or lists here) and that the body of the recursive definition
$\letrec{f}{x}{e}$ satisfies a predicate $\wdef{f}{x}{e}$ which
ensures that
all recursive calls are performed on smaller arguments. The typing
rule for recursive definitions is thus:
$$\inferrule*{\Gamma, f:I\to{\sigma},x:I \vdash e:\sigma \qquad
  \wdef{f}{x}{e} \qquad I \in \{\nat,\listt{\tau}\}}{\Gamma \vdash
  {\rm letrec}\ f\ x\ = e : I\to \sigma}$$
The other rules are standard.
We give set-theoretical semantics to this system. For each type
$\tau$, its interpretation $\sem{\tau}$ is the set of its values:
$$\begin{array}{r@{}c@{}lr@{}c@{}lr@{}c@{}lr@{}c@{}l}
  \sem{\bool} & \:\defeq\: & \bool &
  \sem{\nat} & \:\defeq\: & \nat &
  \sem{\listt{\tau}} & \:\defeq\: & \listt{\sem{\tau}} &
  \sem{\sigma\rightarrow\tau} & \:\defeq\: & \sem{\sigma}\rightarrow\sem{\tau}
\end{array}$$
where $\sem{\sigma}\rightarrow\sem{\tau}$ is the set of total
functions with domain $\sem{\sigma}$ and codomain $\sem{\tau}$.

A valuation $\rho$ for a context $\Gamma$ (written $\rho\models
\Gamma$) is a partial map such that $\rho(x)\in\sem{\tau}$ whenever
$(x:\tau)\in \Gamma$. Given a valuation $\rho$ for $\Gamma$, every
term $t$ such that $\Gamma\vdash t:\tau$ has an interpretation
$\semexpv{t}{\rho}$:

\begin{mathpar}
\semexpv{x}{\rho} \defeq \rho(x) \and \semexpv{\pair{t}{u}}{\rho} \defeq \pair{\semexpv{t}{\rho}}{\semexpv{u}{\rho}} \and
\semexpv{\pi_i\ t}{\rho} \defeq \pi_i(\semexpv{t}{\rho}) \and
\semexpv{\lambda x : \tau. t} \defeq \lambda v : \sem{\tau}. \semexpv{x}{\rho\subst{v}{\semexpv{v}{\rho}}} \and
\semexpv{c}{\rho} \defeq c \and \semexpv{S~t}{\rho} \defeq S~\semexpv{t}{\rho} \and
\semexpv{\cons{t}{u}}{\rho} \defeq \cons{\semexpv{t}{\rho}}{\semexpv{u}{\rho}} \and
\semexpv{\caselist{t}{u}{v}}{\rho} \defeq  \caselist{\semexpv{t}{\rho}}{\semexpv{u}{\rho}}{\semexpv{v}{\rho}} \and
\semexpv{\letrec{f}{x}{t}}{\rho} \defeq F

\end{mathpar}
In the case of $\letrec{f}{x}{e}$, we require that $F$ is the unique
solution of the fixpoint equation extracted from the recursive
definition---existence and unicity of the solution follows from the
validity of the $\wdef{f}{x}{e}$ predicate.

The interpretation of well-typed terms is sound. Moreover, the
interpretation equates convertible terms. (This extends to
$\eta$-conversion.)
\begin{theorem}[Soundness of set-theoretic semantics]\mbox{}
\begin{itemize}
\item If $\Gamma\vdash t:\tau$ and $\rho \models \Gamma$, then
  $\semexpv{t}{\rho} \in \sem{\tau}$.
\item If $\Gamma\vdash t:\tau$ and $\Gamma\vdash u:\tau$ and
  $t=_{\beta\iota\mu} u$ and $\rho \models \Gamma$, then
  $\semexpv{t}{\rho}=\semexpv{u}{\rho}$.
\end{itemize}
\end{theorem}

\section{Higher-Order Logic}
\label{sec:hol}
Higher-Order Logic is defined as a calculus in natural deduction for a predicate
logic over simply-typed terms. More specifically, its assertions are
formulae over typed terms, and are defined by the following grammar:
\[\phi::= P(t_1,\ldots, t_n) \mid \top \mid \bot \mid \phi \land
\phi \mid \phi \lor \phi \mid \phi \Rightarrow \phi \mid \forall
x:\tau. \phi \mid \exists x:\tau.\phi
\]
where $P$ ranges over basic predicates (as usual, we will often omit
the types of bound variables, when clear from the context). We assume
that predicates come equipped with an axiomatization.
For instance, the predicate $\mathsf{All}(l, \lambda x.\phi)$ is
defined to capture lists whose elements 
satisfies $\phi$.  This can be defined axiomatically:
\begin{mathpar}
{\rm All}([], \lambda x.\phi) \and
\forall h t. {\rm All}(t, \lambda x.\phi) \Rightarrow \phi(h) \Rightarrow {\rm All}(\cons{h}{t}, \lambda x. \phi)
\end{mathpar}
We use the notation $\lambda x.\phi$ for simplicity, although we have
not introduced formally a type for propositions---adding such a type
is straightforward and orthogonal to our work: another alternative would be to use axiom scheme.

We define well-typed assertions using a judgment of the form $\Gamma
\vdash \phi$. The typing rules are standard.
%
%
A HOL judgment is then of the form $\Gamma \mid \Psi \vdash \phi$,
where $\Gamma$ is a simply typed context, $\Psi$ is a set of
assertions, and $\phi$ is an assertion, and such that $\Gamma\vdash \psi$
for every $\psi \in \Psi$, and $\Gamma\vdash\phi$. The rules of the
logic are given in Figure~\ref{fig:hol}, where the notation
$\phi\subst{x}{t}$ denotes the (capture-free) substitution of $x$ by
$t$ in $\phi$. In addition to the usual rules for equality,
implication and universal quantification, there are rules for
inductive types (only the rules for lists are displayed; similar rules
exist for booleans and natural numbers): the rule \rname{LIST} models
the induction principle for lists; the rules \rname{NC} and
\rname{CONS} formalise injectivity and non overlap of constructors.
A rule for strong induction \rname{SLIST} can be considered as well,
and is in fact derivable from simple induction.
\begin{figure}[h]
\centering
\begin{mdframed}
\begin{mathpar}
\infer[\sf AX]
       {\jhol{\Gamma}{\Psi}{\phi}}
       {\phi \in \Psi}
\and       
\infer[\sf CONV]
      {\jhol{\Gamma}{\Psi}{t=t'}}
      {\Gamma \vdash t:\tau & \Gamma \vdash t':\tau & t=_{\beta\iota\mu} t'}
\and
\infer[\sf SUBST]
       {\jhol{\Gamma}{\Psi}{\phi\subst{x}{u}}}
       {\jhol{\Gamma}{\Psi}{\phi\subst{x}{t}} & \jhol{\Gamma}{\Psi}{t = u}}
\and
\infer[\sf \Rightarrow_I]
        {\jhol{\Gamma}{\Psi}{\psi \Rightarrow \phi}}
        {\jhol{\Gamma}{\Psi,\psi}{\phi}}
\and
\infer[\sf \Rightarrow_E]
        {\jhol{\Gamma}{\Psi}{\phi}}
        {\jhol{\Gamma}{\Psi}{\psi\Rightarrow\phi} & \jhol{\Gamma}{\Psi}{\psi}}
\and
\infer[\sf \forall_I]
        {\jhol{\Gamma}{\Psi}{\forall x:\sigma. \phi}}
        {\jhol{\Gamma, x:\sigma}{\Psi}{\phi}}
\and
\infer[\sf \forall_E]
        {\jhol{\Gamma}{\Psi}{\phi\subst{x}{t}}}
        {\jhol{\Gamma}{\Psi}{\forall x:\sigma. \phi} & \Gamma \vdash t : \sigma}
\and
\infer[\sf \top_I]
        {\jhol{\Gamma}{\Psi}{\top}}
        {}
\and
\infer[\sf \bot_E]
        {\jhol{\Gamma}{\Psi}{\phi}}
        {\jhol{\Gamma}{\Psi}{\bot} & \Gamma \vdash \phi}
\and
\infer[\sf LIST]
       {\jhol{\Gamma}{\Psi}{\forall t:\listt{\sigma}. \phi}}
       {\jhol{\Gamma}{\Psi}{\phi\subst{l}{[]}} & \jhol{\Gamma,h:\tau,t:\listt{\tau}}{\Psi,\phi}{
           \phi\subst{t}{\cons{h}{t}}}}
\and
\infer[\sf NC]
        {\jhol{\Gamma}{\emptyset}{[]\!\neq\!\cons{h}{t}}}{
        \Gamma \vdash \cons{h}{t}:\listt{\tau}}
\and
\infer[\sf CONS_i]
        {\jhol{\Gamma}{\Psi}{t_i=t'_i}}
        {\jhol{\Gamma}{\Psi}{\cons{t_1}{t_2} \!=\!\cons{t'_1}{t'_2}}}
\and
\infer[\sf SLIST]
  {\jhol{\Gamma}{\Psi}{\forall t : \listt{\tau}. \phi}}
  {\jhol{\Gamma, t:\listt{\tau}}{\Psi, \forall u : \listt{\tau}. |u| < |t| \Rightarrow \phi\subst{t}{u}}{\phi}}
\end{mathpar}
\end{mdframed}
\caption{Selected rules for HOL}\label{fig:hol}
\end{figure}

Higher-Order Logic inherits a set-theoretical interpretation from its
underlying simply-typed $\lambda$-calculus. We assume given for each
predicate $P$ an interpretation $\sem{P}$ which is compatible with the
type of $P$ and its axioms. The interpretation of assertions is then
defined in the usual way. Specifically, the interpretation
$\semexpv{\phi}{\rho}$ of an assertion $\phi$ w.r.t.\, a valuation
$\rho$ includes the clauses:
\begin{mathpar}
\semexpv{P(t_1,\dots,t_n)}{\rho} \defeq (\semv{t_1}{\rho},\dots,\semv{t_n}{\rho}) \in \sem{P} \and
\semexpv{\top}{\rho} \defeq \tilde\top \and
\semexpv{\bot}{\rho} \defeq \tilde\bot \and
\semexpv{\phi_1 \wedge \phi_2}{\rho}  \defeq \semexpv{\phi_1}{\rho} \:\tilde\wedge\: \semexpv{\phi_2}{\rho} \and
\semexpv{\phi_1 \Rightarrow \phi_2}{\rho} \defeq \semexpv{\phi_1}{\rho} \:\tilde{\Rightarrow}\: \semexpv{\phi_2}{\rho} \and
\semexpv{\forall x:\tau. \phi}{\rho} \defeq \tilde{\forall} v. v \in \sem{\tau} \:\tilde{\Rightarrow}\: \semexpv{\phi}{\rho\subst{x}{v}}
\end{mathpar}
Higher-order logic is sound with respect to this semantics.
\begin{theorem}[Soundness of set-theoretical semantics]
If $\jhol{\Gamma}{\Psi}{\phi}$, then for every valuation
$\rho\models\Gamma$, $\bigwedge_{\psi\in\Psi} \semexpv{\psi}{\rho}$
implies $\semexpv{\phi}{\rho}$.
\end{theorem}
In particular, higher-order logic is consistent, i.e.\ there is no
derivation of $\jhol{\Gamma}{\emptyset}{\bot}$ for any $\Gamma$.

\section{Unary Higher-Order Logic}
\label{sec:uhol}

As a stepping stone towards Relational Higher-Order Logic, we define
Unary Higher-Order Logic (UHOL). UHOL retains the flavor of refinement
types, but dissociates typing from assertions; judgments of UHOL are
of the form:
$$\juhol{\Gamma}{\Psi}{t}{\tau}{\phi}$$
where a distinguished variable $\res$, which doesn't appear in $\Gamma$, may appear
(free) in $\phi$ as a synonym of $t$. A judgment is well-formed if $t$
has type $\tau$, $\Psi$ is a valid set of assertions in the context
$\Gamma$, and $\phi$ is a valid assertion in the context
$\Gamma,\res:\tau$. Figure~\ref{fig:uhol} presents selected typing
rules. The \rname{ABS} rule allows proving formulas that refer to
$\lambda$-abstractions, expressing that if the argument satisfies a
precondition $\phi'$, then the result satisfies a postcondition
$\phi$.  The \rname{APP} rule, dually, proves a condition $\phi$ on an
application $t\ u$ provided that the argument $u$ satisfies the
precondition $\phi'$ of the function $t$. The \rname{VAR} rule
introduces a variable from the context with a formula proven in
HOL. Rules for constants (e.g. \rname{NIL}) work in the same way.
Rule \rname{CONS} proves a formula $\phi$ for a non-empty list,
provided that $\phi$ is a consequence of some conditions
$\phi',\phi''$ on its head and its tail.  Rule \rname{PAIR} allows the
construction of judgments about pairs in a similar manner. The rules
\rname{PROJ$_i$} prove judgments about the projections of a pair.  The
rule \rname{SUB} (subsumption) allows strengthening the assumed
assertions $\Psi$ and weakening the concluding assertion $\phi$. It
generates a HOL proof obligation. The rule \rname{CASE} can be used
for a case analysis over the constructor of a term. Finally, the rule
\rname{LETREC} supports inductive reasoning about recursive
definitions.  Recall that the domain of a recursive definition is an
inductive type, for which a natural notion of size exists. If,
assuming that a proposition holds for all elements smaller than the
argument, we can prove that the proposition holds for the body, then
the proposition must hold as well for the function. Furthermore, we
require that the function we are verifying satisfies the predicate
$\wdef{f}{x}{i}$, as was the case in HOL. The induction is performed
over the $<$ order, which varies depending on the type of the
argument.

\begin{figure*}
\begin{mdframed}
\small
\begin{mathpar}
\infer[\sf VAR]{\juhol{\Gamma}{\Psi}{x}{\sigma}{\phi}}
      {\jlc{\Gamma}{x}{\sigma} \\ \jhol{\Gamma}{\Psi}{\phi\subst{\res}{x}}}
\and
\infer[\sf ABS]
      {\juhol{\Gamma}{\Psi}{\lambda x:\tau. t}{\tau \to \sigma}{\forall x. \phi'
      \Rightarrow \phi\subst{\res}{\res\ x}}}{\juhol{\Gamma,x:\tau}{\Psi,\phi'}{t}{\sigma}{\phi}}
\and
\infer[\sf APP]{\juhol{\Gamma}{\Psi}{t\ u}{\sigma}{\phi\subst{x}{u}}}
      {\juhol{\Gamma}{\Psi}{t}{\tau\to \sigma}{\forall x.
          \phi'\subst{\res}{x}\Rightarrow\phi\subst{\res}{\res\ x}} \qquad
\juhol{\Gamma}{\Psi}{u}{\tau}{\phi'}}
\and
\infer[\sf NIL]
       {\juhol{\Gamma}{\Psi}{[]}{\listt{\sigma}}{\phi}}
       {
        \jholn{\Gamma}{\Psi}{\phi\subst{\res}{[]}}
        }
\and
\infer[\sf CONS]
       {\juhol{\Gamma}{\Psi}{\cons{h}{t}}{\listt{\sigma}}{\phi}}
       {\begin{array}{c}
        \juhol{\Gamma}{\Psi}{h}{\sigma}{\phi'} \hspace{2cm}
        \juhol{\Gamma}{\Psi}{t}{\listt{\sigma}}{\phi''} \\
        \jholn{\Gamma}{\Psi}{\forall x y. \phi'\subst{\res}{x}
          \Rightarrow \phi''\subst{\res}{y}
             \Rightarrow \phi\subst{\res}{\cons{x}{y}}}
       \end{array}}
\and
\infer[\sf PROJ_i]
       {\juhol{\Gamma}{\Psi}{\pi_i(t)}{\sigma}{\phi}}
       {\juhol{\Gamma}{\Psi}{t}{\sigma\times\tau}
       {\phi\subst{\res}{\pi_i(\res)}}}
\and
\infer[\sf PAIR]
       {\juhol{\Gamma}{\Psi}{\pair{t}{u}}{\sigma\times\tau}{\phi}}
       {\begin{array}{c}
        \juhol{\Gamma}{\Psi}{t}{\sigma}{\phi'} \hspace{2cm}
        \juhol{\Gamma}{\Psi}{u}{\tau}{\phi''} \\
        \jholn{\Gamma}{\Psi}{\forall x y. \phi'\subst{\res}{x}
          \Rightarrow \phi''\subst{\res}{y}
          \Rightarrow \phi\subst{\res}{\pair{x}{y}}}
        \end{array}}
\and
\infer[\sf SUB]{\juhol{\Gamma}{\Psi}{t}{\sigma}{\phi}}
      {\juhol{\Gamma}{\Psi}{t}{\sigma}{\phi'} &
       \jholn{\Gamma}{\Psi}{\phi'\subst{\res}{t} \Rightarrow \phi\subst{\res}{t}}}
\and
\infer[\sf LISTCASE]
  {\juhol{\Gamma}{\Psi}{\caselist{l}{u}{v}}{\sigma}{\phi}}
  {\begin{array}{c}
      \jlc{\Gamma}{l}{\listt{\tau}} \\
      \juhol{\Gamma}{\Psi, l = \nil}{u}{\sigma}{\phi} \\
      \juhol{\Gamma}{\Psi}{v}{\tau\to\listt{\tau}\to\sigma}
            {\forall h t. l = \cons{h}{t} \Rightarrow \phi\subst{\res}{\res\ h\ t}}
   \end{array}}
\and
\infer[\sf LETREC]
  {\juhol{\Gamma}{\Psi}{{\rm letrec}\ f\ x\ = e}{I\to{\sigma}}{\forall x. \phi' \Rightarrow \phi\subst{\res}{\res\ x}}}
  {\begin{array}{c}
      \wdef{f}{x}{e} \\
      \juhol{\Gamma, x:I, f:I\to{\sigma}}
         {\Psi, \phi', \forall m. |m| < |x| \Rightarrow \phi'\subst{x}{m}\Rightarrow \phi\subst{x}{m}\subst{\res}{f\ m}}{e}{\sigma}{\phi}
   \end{array}}
\and  
\text{ where }I \in \{\nat,\listt{\tau}\}
\end{mathpar}
\end{mdframed}                
\caption{Unary Higher-Order Logic rules}\label{fig:uhol}
\end{figure*}

We now discuss the main meta-theoretic results of UHOL. The following
result establishes that every HOL judgment can be proven in UHOL and
viceversa.

\begin{thm}[Equivalence with HOL] \label{thm:equiv-uhol-hol} For every context 
  $\Gamma$, simple type $\sigma$, term $t$, set of assertions $\Psi$
  and assertion $\phi$, the following are equivalent:
  \begin{itemize}
  \item
    $\juhol{\Gamma}{\Psi}{t}{\sigma}{\phi}$
  \item
    $\jhol{\Gamma}{\Psi}{\phi\subst{\res}{t}}$
  \end{itemize}
\end{thm}
The forward implication follows by induction on the derivation of
$\juhol{\Gamma}{\Psi}{t}{\sigma}{\phi}$.  The reverse implication is
immediate from the rule \rname{SUB} and the observation that
$\juhol{\Gamma}{\Psi}{t}{\sigma}{\top}$ whenever $t$ is a term of type
$\sigma$.

We lift the HOL semantics to UHOL. Terms, types and formulas are
interpreted as before. Additionally, for every valuation $\rho$ let
$\rho\subst{x}{v}$ denote its unique extension $\rho'$ such that
$\rho'(y)=v$ if $x=y$ and $\rho'(y)=\rho(y)$ otherwise. The following
corollary states the soundness of UHOL.

\begin{cor}[Set-theoretical soundness and consistency]\label{cor:uhol:sound}
If $\juhol{\Gamma}{\Psi}{t}{\sigma}{\phi}$, then for every valuation
$\rho\models\Gamma$, $\bigwedge_{\psi\in\Psi} \semexpv{\psi}{\rho}$
implies $\semexpv{\phi}{\rho\subst{\res}{\semexpv{t}{\rho}}}$. In
particular, there is no proof of
$\juhol{\Gamma}{\emptyset}{t}{\sigma}{\bot}$ in UHOL.
\end{cor}

Next, we prove subject conversion for UHOL. The result follows
immediately from Theorem~\ref{thm:equiv-uhol-hol} and subject
conversion of HOL, which is itself a direct consequence of the
\rname{CONV} and \rname{SUBST} rules.

\begin{cor}[Subject conversion]
Assume that $t =_{\beta\iota\mu} t'$ and
$\juhol{\Gamma}{\Psi}{t}{\sigma}{\phi}$.  If $\Gamma\vdash t':\sigma$
then $\juhol{\Gamma}{\Psi}{t'}{\sigma}{\phi}$.
\end{cor}

\section{Relational Higher-Order Logic}
\label{sec:rhol}
Relational Higher-Order Logic (RHOL) extends UHOL's separation of
assertions from types to a relational setting. Formally, RHOL is a
relational type system which manipulates judgments of the form
$$\jrhol{\Gamma}{\Psi}{t_1}{\tau_1}{t_2}{\tau_2}{\phi}$$
which combine a typing judgment for a pair of PCF terms and permit
reasoning about the relation between them. We therefore require that
$t_1$ and $t_2$ respectively have types $\tau_1$ and $\tau_2$ in
$\Gamma$. Well-formedness of the judgment requires $\Psi$ to be a
valid set of assertions in $\Gamma$ and $\phi$ to be a valid assertion
in $\Gamma, \res\ltag:\tau_1, \res\rtag:\tau_2$, where the special
variables $\res\ltag$ and $\res\rtag$ are used as synomyms for $t_1$
and $t_2$ in $\phi$. The informal meaning of the judgment is the
expected one: If the variables in $\Gamma$ are related by the
assertions in $\Psi$, then the terms $t_1$ and $t_2$ are related by
the assertion $\phi$.

\subsection{Proof rules}
\label{ssec:rhol-pr}

The type system combines two-sided rules (Figure \ref{fig:twosided}),
which apply when the two terms have the same top-level constructors
and one-sided rules (Figure \ref{fig:onesided}), which analyze either
one of the two terms.
For instance, the \rname{APP} rule applies when the two terms are
applications, and requires that the functions $t\ltag$ and $t\rtag$
relate and the arguments $u_1$ and $u_2$ relate. Specifically,
$t\ltag$ and $t\rtag$ must map values related by $\phi'$ to values
related by $\phi$, and $u_1$ and $u_2$ must be related by $\phi'$. The
\rname{ABS} rule is dual.  The \rname{PAIR} rule requires that the
left and right components of a pair relate independently (a stronger
rule is discussed at the end of the section). The \rname{PROJ$_i$}
rules require in their premise an assertion that only refers to the
the first or the second component of the pair. The rules for lists
require that the two lists are either both empty, or both
non-empty. The rule \rname{CONS} requires that the two heads and the
two tails relate independently. The \rname{CASE} rule derives
judgements about two case constructs when the terms over which the
matching happens reduce to the same branch (i.e. have the same
constructor) on both sides.

In contrast, one-sided typing rules only analyze one term; therefore,
they come in two flavours: left rules (shown in
Figure~\ref{fig:onesided}) and right rules (ommitted but
similar). Rule \rname{ABS{-}L} considers the case where the left term
is a $\lambda$-abstraction, and requires the body of the abstraction
to be related to the right term $u_2$ whenever the argument (on the
left side) satisfies a non-relational assertion $\phi'$. Dually, rule
\rname{APP{-}L} considers the case where the left term is of the form
$t_1\ u_1$, and $t_1$ is related to the right term $u_2$;
specifically, $t_1$ should map every value satisfying $\phi'$ to a
value satisfying $\phi$. Moreover, $u_1$ should satisfy $\phi'$. Since
$\phi'$ is a non-relational assertion, we demand that it can be
established using UHOL, not RHOL. One-sided rules for pairs and lists
follow a similar pattern.

In addition, RHOL has structural rules
(Figure~\ref{fig:structrules}). The rule \rname{SUB} can be used for
strengthening the assumed assertions and for weakening the concluding
assertion; the ensuing side-conditions are discharged in HOL. Other
structural rules assimilate rules of HOL. For instance, if we can
prove two different assertions for the same terms we can prove the
conjunction of the assertions (\rname{$\wedge_I$}). Other logical
connectives have similar rules. Finally, the rule \rname{UHOL-L} (and
a dual rule \rname{UHOL-R}) allow falling back to UHOL in a RHOL
proof.

Rules \rname{LETREC} and \rname{LETREC-L} introduce recursive function
definitions (Figure~\ref{fig:recursionrules}). These rules allow for a
style of reasoning very similar to strong induction. If, assuming that
the function's specification holds for all smaller arguments, we can
prove that the functions specification holds, then the specification
must hold for all arguments. We require that the two functions we are
relating satisfy the predicates $\wdef{f_i}{x_i}{e_i}$, as was the
case in HOL and UHOL. The induction is performed over the order $(a,b)
< (c,d)$, which holds whenever both $a\leq b$ and $c\leq d$, and at
least one of the inequalities is strict.

 
\begin{figure*}
\begin{mdframed}
\small
\begin{mathpar}
\infer[\sf ABS]
      {\jrhol{\Gamma}{\Psi}{\lambda x_1:\tau_1. t_1}{\tau_1 \to \sigma_1}{\lambda x_2:\tau_2. t_2}{\tau_2\to \sigma_2}{\forall x_1,x_2. \phi' \Rightarrow \phi\subst{\res\ltag}{\res\ltag\ x_1}\subst{\res\rtag}{\res\rtag\ x_2}}}
      {\jrhol{\Gamma,x_1:\tau_1,x_2:\tau_2}{\Psi,\phi'}{t_1}{\sigma_1}{t_2}{\sigma_2}{\phi}}
\and
\infer[\sf APP]{\jrhol{\Gamma}{\Psi}{t_1 u_1}{\sigma_1}{t_2 u_2}{\sigma_2}{\phi\subst{x_1}{u_1}\subst{x_2}{u_2}}}
      {\begin{array}{c}
\jrhol{\Gamma}{\Psi}{t_1}{\tau_1\to \sigma_1}{t_2}{\tau_2\to \sigma_2}{
          \forall x_1,x_2. \phi'\subst{\res\ltag}{x_1}\subst{\res\rtag}{x_2}\Rightarrow \phi\subst{\res\ltag}{\res\ltag\ x_1}\subst{\res\rtag}{\res\rtag\ x_2}}\\
\jrhol{\Gamma}{\Psi}{u_1}{\tau_1}{u_2}{\tau_2}{
          \phi'}
\end{array}}
\and
\infer[\sf VAR]{\jrhol{\Gamma}{\Psi}{x_1}{\sigma_1}{x_2}{\sigma_2}{\phi}}
      {\jlc{\Gamma}{x_1}{\sigma_1} & \jlc{\Gamma}{x_2}{\sigma_2} & \jhol{\Gamma}{\Psi}{\phi\subst{\res\ltag}{x_1}\subst{\res\rtag}{x_2}}}
\and
\infer[\sf TRUE]
       {\jrhol{\Gamma}{\Psi}{\tbool}{\bool}{\tbool}{\bool}{\phi}}
       {
        \jholn{\Gamma}{\Psi}{\phi\defsubst{\tbool}{\tbool}}
        }
\and
\infer[\sf NIL]
       {\jrhol{\Gamma}{\Psi}{[]}{\listt{\sigma_1}}{[]}{\listt{\sigma_2}}{\phi}}
       {
        \jholn{\Gamma}{\Psi}{\phi\defsubst{[]}{[]}}
        }
\and
\infer[\sf CONS]
       {\jrhol{\Gamma}{\Psi}{\cons{h_1}{t_1}}{\listt{\sigma_1}}{\cons{h_2}{t_2}}{\listt{\sigma_2}}{\phi}}
       {\begin{array}{c}
        \jrhol{\Gamma}{\Psi}{h_1}{\sigma_1}{h_2}{\sigma_2}{\phi'} \hspace{2cm}
        \jrhol{\Gamma}{\Psi}{t_1}{\listt{\sigma_1}}{t_2}{\listt{\sigma_2}}{\phi''} \\
        \jholn{\Gamma}{\Psi}{\forall x_1 x_2 y_1 y_2. \phi'\defsubst{x_1}{x_2} \Rightarrow \phi''\defsubst{y_1}{y_2}
             \Rightarrow \phi\defsubst{\cons{x_1}{y_1}}{\cons{x_2}{y_2}}}
       \end{array}}
\and
\infer[\sf LISTCASE]
  {\jrhol{\Gamma}{\Psi}{\caselist{l_1}{u_1}{v_1}}{\sigma_1}{\caselist{l_2}{u_2}{v_2}}{\sigma_2}{\phi}}
  {\begin{array}{c}
      \jrhol{\Gamma}{\Psi}{l_1}{\listt{\tau_1}}{l_2}{\listt{\tau_2}}{\res\ltag = \nil \Leftrightarrow \res\rtag = \nil} \\
      \jrhol{\Gamma}{\Psi, l_1= \nil, l_2= \nil}{u_1}{\sigma_1}{u_2}{\sigma_2}{\phi} \\
      \jrhol{\Gamma}{\Psi}{v_1}{\tau_1\to\listt{\tau_1}\to\sigma_1}{v_2}{\tau_2\to\listt{\tau_2}\to\sigma_2}\\
            {\forall h_1 h_2 t_1 t_2. l_1 = \cons{h_1}{t_1} \Rightarrow l_2 = \cons{h_2}{t_2} \Rightarrow \phi\defsubst{\res\ltag\ h_1\ t_1}{\res\rtag\ h_2\ t_2}}
   \end{array}}
\and
\infer[\sf PAIR]
       {\jrhol{\Gamma}{\Psi}{\pair{t_1}{u_1}}{\sigma_1\times\tau_1}{\pair{t_2}{u_2}}{\sigma_2\times\tau_2}{\phi}}
       {\begin{array}{c}
        \jrhol{\Gamma}{\Psi}{t_1}{\sigma_1}{t_2}{\sigma_2}{\phi'} \hspace{2cm}
        \jrhol{\Gamma}{\Psi}{u_1}{\tau_1}{u_2}{\tau_2}{\phi''} \\
        \jholn{\Gamma}{\Psi}{\forall x_1 x_2 y_1 y_2. \phi'\defsubst{x_1}{x_2} \Rightarrow \phi''\defsubst{y_1}{y_2}
             \Rightarrow \phi\defsubst{\pair{x_1}{y_1}}{\pair{x_2}{y_2}}}
        \end{array}}
\and
\infer[\sf PROJ_i]
       {\jrhol{\Gamma}{\Psi}{\pi_i(t_1)}{\sigma_1}{\pi_i(t_2)}{\sigma_2}{\phi}}
       {\jrhol{\Gamma}{\Psi}{t_1}{\sigma_1\times\tau_1}{t_2}{\sigma_2\times\tau_2}
       {\phi\defsubst{\pi_i(\res\ltag)}{\pi_i(\res\rtag)}}}
\end{mathpar}
\end{mdframed}
\caption{Two-sided rules}\label{fig:twosided}
\end{figure*}

\begin{figure*}
\begin{mdframed}
\small
\begin{mathpar}
\infer[\sf SUB]{\jrhol{\Gamma}{\Psi}{t_1}{\sigma_1}{t_2}{\sigma_2}{\phi}}
      {\jrhol{\Gamma}{\Psi}{t_1}{\sigma_1}{t_2}{\sigma_2}{\phi'} & 
       \Gamma \mid \Psi \vdash_{\sf HOL} \phi'\defsubst{t_1}{t_2} \Rightarrow \phi\defsubst{t_1}{t_2}}
\and
\infer[\sf \wedge_I]
      {\jrhol{\Gamma}{\Psi}{t_1}{\sigma_2}{t_2}{\sigma_2}{\phi \wedge \phi'}}
      {\jrhol{\Gamma}{\Psi}{t_1}{\sigma_2}{t_2}{\sigma_2}{\phi}
       & \jrhol{\Gamma}{\Psi}{t_1}{\sigma_2}{t_2}{\sigma_2}{\phi'}}
\and
\infer[\sf \Rightarrow_I]
      {\jrhol{\Gamma}{\Psi}{t_1}{\sigma_2}{t_2}{\sigma_2}{\phi' \Rightarrow \phi}}
      {\jrhol{\Gamma}{\Psi,\phi'\defsubst{t_1}{t_2}}{t_1}{\sigma_2}{t_2}{\sigma_2}{\phi}}
\and
\infer[\sf UHOL-L]
      {\jrhol{\Gamma}{\Psi}{t_1}{\sigma_1}{t_2}{\sigma_1}{\phi}}
      {\juhol{\Gamma}{\Psi}{t_1}{\sigma_1}{\phi\defsubst{\res}{t_2}}}
\end{mathpar}
\end{mdframed}                
\caption{Structural rules}\label{fig:structrules}
\end{figure*}

\begin{figure*}
\begin{mdframed}
\small
\begin{mathpar}
\infer[\sf ABS{-}L]
      {\jrhol{\Gamma}{\Psi}{\lambda x_1:\tau_1. t_1}{\tau_1 \to \sigma_1}{t_2}{\sigma_2}{\forall x_1. \phi' \Rightarrow \phi \subst{\res\ltag}{\res\ltag\ x_1}}}
      {\jrhol{\Gamma,x_1:\tau_1}{\Psi, \phi'}{t_1}{\sigma_1}{t_2}{\sigma_2}{\phi}}
\and
\infer[\sf APP{-}L]
      {\jrhol{\Gamma}{\Psi}{t_1 u_1}{\sigma_1}{u_2}{\sigma_2}{\phi}}
      {\begin{array}{c}
          \jrhol{\Gamma}{\Psi}{t_1}{\tau_1\to \sigma_1}{u_2}{\sigma_2}{\forall x_1. \phi'\subst{\res\ltag}{x_1} \Rightarrow \phi\subst{\res\ltag}{\res\ltag\ x_1}}\\
          \juhol{\Gamma}{\Psi}{u_1}{\sigma_1}{\phi'\subst{x_1}{u_1}}
\end{array}}
\and
\infer[\sf VAR{-}L]
      {\jrhol{\Gamma}{\Psi}{x_1}{\sigma_1}{t_2}{\sigma_2}{\phi}}
      {\phi\subst{\res\ltag}{x_1} \in \Psi &  \res\rtag\not\in\ FV(\phi) &
       \jlc{\Gamma}{t_2}{\sigma_2}}
\and
\infer[\sf NIL{-}L]
       {\jrhol{\Gamma}{\Psi}{[]}{\listt{\sigma_1}}{t_2}{\sigma_2}{\phi}}
       {\jhol{\Gamma}{\Psi}{\phi\defsubst{[]}{t_2}} &
        \jlc{\Gamma}{t_2}{\sigma_2}}
\and
\infer[\sf CONS{-}L]
       {\jrhol{\Gamma}{\Psi}{\cons{h_1}{t_1}}{\listt{\sigma_1}}{t_2}{\sigma_2}{\phi}}
       {\begin{array}{c}
        \jrhol{\Gamma}{\Psi}{h_1}{\sigma_1}{t_2}{\sigma_2}{\phi'}\hspace{2cm}
        \jrhol{\Gamma}{\Psi}{t_1}{\listt{\sigma_1}}{t_2}{\sigma_2}{\phi''} \\
        \jholn{\Gamma}{\Psi}{\forall x_1 x_2 y_1. \phi'\defsubst{x_1}{x_2} \Rightarrow \phi''\defsubst{y_1}{x_2}
             \Rightarrow \phi\defsubst{\cons{x_1}{y_1}}{x_2}}
        \end{array}}
\and
\infer[\sf LISTCASE-L]
  {\jrhol{\Gamma}{\Psi}{\caselist{t_1}{u_1}{v_1}}{\sigma_1}{t_2}{\sigma_2}{\phi}}
  {\begin{array}{c}
      \jlc{\Gamma}{t_1}{\listt{\tau}} \\ 
      \jrhol{\Gamma}{\Psi, t_1=\nil}{u_1}{\sigma_1}{t_2}{\sigma_2}{\phi} \\
      \jrhol{\Gamma}{\Psi}{v_1}{\tau\to\listt{\tau}\to\sigma_1}{t_2}{\sigma_2}{\forall h_1 l_1. t_1 = \cons{h_1}{l_1} \Rightarrow \phi\subst{\res\ltag}{\res\ltag\ h_1\ l_1}}
   \end{array}}
\and        
\infer[\sf PAIR{-}L]
       {\jrhol{\Gamma}{\Psi}{\pair{t_1}{u_1}}{\sigma_1\times\tau_1}{t_2}{\sigma_2}{\phi}}
       {\begin{array}{c}
        \jrhol{\Gamma}{\Psi}{t_1}{\sigma_1}{t_2}{\sigma_2}{\phi'}\hspace{2cm}
        \jrhol{\Gamma}{\Psi}{u_1}{\tau_1}{t_2}{\sigma_2}{\phi''} \\
        \jholn{\Gamma}{\Psi}{\forall x_1 x_2 y_1. \phi'\defsubst{x_1}{x_2} \Rightarrow \phi''\defsubst{y_1}{x_2}
             \Rightarrow \phi\defsubst{\pair{x_1}{y_1}}{x_2}}
        \end{array}}
\and
\infer[\sf PROJ_1{-}L]
       {\jrhol{\Gamma}{\Psi}{\pi_1(t_1)}{\sigma_1}{t_2}{\sigma_2}{\phi}}
       {\jrhol{\Gamma}{\Psi}{t_1}{\sigma_1\times\tau_1}{t_2}{\sigma_2}
          {\phi\subst{\res\ltag}{\pi_1(\res\ltag)}}}
\end{mathpar}
\end{mdframed}
\caption{One-sided rules}\label{fig:onesided}
\end{figure*}

\begin{figure*}
\begin{mdframed}
\small
\begin{mathpar}
\infer[\sf LETREC]
  {\jrhol{\Gamma}{\Psi}{{\rm letrec}\ f_1\ x_1\ = e_1}{I_1\to{\sigma_2}}{{\rm letrec}\ f_2\ x_2\ = e_2}{I_2\to{\sigma_2}}{\forall x_1 x_2. \phi' \Rightarrow \phi\defsubst{\res\ltag\ x_1}{\res\rtag\ x_2}}}
  {\begin{array}{c}
      \wdef{f_1}{x_1}{e_1} \:\: \wdef{f_2}{x_2}{e_2} \\
      \jrhol{\Gamma, x_1:I_1, x_2:I_2, f_1:I_1\to{\sigma}, f_2:I_2\to{\sigma_2}}
         {\\ \Psi, \phi', \forall m_1 m_2. (|m_1|, |m_2|) < (|x_1|, |x_2|) \Rightarrow \phi'\subst{x_1}{m_1}\subst{x_2}{m_2}\Rightarrow
                  \phi\subst{x_1}{m_1}\subst{x_2}{m_2}\defsubst{f_1\ m_1}{f_2\ m_2}}{\\ e_1}{\sigma_1}{e_2}{\sigma_2}{\phi}
   \end{array}}
\and
\infer[\sf LETREC-L]
  {\jrhol{\Gamma}{\Psi}{{\rm letrec}\ f_1\ x_1\ = e_1}{I_1\to{\sigma_2}}{t_2}{\sigma_2}{\forall x_1 . \phi' \Rightarrow \phi\subst{\res\ltag}{\res\ltag\ x_1}}}
  {\begin{array}{c}
      \wdef{f_1}{x_1}{e_1} \\
      \jrhol{\Gamma, x_1:I_1, f_1:I_1\to{\sigma}}
         {\Psi, \phi', \forall m_1. |m_1| < |x_1| \Rightarrow \phi'\subst{x_1}{m_1} \Rightarrow \phi\subst{x_1}{m_1}\defsubst{f_1\ m_1}{t_2}}{\\ e_1}{\sigma_1}{t_2}{\sigma_2}{\phi}
   \end{array}}
\and
\text{ where }I_1,I_2 \in \{\nat,\listt{\tau}\}
\end{mathpar}
\end{mdframed}
\caption{Recursion rules}
\label{fig:recursionrules}
\end{figure*}

\begin{figure*}
\begin{mdframed}
\small
\begin{mathpar}
\infer[\sf APP-FUN]{\jrhol{\Gamma}{\Psi}{t_1\ u_1}{\sigma_1}{t_2\ 
    u_2}{\sigma_2}{\phi}} {\jrhol{\Gamma}{\Psi}{t_1}{\tau_1\to \sigma_1}{t_2}{\tau_2\to
      \sigma_2}{\phi\subst{\res\ltag}{\res\ltag\ u_1}\subst{\res\rtag}{\res\rtag\ u_2}}}
\and
\infer[\sf APP-ARG]{\jrhol{\Gamma}{\Psi}{t_1\ u_1}{\sigma_1}{t_2\ u_2}{\sigma_2}{\phi}}
      {\jrhol{\Gamma}{\Psi}{u_1}{\tau_1}{u_2}{\tau_2}{\phi\subst{\res\ltag}{t_1\ \res\ltag}
          \subst{\res\rtag}{t_2\ \res\rtag}}}
\and
\infer[\sf PAIR-FST]{\jrhol{\Gamma}{\Psi}{\pair{t_1}{u_1}}{\tau_1\times\sigma_1}{\pair{t_2}
      {u_2}}{\tau_2\times\sigma_2}{\phi}}
        {\jrhol{\Gamma}{\Psi}{t_1}{\tau_1}{t_2}{\tau_2}
          {\phi\defsubst{\pair{\res\ltag}{u_1}}{\pair{\res\rtag}{u_2}}}}
\and
\infer[\sf LLCASE-A]
  {\jrhol{\Gamma}{\Psi}{\caselist{t_1}{u_1}{v_1}}{\sigma_1}{\caselist{t_2}{u_2}{v_2}}{\sigma_2}{\phi}}
  {\begin{array}{c}
      \jrhol{\Gamma}{\Psi}{t_1}{\listt{\tau_1}}{t_2}{\listt{\tau_2}}{\top} \\
      \jrhol{\Gamma}{\Psi, t_1=\nil, t_2=\nil}{u_1}{\sigma_1}{u_2}{\sigma_2}{\phi} \\
      \jrhol{\Gamma}{\Psi, t_2 = \nil}{v_1}{\tau_1\to\listt{\tau_1}\to\sigma_1}{u_2}{\sigma_2}{\forall h_1 l_1. t_1 = \cons{h_1}{l_1} \Rightarrow \phi\subst{\res\ltag}{\res\ltag\ h_1\ l_1}} \\
      \jrhol{\Gamma}{\Psi,t_1 = \nil}{u_1}{\tau_1\to\listt{\tau_1}\to\sigma_1}{v_2}{\tau_2\to\listt{\tau_2}\to\sigma_2}
            {\\ \forall h_2. t_2 = \cons{h_2}{l_2} \Rightarrow \phi\subst{\res\rtag}{\res\rtag\ h_2\ l_2}} \\
      \jrhol{\Gamma}{\Psi}{v_1}{\tau_1\to\listt{\tau_1}\to\sigma_1}{v_2}{\tau_2\to\listt{\tau_2}\to\sigma_2}
            \\{\forall h_1 h_2 l_1 l_2. t_1 = \cons{h_1}{l_1} \Rightarrow t_2 = \cons{h_2}{l_2} \Rightarrow \phi\defsubst{\res\ltag\ h_1\ l_1}{\res\rtag\ h_2\ l_2}}
   \end{array}}

\end{mathpar}
\end{mdframed}
\caption{Some derived rules}\label{fig:derived}
\end{figure*}

\subsection{Discussion}
\label{ssec:rhol-dr}
Our choice of the rules is guided by the practical considerations of
being able to verify examples easily, without specifically aiming for
minimality or exhaustiveness. In fact, many of our rules can be
derived from others, or reduced to a more elementary form. For
instance:
\begin{itemize}
\item The structural rules to reason about logical connectives, such
  as \rname{$\wedge_I$}, can be derived by induction on the length of
  derivations with the help of \rname{SUB}.
\item The \rname{VAR-L} (similarly, \rname{NIL-L}) rule can be weakened, without affecting the
  strength of the system,
$$\infer[\sf VAR{-}L]
      {\jrhol{\Gamma}{\Psi}{x_1}{\sigma_1}{x_2}{\sigma_2}{\phi}}
      {\phi\subst{\res\ltag}{x_1} \in \Psi &  \res\rtag\not\in\ FV(\phi)}$$

\item The premise of the \rname{VAR} rule (and similarly, \rname{NIL})
  can be changed to $\phi\subst{\res}{x} \in \Psi$.  We can recover
  the original rule by one application of \rname{SUB}.
  
\item The rules \rname{APP-FUN} and \rname{APP-ARG} in
  Figure~\ref{fig:derived} (adapted from \citet{GNFS16}) can be
  derived from the rule \rname{APP}. To derive \rname{APP-FUN},
  instantiate $\phi'$ to $\res\ltag = u_1 \wedge \res\rtag = u_2$ in
  \rname{APP}. To derive \rname{APP-ARG}, we have to prove a trivial
  condition $\forall x_1 x_2. \phi\defsubst{t_1\ x_1}{t_2\ x_2}
  \Rightarrow \phi\defsubst{t_1\ x_1}{t_2\ x_2}$ on $t_1,t_2$.
  
\item The \rname{PAIR-FST} and \rname{PAIR-SND} rules in
  Figure~\ref{fig:derived} can be derived in a similar way. These
  rules overcome a limitation of the original \rname{PAIR} rule,
  namely, that the relations for the two components of the pair must
  be independent. \rname{PAIR-FST} and \rname{PAIR-SND} allow
  relating, for instance, pairs of integers $\langle m_1, n_1\rangle$
  and $\langle m_2, n_2\rangle$ such that $m_1+n_1=m_2+n_2$.
  
\item The \rname{LLCASE-A} rule can be used to reason about case
  constructs when the terms over which we discriminate do not
  necessarily reduce to the same branch. It is equivalent to applying
  \rname{CASE-L} followed by \rname{CASE-R}.
\end{itemize}

\subsection{Meta-theory}
\label{ssec:rhol-mt}
RHOL retains the expressiveness of HOL, as formalized in the following
theorem.
\begin{thm}[Equivalence with HOL] \label{thm:equivhol} For every context 
  $\Gamma$, simple types $\sigma_1$ and $\sigma_2$, terms $t_1$ and
  $t_2$, set of assertions $\Psi$ and assertion $\phi$, if
  $\jlc{\Gamma}{t_1}{\sigma_1}$ and $\jlc{\Gamma}{t_2}{\sigma_2}$, then
  the following are equivalent:
  \begin{itemize}
  \item
    $\jrhol{\Gamma}{\Psi}{t_1}{\sigma_1}{t_2}{\sigma_2}{\phi}$
  \item
    $\jhol{\Gamma}{\Psi}{\phi\subst{\res\ltag}{t_1}\subst{\res\rtag}{t_2}}$
  \end{itemize}
\end{thm}
The proof of the forward implication proceeds by induction on the
structure of derivations. The proof of the reverse implication is
immediate from the rule \rname{SUB} and the observation that
$\jrhol{\Gamma}{\emptyset}{t_1}{\sigma_1}{t_2}{\sigma_2}{\top}$
whenever $t_1$ and $t_2$ are typable terms of types $\sigma_1$ and
$\sigma_2$ respectively.

This immediately entails soundness of RHOL, which is expressed in the
following result:
\begin{cor}[Set-theoretical soundness and consistency]\label{cor:rhol:sound}
If $\jrhol{\Gamma}{\Psi}{t_1}{\sigma_1}{t_2}{\sigma_2}{\phi}$, then
for every valuation $\rho\models\Gamma$, $\bigwedge_{\psi\in\Psi}
\semexpv{\psi}{\rho}$ implies
$\semexpv{\phi}{\rho\subst{\res\ltag}{\semexpv{t_1}{\rho}},
  \subst{\res\rtag}{\semexpv{t_2}{\rho}}}$. In particular, there is no
proof of $\jrhol{\Gamma}{\emptyset}{t_1}{\sigma_1}{
  t_2}{\sigma_2}{\bot}$ for any $\Gamma$.
\end{cor}
The equivalence also entails subject conversion (and as special cases
subject reduction and subject expansion). This follows immediately
from subject conversion of HOL (which, as stated earlier, is itself a
direct consequence of the \rname{CONV} and \rname{SUBST} rules).

\begin{cor}[Subject conversion]
Assume that $t_1 =_{\beta\iota\mu} t_1'$ and $t_2 =_{\beta\iota\mu}
t_2'$ and $\jrhol{\Gamma}{\Psi}{t_1}{\sigma_1}{t_2}{\sigma_2}{\phi}$.
If $\Gamma\vdash t'_1:\sigma_1$ and $\Gamma\vdash t'_2:\sigma_2$ then
$\jrhol{\Gamma}{\Psi}{t_1'}{\sigma_1}{t'_2}{\sigma_2}{\phi}$.
\end{cor}

Another useful consequence of the equivalence is the admissibility of
the transitivity rule. 
\begin{cor}[Admissibility of transitivity rule] Assume that:
  \begin{itemize}
  \item  $\jrhol{\Gamma}{\Psi}{t_1}{\sigma_1}{t_2}{\sigma_2}{\phi}$
  \item $\jrhol{\Gamma}{\Psi}{t_2}{\sigma_2}{t_3}{\sigma_3}{\phi'}$
 \end{itemize}
  Then, $\jrhol{\Gamma}{\Psi}{t_1}{\sigma_1}{t_3}{\sigma_3}{
    \phi\subst{\res\rtag}{t_2} \land \phi'\subst{\res\ltag}{t_2}}$.
\end{cor}
Finally, we prove an embedding lemma for UHOL. The proof can be
carried by induction on the structure of derivations, or using the
equivalence between UHOL and HOL (Theorem~\ref{thm:equiv-uhol-hol}).
\begin{lem}[Embedding lemma]\label{lem:emb-uhol-rhol} Assume that:
\begin{itemize}
\item $\juhol{\Gamma}{\Psi}{t_1}{\sigma_1}{\phi}$
\item $\juhol{\Gamma}{\Psi}{t_2}{\sigma_2}{\phi'}$
\end{itemize}
Then
$\jrhol{\Gamma}{\Psi}{t_1}{\sigma_1}{t_2}{\sigma_2}{
  \phi\subst{\res}{\res\ltag}\land \phi'\subst{\res}{\res\rtag}}$.
\end{lem}
The embedding is reminiscent of the approach of~\citet{BeringerH07} to
encode information flow properties in Hoare logic.

\section{Embeddings}
\label{sec:emb}

In this section, we establish the expressiveness of RHOL and UHOL by
embedding several existing refinement type systems (3 relational and 1
non-relational) from a variety of domains. All embeddings share the
common idea of \emph{separating} the simple typing information from
the more fine-grained refinement information in the translation. We
use uniform notation to represent similar ideas across the different
embeddings.  In particular, we use vertical bars $|\cdot|$ to denote
the erasure of a type into a simple type, and floor bars $\lfloor
\cdot \rfloor$ to denote the embedding of the refinement of a type in
a HOL formula.

For the clarity of exposition, we often present fragments or variants
of systems that appear in the literature, notably excluding recursive
functions that do not satisfy our well-definedness predicate.
Moreover, the embeddings are given for a version of RHOL \`a la Curry,
in which $\lambda$-abstractions do not carry the type of their bound
variable.


\subsection{Refinement types}
\label{ssec:uhol-rt}

\begin{figure*}
  \begin{mdframed}
    \begin{mathpar}

    \inferrule{\Gamma \vdash \tau}{\Gamma \vdash \tau \preceq \tau} \and
    \inferrule{\Gamma \vdash \tau_1 \preceq \tau_2 \\ \Gamma \vdash \tau_2 \preceq \tau_3}{\Gamma \vdash \tau_1 \preceq \tau_3} \and
    \inferrule{\Gamma \vdash \tau_1  \preceq \tau_2}{\Gamma \vdash \listt{\tau_1} \preceq \listt{\tau_2}} \and
    \inferrule{\Gamma \vdash \reftype{x}{\tau}{\phi}}{\Gamma \vdash \reftype{x}{\tau}{\phi} \preceq \tau} \and
    \inferrule{\Gamma \vdash \tau \preceq \sigma \\ \Gamma, \res:\tau \vdash \phi}{\Gamma \vdash \tau \preceq \reftype{x}{\sigma}{\phi}} \and
    \inferrule{\Gamma \vdash \sigma_2 \preceq \sigma_1 \\ \Gamma, x:\sigma_2 \vdash \tau_1 \preceq \tau_2}
              {\Gamma \vdash \Pi (x:\sigma_1).\tau_1 \preceq \Pi (x:\sigma_2).\tau_2} \\
    \inferrule{{}}{\Gamma, x:\tau \vdash x : \tau} \and
    \inferrule{\Gamma, x:\tau \vdash t : \sigma}{\Gamma \vdash \lambda x. t : \Pi(x:\tau).\sigma} \and
    \inferrule{\jreft{\Gamma}{t_1}{\pitype{x}{\tau}{\sigma}} \\ \jreft{\Gamma}{t_2}{\tau}}{\jreft{\Gamma}{t_1\ t_2}{\sigma\subst{x}{t_2}}} \and
    \inferrule{\jreft{\Gamma}{t}{\listt{\tau}} \\ \jreft{\Gamma}{t_1}{\sigma} \\ \jreft{\Gamma}{t_2}{\pitype{h}{\tau}{\pitype{l}{\listt{\tau}}{\sigma}}}}
              {\jreft{\Gamma}{{\rm case}\ t\ {\rm of}\ \nil \mapsto t_1 ; \cons{\_}{\_} \mapsto t_2}{\sigma}} \and
    \inferrule{\Gamma \vdash \tau \preceq \sigma \\ \jreft{\Gamma}{t}{\tau}}{\jreft{\Gamma}{t}{\sigma}} \and
    \inferrule{\jreft{\Gamma, x: \tau, f:\Pi(\reftype{y}{\tau}{y<x}).\sigma\subst{x}{y}}{t}{\sigma} \\ \wdef{f}{x}{t}}
              {\jreft{\Gamma}{\letrec{f}{x}{t}}{\pitype{x}{\tau}{\sigma}}}
                                    
    \end{mathpar}
  \end{mdframed}
  \caption{Refinement types rules (subtyping and typing)}
  \label{fig:ref-types}
\end{figure*}

Refinement types \citep{FreemanP91,Swamy+16,VazouSJVJ14} are a variant
of simple types where for every basic type $\tau$, there is a type
$\{x:\tau \mid \phi\}$ which is inhabited by the elements $t$ of
$\tau$ that satisfy the logical assertion $\phi\subst{x}{t}$. This
includes dependent refinements $\pitype{x}{\tau}{\sigma}$, in which
the variable $x$ is also bound in the refinement formulas appearing in
$\sigma$. Here we present a simplified variant of these systems.
(Refined) types are defined by the grammar
\begin{align*}
  \tau \:&:=\: \bool \mid \nat \mid \listt{\tau} \mid \reftype{x}{\tau}{\phi} \mid \pitype{x}{\tau}{\tau}
\end{align*}

As usual, we shorten $\pitype{x}{\tau}{\sigma}$ to $\tau\to\sigma$ if
$x \not\in FV(\sigma)$. We also shorten bindings of the form $x
:\reftype{x}{\tau}{\phi}$ to $\reftype{x}{\tau}{\phi}$.  Typing rules
are presented in Figure~\ref{fig:ref-types}; note that the
\rname{LETREC} rule requires that recursive definitions satisfy the
well-definedness predicate. Judgments of the form $\Gamma \vdash \tau$
are well-formedness judgments. Judgments of the form $\Gamma \vdash
\phi$ are logical judgments; we omit a formal description of the
rules, but assume that the logic of assertions is consistent with HOL,
i.e.\ $\Gamma \vdash \phi$ implies
$\jhol{\erfor{\Gamma}}{\ertype{\Gamma}{}}{\phi}$, where the erasure
functions are defined below.

This system can be embedded into UHOL in a straightforward manner. The
embedding highlights the relation between these two systems,
i.e.\ between logical assertions embedded in the types (as in
refinement types) and logical assertions at the top-level, separate
from simple types (as in UHOL).  The intuitive idea behind the
embedding is therefore to separate refinement assertions from
types. Specifically, from every refinement type we can
obtain a simple type by repeteadly extracting the type $\tau$ from
$\reftype{x}{\tau}{\phi}$. We will denote this extraction by the
translation function $\erfor{\tau}$:
\begin{mathpar}
\erfor{\bool} \defeq \bool \and
\erfor{\nat} \defeq \nat \and
\erfor{\listt{\tau}} \defeq \listt{\erfor{\tau}} \and
\erfor{\reftype{x}{\tau}{\phi}} \defeq \erfor{\tau} \and
\erfor{\pitype{x}{\tau}{\sigma}} \defeq \erfor{\tau}\to\erfor{\sigma}
\end{mathpar}

Since $\erfor{\tau}$ loses refinement information, we define a second
translation that extracts the refinement as a logical predicate over a
variable $x$ that names the typed expression. This second translation
is written $\ertype{\tau}{x}$.

%
\begin{mathpar}
\ertype{\bool}{x},\ertype{\nat}{x} \defeq \top \and
\ertype{\listt{\tau}}{x} \defeq {\rm All}(x, \lambda y.\ertype{\tau}{y}) \and
\ertype{\reftype{y}{\tau}{\phi}}{x} \defeq \ertype{\tau}{x} \wedge \phi\subst{y}{x} \and
\ertype{\pitype{x}{\tau}{\sigma}}{x} \defeq \forall y. \ertype{\tau}{y} \Rightarrow \ertype{\sigma}{xy}
\end{mathpar}
The refinement of simple types is trivial.
If $t$ is an expression of type $\reftype{x}{\tau}{\phi}$, then $t$ must
satisfy both the refinement formula $\phi$ and the refinement of
$\tau$. If $t$ is an expression of type $\pitype{x}{\tau}{\sigma}$,
then it must be the case that for every $x$ satisfying the refinement
of $\tau$, $(t\ x)$ satisfies the refinement of $\sigma$.
The refinement of a list uses the predicate ${\rm All}$, which as
defined in \S\ref{sec:hol}, means that all elements of a list satisfy
a given formula.

The syntax of assertions and expressions is exactly the same as in
HOL, and therefore there is no need for a translation. 
%
%
Embedding of types can be lifted to contexts in the natural way.
%
\begin{mathpar}
\erfor{x:\tau,\Gamma} \defeq x:\erfor{\tau},\erfor{\Gamma} \and
\ertype{x:\tau,\Gamma}{} \defeq \ertype{\tau}{x}, \ertype{\Gamma}{}
\end{mathpar}

To encode judgments, all that remains is to put the previous
definitions together. The main result about embedding typing
judgments is the following:

\begin{thm}\label{thm:ref-types-typ}
If $\jlc{\Gamma}{t}{\tau}$ is derivable in the refinement type system,
then
$\juhol{\erfor{\Gamma}}{\ertype{\Gamma}{}}{t}{\erfor{\tau}}{\ertype{\tau}{\res}}$
is derivable in UHOL.
\end{thm}
The proof is performed by induction on the structure of derivations,
using as helper result the embedding of subtyping judgments into
HOL. Since it can be proven by induction that, whenever
$\tau\preceq\sigma$, the type extractions $\erfor{\tau}$ and
$\erfor{\sigma}$ coincide, all that needs to be checked is that
$\ertype{\sigma}{}$ is a consequence of $\ertype{\tau}{}$. This is
captured by the following statement.
\begin{thm}\label{thm:ref-types-sub}
If $\Gamma \vdash \tau \preceq \sigma$ is derivable in a refinement
type system, then $\jhol{\erfor{\Gamma},
  x:\erfor{\tau}}{\ertype{\Gamma}{},
  \ertype{\tau}{x}}{\ertype{\sigma}{x}}$ is derivable in HOL.
\end{thm}
Soundness of refinement types w.r.t.\ the set-theoretic semantics
follows immediately from Theorem~\ref{thm:ref-types-typ} and the
set-theoretic soundness of UHOL (Corollary~\ref{cor:uhol:sound}).



\subsection{Relational refinement types}
\label{ssec:emb-rrt}
Relational refinement types~\citep{BFGSSZ14,BGGHRS15} are a variant of
refinement types that can be used to express relational properties via
a syntax of the form $\rtref{\res}{\tau}{\phi}$ where $\phi$ is a
relational assertion---i.e.\, it may contain a left and right copy of
$\res$, which are denoted as $\res\ltag$ and $\res\rtag$ respectively,
as well as a left and a right copy of every variable in the context.
In this section, we introduce a simple relational refinement type
system and establish a type-preserving translation to RHOL---we
compare with existing type systems at the end of the paragraph.

The syntax of relational refinement types is given by the grammar:
\begin{equation*}
\begin{array}{r@{\hspace{0.3em}}l}
    \tau ::= & \bool \mid \nat \mid \tau\rightarrow\tau \\
    T, U ::= &  \tau \vbar \listt{T} \vbar \rtprod{x}{T}{U} \vbar \rtref{x}{T}{\phi}
  \end{array}
\end{equation*}

Relational refinement types are naturally ordered by a subtyping
relation $\Gamma \vdash T \tyle U$, where $\Gamma$ is a sequence
of variables declarations of the form $x::U$.

Typing judgments are of the form $\Gamma \vdash t_1 \sim t_2 :: T$.
We present selected typing rules in Figure~\ref{fig:relty}.  Note that
the form of judgments requires that $t_1$ and $t_2$ must have the same
simple type, and the typing rules require that $t_1$ and $t_2$ have
the same structure\footnote{The typing rules displayed in the figure
  will in fact force $t_1$ and $t_2$ to be the same term modulo
  renaming. This is not the case in existing relational refinement
  type systems; however, rules that introduce different terms on the
  right and on the left are very limited; essentially, there is a rule
  similar to \rname{LLCASE-A}, and a rule for reducing in the terms of
  a judgment.}. In the \rname{CASELIST} rule, we require that both
terms reduce to the same branch; the case rule for natural numbers is
similar. The \rname{LETREC} rule uses (a straightforward adaptation
of) the $\wdef{f}{x}{t}$ predicate from our simply-typed language, and
requires that the two recursive definitions perform exactly the same
recursive calls.

Subtyping rules are
the same as in the unary case, and therefore we refer to
Figure~\ref{fig:ref-types} for them (allowing their instantiation for
relational types $T,U$ as well as unary types $\sigma, \tau$).

\begin{figure*}
  \begin{mdframed}
 \begin{mathpar}
  \inferrule*[left=VAR-RT]
   {(x:T) \in \Gamma}
   {\Gamma \vdash x_1 \sim x_2 :: T} \and
  \inferrule*[left=ABS-RT]
   {\Gamma, x :: T \vdash u_1 \sim u_2 :: U}
   {\Gamma \vdash \lambda x_1.u_1 \sim \lambda x_2.u_2 :: \rtprod{x}{T}{U}} \and
  \inferrule*[left=APP-RT]
   {\Gamma \vdash t_1 \sim t_2 :: \rtprod{x}{T}{U} \\
    \Gamma \vdash u_1 \sim u_2 :: T}
   {\Gamma \vdash t_1\ u_1 \sim t_2\ u_2 :: U \subst{x_1}{u_1}\subst{x_2}{u_2}} \and
   \inferrule*[left=NIL]
   {\Gamma \vdash T}
   {\Gamma \vdash \nil \sim \nil :: \listt{T}} \and
   \inferrule*[left=CONS]
   {\Gamma \vdash h_1 \sim h_2 :: T \\
    \Gamma \vdash t_1 \sim t_2 :: \listt{T}}
   {\Gamma \vdash \cons{h_1}{t_1} \sim \cons{h_2}{t_2} :: \listt{T}} \and
   \inferrule*[left=Sub]
   {\Gamma \vdash t_1 \sim t_2 :: T \\
    \Gamma \vdash T \tyle U}
   {\Gamma \vdash t_1 \sim t_2 :: U} \and

  \inferrule*[left=LETREC-RT]
             { \Gamma, x :: T, f :: \Pi(\rtref{y}{T}{(y_1, y_2) < (x_1, x_2)}).U\subst{x}{y} \vdash t_1 \sim t_2 :: U\\\\
     \Gamma \vdash \rtprod{x}{T}{U} \\ \wdef{f}{x}{t}}
   {\Gamma \vdash {\rm letrec}\ f_1\ x_1 = t_1 \sim {\rm letrec}\ f_2\ x_2 = t_2 :: \rtprod{x}{T}{U}} \and
 \inferrule*[left=CASELIST]
            {\Gamma \vdash t_1 \sim t_2 :: \listt{\tau} \\ \Gamma \vdash t_1 = \nil \Leftrightarrow t_2 = \nil \\
              \Gamma \vdash u_1 \sim u_2 :: T \\ \Gamma \vdash v_1 \sim v_2 :: \rtprod{h}{\tau}{\rtprod{t}{\listt{\tau}}{T}}}
            {\Gamma \vdash {\rm case}\ t_1\ {\rm of}\ \nil \mapsto u_1 ; \cons{\_}{\_} \mapsto v_1 \sim
                             {\rm case}\ t_2\ {\rm of}\ \nil \mapsto u_2 ; \cons{\_}{\_} \mapsto v_2 :: T}
 \end{mathpar}
  \end{mdframed}
  \caption{Relational Typing (Selected Rules)}\label{fig:relty} 
\end{figure*}

The embedding of refinement types into UHOL can be adapted to the
relational setting. From each relational refinement type $T$ we can
extract a simple type $\erfor{T}$. On the other hand, we can erase
every relational refinement type $T$ into a relational formula
$\errtype{T}{}$, which is parametrized by two expressions and defined
as follows:

\begin{mathpar}
\errtype{\listt{\tau}}{x_1, x_2} \defeq \bigwedge_{i\in\{1,2\}} {\rm All}(x_i, \lambda y.\ertype{\tau}{y}) \and
\errtype{\listt{T}}{x_1, x_2} \defeq {\rm All2}(x_1, x_2, \lambda y_1.\lambda y_2.\errtype{T}{y_1, y_2}) \and
\errtype{\reftype{y}{\tau}{\phi}}{x_1, x_2} \defeq \bigwedge_{i\in\{1,2\}} \ertype{\tau}{x_i} \wedge \phi\subst{y}{x_i} \and
\errtype{\rtref{y}{T}{\phi}}{x_1, x_2} \defeq \errtype{T}{x_1,x_2} \wedge \phi\subst{y_1}{x_1}\subst{y_2}{x_2} \and
\errtype{\pitype{y}{\tau}{\sigma}}{x} \defeq \bigwedge_{i\in\{1,2\}} \forall y. \ertype{\tau}{y} \Rightarrow \ertype{\sigma}{xy} \and
\errtype{\rtprod{y}{T}{U}}{x_1, x_2} \defeq \forall y_1 y_2. \errtype{T}{y_1, y_2} \Rightarrow \errtype{\sigma}{x_1 y_1,\ x_2 y_2}
\end{mathpar}

The predicate ${\rm All2}$ relates two lists elementwise and is defined axiomatically:
\begin{mathpar}
{\rm All2}([], [], \lambda x_1. \lambda x_2.\phi) \and
\forall h_1 h_2 t_1 t_2. {\rm All}(t_1, t_2, \lambda x_1. \lambda x_2.\phi) \Rightarrow \phi(h_1, h_2) \Rightarrow {\rm All}(\cons{h_1}{t_1}, \cons{h_2}{t_2}, \lambda x_1. \lambda x_2. \phi)
\end{mathpar}

To extend the embedding to contexts, we need to duplicate every variable in them:
\begin{mathpar}
  \erfor{x::T, \Gamma} \defeq x_1,x_2 : \erfor{T}, \erfor{\Gamma} \and
  \errtype{x::T, \Gamma}{} \defeq \errtype{T}{x_1,x_2}, \errtype{\Gamma}{}
\end{mathpar}

Now we state the main result:

\begin{thm}[Soundness of embedding]\label{thm:emb-rrt}
  If $\Gamma \vdash t_1 \sim t_2 :: T$, then
  $\jrhol{\erfor{\Gamma}}{\errtype{\Gamma}{}}{t_1}{\erfor{T}}{t_2}{\erfor{T}}{\errtype{T}{\res\ltag, \res\rtag}}$
  Also, if $\Gamma \vdash T \tyle U$ then
  $\jhol{\erfor{\Gamma},x_1,x_2:\erfor{T}}{\errtype{\Gamma}{}, \errtype{T}{x_1,x_2}}{\errtype{U}{x_1.x_2}}$.
\end{thm}
\begin{proof} The proof proceeds by induction on the structure of derivations. 
\end{proof}
Soundness of relational refinement types w.r.t.\, set-theoretical
semantics follows immediately from Theorem~\ref{thm:emb-rrt} and the
set-theoretical soundness of RHOL (Corollary~\ref{cor:rhol:sound}).

\begin{cor}[Soundness of relational refinement types]
If $\Gamma \vdash t_1 \sim t_2 :: T$, then for every valuation
$\theta\models\Gamma$ we have
$(\rinterp{\theta}{t_1},\rinterp{\theta}{t_2})\in \rinterp{\theta}{T}$.
\end{cor}






\subsection{Dependency core calculus}
\label{sec:dcc}

The Dependency Core Calculus (DCC)~\citep{AbadiBHR99} is a
higher-order calculus with a type system that tracks data
dependencies. DCC was designed as a unifying framework for dependency
analysis and it was shown that many other calculi for information flow
analysis~\citep{HeintzeR98,volpano-smith-irvine}, binding-time
analysis~\citep{binding-analysis}, and program slicing, all of which
track dependencies, can be embedded in DCC. Here, we show how a
fragment of DCC can be embedded into RHOL. Transitively, the
corresponding fragments of all the aforementioned calculi can also be
embedded in RHOL. (The fragment of DCC we consider excludes recursive
functions. DCC admits general recursive functions, while our
definition of RHOL only admits a subset of these. Extending the
embedding to recursive functions admitted by RHOL is not difficult.)

DCC is an extension of the simply typed lambda-calculus with a monadic
type family $\ttmonad{\ell}(\tau)$, indexed by \emph{labels} $\ell$,
which are elements of a lattice. Unlike other uses of monads, DCC's
monad does not encapsulate any effects. Instead, its only goal is to
track dependence. The type system forces that the result of an
expression of type $\ttmonad{\ell}(\tau)$ can \emph{depend} on an
expression of type $\ttmonad{\ell'}(\tau')$ only if $\ell' \sqsubseteq
\ell$ in the lattice. Dually, if $\ell' \not\sqsubseteq \ell$, then
even if an expression $e$ of type $\ttmonad{\ell}(\tau)$ mentions a
variable $x$ of type $\ttmonad{\ell'}(\tau')$, then $e$'s result must
be \emph{independent} of the substitution provided for $x$ during
evaluation.

For simplicity and without any loss of generality, we consider here
only a two point lattice $\{L, H\}$ with $L \sqsubset H$. The syntax
of DCC's types and expressions is shown below. We use $e$ to denote
DCC expressions, to avoid confusion with HOL's expressions.
\[\begin{array}{lll}
\tau & ::= & \bool \mid \tau \rightarrow \tau \mid \tau \times \tau \mid
\ttmonad{\ell}(\tau)\\
e & ::= & x \mid \lambda x.e \mid e_1\ e_2 \mid \tbool \mid \fbool
\mid \casebool{e}{e_t}{e_f} \mid \pair{e_1}{e_2} \mid \pi_1(e) \mid
\pi_2(e) \mid \eret{\ell}(e) \mid \ebind(e_1, x.e_2)
\end{array}
\]
Here, $\eret{\ell}(e)$ and $\ebind(e_1, x.e_2)$ are respectively the
return and bind constructs for the monad
$\ttmonad{\ell}(\tau)$. Typing rules for these two constructs are
shown below. Typing rules for the remaining constructs are the
standard ones.
\begin{mathpar}
\inferrule{\Gamma \vdash e: \tau}
      {\Gamma \vdash \eret{\ell}(e) : \ttmonad{\ell}(\tau)}
\and
\inferrule{\Gamma \vdash e_1: \ttmonad{\ell}(\tau_1) \qquad
  \Gamma, x: \tau_1 \vdash e_2: \tau_2 \qquad \tau_2 \searrow \ell}
      {\Gamma \vdash \ebind(e_1, x.e_2): \tau_2}
\end{mathpar}
The crux of the dependency tracking is the relation $\tau_2 \searrow
\ell$ in the premise of the rule for $\ebind$. The relation, read
``$\tau_2$ protected at level $\ell$'' and defined below, informally
means that all primitive (boolean) values extractable from $e_2$ are
protected by a monadic construct of the form $\ttmonad{\ell'}(\tau)$,
with $\ell \sqsubseteq \ell'$. Hence, the rule forces that the result
obtained by eliminating the type $\ttmonad{\ell}(\tau_1)$ flow only
into types protected at $\ell$ in this sense.
\begin{mathpar}
  \inferrule{\ell \sqsubseteq \ell'}{\ttmonad{\ell'}(\tau) \searrow \ell} \and
  \inferrule{\tau \searrow \ell}{\ttmonad{\ell'}(\tau) \searrow \ell} \and
  \inferrule{\tau_1 \searrow \ell \qquad \tau_2 \searrow \ell}{\tau_1 \times \tau_2 \searrow \ell}\and
  \inferrule{\tau_2 \searrow \ell}{\tau_1 \rightarrow \tau_2 \searrow \ell}
\end{mathpar}
This fragment of DCC has a relational set-theoretic
interpretation. For every type $\tau$, we define a carrier set
$\erase{\tau}$:
\begin{mathpar}
  \erase{\bool} \defeq \bool \and
  \erase{\tau_1 \rightarrow \tau_2} \defeq \erase{\tau_1} \rightarrow \erase{\tau_2} \and
  \erase{\tau_1 \times \tau_2} \defeq \erase{\tau_1} \times \erase{\tau_2} \and
  \erase{\ttmonad{\ell}(\tau)} \defeq \erase{\tau}
\end{mathpar}
Next, every type $\tau$ is interpreted as a lattice-indexed family of
relations $\dccform{\adversary}{\tau} \subseteq \erase{\tau} \times
\erase{\tau}$. The role of the lattice element $\adversary$ is that it
defines what can be \emph{observed} in the system. Specifically, an
expression of type $\ttmonad{\ell}(\tau)$ can be observed only if
$\ell \sqsubseteq \adversary$. When $\ell\not\sqsubseteq \adversary$,
expressions of type $\ttmonad{\ell}(\tau)$ look like
``black-boxes''. Technically, we force
$\dccform{\adversary}{\ttmonad{\ell}(\tau)} = \erase{\tau} \times
\erase{\tau}$ when $\ell \not\sqsubseteq \adversary$. DCC's typing
rules are sound with respect to this model. The soundness implies that
if $\ell \not \sqsubseteq \ell'$ and $x: \ttmonad{\ell}(\bool) \vdash
e: \ttmonad{\ell'}(\bool)$, then for $e_1, e_2:
\ttmonad{\ell}(\bool)$, $e[e_1/x]$ and $e[e_2/x]$ are equal booleans
in the set-theoretic model. This result, called noninterference,
formalizes that DCC's dependency tracking is correct.

To translate DCC to RHOL, we actually embed this set-theoretic model
in RHOL. We start by defining an erasing translation, $\erase{\tau}$,
from DCC's types into RHOL's simple types. This translation is exactly
the same as the definition of carrier sets shown above, except that we
treat $\times$ and $\to$ as RHOL's syntactic type constructs instead
of set-theoretic constructs. Next, we define an erasure of terms:
\begin{mathpar}
  \erase{\tbool} \defeq \tbool \and
  \erase{\fbool} \defeq \fbool \and
  \erase{\casebool{e}{e_t}{e_f}} \defeq \casebool{\erase{e}}{\erase{e_t}}{\erase{e_f}} \and 
  \erase{x} \defeq x \and
  \erase{\lambda x.e} \defeq \lambda x. \erase{e} \and
  \erase{e_1\ e_2} \defeq \erase{e_1}\ \erase{e_2} \and
  \erase{\pair{e_1}{e_2}} \defeq \pair{\erase{e_1}}{\erase{e_2}} \and
  \erase{\pi_1(e)} \defeq \pi_1(\erase{e}) \and
  \erase{\pi_2(e)} \defeq \pi_2(\erase{e}) \and
  \erase{\eret{\ell}(e)} \defeq \erase{e} \and
  \erase{\ebind(e_1, x.e_2)} \defeq (\lambda x.\erase{e_2})\ \erase{e_1}
\end{mathpar}
It is fairly easy to see that if $\vdash e: \tau$ in DCC, then $\vdash
\erase{e}: \erase{\tau}$. Next, we define the lattice-indexed family
of relations $\dccform{\adversary}{\tau}$ in HOL. For technical
convenience, we write the relations as logical assertions, indexed by
variables $x, y$ representing the two terms to be related.
%
\begin{mathpar}
  \dccform{\adversary}{\bool}(x,y) \defeq (x=\tbool \wedge y = \tbool)
  \vee (x=\fbool \wedge y=\fbool) \and
  \dccform{\adversary}{\tau_1 \rightarrow \tau_2}(x,y) \defeq
  \forall v, w.\, \dccform{\adversary}{\tau_1}(v,w) \Rightarrow
  \dccform{\adversary}{\tau_2}(x \ v, y\ w) \and
  \dccform{\adversary}{\tau_1 \times \tau_2}(x,y) \defeq
  \dccform{\adversary}{\tau_1}(\pi_1(x), \pi_1(y)) \wedge
  \dccform{\adversary}{\tau_2}(\pi_2(x), \pi_2(y)) \and 
  \dccform{\adversary}{\ttmonad{\ell}(\tau)}(x,y) \defeq
  \left\lbrace\begin{array}{ll}
  \dccform{\adversary}{\tau}(x,y) & \ell \sqsubseteq \adversary \\
  \top & \ell \not\sqsubseteq \adversary
  \end{array}
  \right.
\end{mathpar}
The most important clause is the last one: When $\ell \not\sqsubseteq
\adversary$, any two $x, y$ are in the relation
$\dccform{\adversary}{\ttmonad{\ell}(\tau)}$. This generalizes to all
protected types in the following sense.
\begin{lem}\label{lem:dcc:confinement}
If $\ell \not\sqsubseteq \adversary$ and $\tau \searrow \ell$, then
$\vdash \forall x,y. (\dccform{\adversary}{\tau}(x,y) \equiv \top)$ in HOL.
\end{lem}
The translations extend to contexts as follows:
\[\begin{array}{c}
\erase{x^1: \tau_1, \ldots, x^n: \tau_n} \defeq x^1_1: \erase{\tau_1}, x^1_2: \erase{\tau_1}, \ldots, x^n_1: \erase{\tau_n}, x^n_2: \erase{\tau_n}\\
\dccform{\adversary}{x^1: \tau_1, \ldots, x^n: \tau_n} \defeq
\dccform{\adversary}{\tau_1}(x^1_1, x^1_2), \ldots,
\dccform{\adversary}{\tau_n}(x^n_1, x^n_2)
\end{array}
\]

The following theorem states that the whole translation is sound: It
preserves well-typedness. In the statement of the theorem,
$\erase{e}_1$ and $\erase{e}_2$ replace each variable $x$ in
$\erase{e}$ with $x_1$ and $x_2$, respectively.

\begin{thm}[Soundness of embedding]\label{thm:dcc:embed}
If $\Gamma \vdash e: \tau$ in DCC, then for all $\adversary \in \{L,
H\}$:
$\jrhol{\erase{\Gamma}} {\dccform{\adversary}{\Gamma}} {\erase{e}_1} {\erase{\tau}} {\erase{e}_2} {\erase{\tau}} {\dccform{\adversary}{\tau}(\res\ltag,\res\rtag)}$.
\end{thm}

DCC's noninterference theorem is a corollary of this theorem and the
soundness of RHOL in set theory.



\subsection{Relational cost}
\label{ssec:emb-rc}
RelCost~\citep{CBGGH17} is a relational refinement type-and-effect
system designed to reason about relative cost---the difference in the
execution costs of two similar programs or of two runs of the same
program on two different inputs. RelCost combines reasoning about the
maximum and minimum costs of a single program with relational
reasoning about the relative cost of two programs. RelCost is based on
the observation that relational reasoning about structurally related
expressions can improve precision in reasoning about the relative
cost, but if this approach fails one can always fall back to
establishing an upper bound on the relative cost the difference of the
maximum cost of one program and the minimum cost of the other.  Here,
we show how a fragment of RelCost can be embedded into RHOL. Similar
to what we did for DCC, to just convey the main intuition, we consider
a fragment of RelCost excluding recursive functions.  The syntax of
RelCost is based on two sorts of types:
\[
\begin{array}{rclr}
A &::=& \nat \mid \listt{A}[n] \mid A \uarr{k}{l} A \mid
        \tforall{i}{k}{l} A & \text{(unary types)}\\
\tau &::=& \nat_r \mid \listt{\tau}[n]^\alpha \mid
  \tau \tarrd{k} \tau \mid \tforalld{i}{k} \tau \mid U A \mid \tbox
           \tau 
& \text{(relational types)}
\end{array}
\]
Unary types are used to type one program and they are mostly standard
except for the effect annotation $\mbox{exec}(k,l)$ on arrow types and
universally quantified types representing the min and max cost $k$ and
$l$ of the body of the closure, respectively. Relational types ascribe
two programs, so they are interpreted as pairs of expressions.  In
relational types, arrow types and universally quantified types have an
effect annotation $\mbox{diff}(k)$ representing the relative cost $k$
of the two closures. Besides, the superscript $\alpha$ refines list
types with the number of elements that can differ in two lists. The
type $U A$ is the weakest relation over elements of the unary type
$A$, i.e.\ it can be used to type two arbitrary terms, while the type
$\tbox \tau$ is the diagonal subrelation of $\tau$, i.e.\ it can be
used to type only two terms that are equal.
\begin{figure*}
  \begin{mdframed}
 \begin{mathpar}
 \inferrule*[left=var]
 {\Omega(x) = \grt }{\octx \jtype{0}{0}{x}{\grt}}
 \and
 \inferrule*[left=const]
  {
  }
  {
    \octx \jtype{0}{0}{\econst}{\tint}
  }
  \and
  \inferrule*[left=lam]
    {\Delta; \Phi_a;  x: \grt_1, \Omega  \jtype{k}{l}{t}{\grt_2} }
    {
      \octx
      \jtype{0}{0}{ \lambda x.t }{\grt_1 \uarr{k}{l}
        \grt_2}
    }
 \and
 \inferrule*[left=app]{
   \octx \jtype{k_1}{l_1}{t_1}{\grt_1\uarr{k}{l} \grt_2} \\
   \octx \jtype{k_2}{l_2}{t_2}{\grt_1}}
 {
   \octx \jtype{k_1+k_2+k+\kapp}{l_1+l_2+l+\kapp}{ t_1 \, t_2}{
     \grt_2}}
\end{mathpar}
   \hrule
 \begin{mathpar} 
 \inferrule*[left=r-var]
  {\Gamma(x) = \tau}{\ctx \jtypediff{0}{x}{x}{\tau}}
  \and
  \inferrule*[left=r-const]{\ }
  {\ctx \jtypediff{0}{\econst}{\econst}{\trint}}
  \and
  \inferrule*[left=r-cons1]
  { \ctx \jtypediff{l_1}{t_1}{t_1'}{\tau} \\
    \ctx \jtypediff{l_2}{t_2}{t_2'}{\tlist{n}{\alpha}{\tau}}
  }
  { \ctx \jtypediff{l_1+l_2}{ \econs(t_1,t_2) }{ \econs(t_1',t_2')}{\tlist{n+1}{\alpha+1}{\tau}}}
  \and
  \inferrule*[left=r-cons2]
  { \ctx \jtypediff{l_1}{t_1}{t_1'}{\tbox{\tau}} \\
    \ctx \jtypediff{l_2}{t_2}{t_2'}{\tlist{n}{\alpha}{\tau}}}
  { \ctx \jtypediff{l_1+l_2}{ \econs(t_1,t_2) }{ \econs(t_1',t_2')}{\tlist{n+1}{\alpha}{\tau}}}
  \and
  \inferrule*[left=r-caseL]
  {\ctx \jtypediff{l}{t}{t'}{\tlist{n}{\alpha}{\tau}}\\
   \Delta; \Phi_a \wedge n = 0  ;\Gamma \jtypediff{l'}{t_1}{t_1'}{\tau'}\\
   i, \Delta; \Phi_a \wedge n = i+1;  h: \tbox{\tau},
   tl : \tlist{i}{\alpha}{\tau}, \Gamma \jtypediff{l'}{t_2}{t_2'}{\tau'}\\
   i, \beta, \Delta; \Phi_a \wedge n = i+1 \wedge \alpha = \beta +1 ;
   h: \tau, tl : \tlist{i}{\beta}{\tau}, \Gamma \jtypediff{l'}{t_2}{t_2'}{\tau'} 
   }
  {\ctx \jtypediff{l+l'}{ \ecase t \eof \enil \rightarrow t_1 ~|~ h::tl \rightarrow
    t_2}{\ecase t' \eof \enil \rightarrow t_1' ~|~ h::tl \rightarrow
    t_2'}{\tau'}}
 \and
 \inferrule*[left=r-lam]{
       \Delta; \Phi_a;  x: \tau_1, \Gamma  \jtypediff{l}{t_1}{t_2}{\tau_2}}
            {\ctx \jtypediff{0}{ \lambda x.t_1}{\lambda x.t_2}{
              \tau_1 \tarrd{l} \tau_2}}
 \and
 \inferrule*[left=r-app]
 {\ctx \jtypediff{l_1}{t_1}{t_1'}{\tau_1}
  \tarrd{l} \tau_2 \\\\
 \ctx \jtypediff{l_2}{ t_2}{t_2'}{\tau_1}}
  {\ctx \jtypediff{l_1+l_2+l}{t_1 \, t_2}{t_1' \, t_2'}{\tau_2}}
 \and
\inferrule*[left=r-iLam]{
    i::S, \ctx \jtypediff{l}{t}{t'}{\tau} \\\\
    i \not \in \fiv{\Phi_a; \Gamma} 
  }
  { \ctx \jtypediff{0}{\eLam t}{\eLam t'}{\tforalld{i}{l} \tau}}
 \and
\inferrule*[left=r-iApp]{
    \ctx \jtypediff{l}{t}{t'}{\tforalld{i}{l'} \tau}\\\\
    \Delta \vdash J : S
  }
  {\ctx \jtypediff{l+l'[J/i]}{t \eApp}{t' \eApp}{\tau\{ J/ i\}}}
 \and
  \inferrule*[left=nochange]
{
\Delta; \Phi_a; \Gamma \jtypediff{l}{t}{t}{\tau}\\\\
\forall x \in dom(\Gamma).~~
\Delta; \Phi_a  \jsubtype{\Gamma(x)}{\tbox{\Gamma(x)}}
}
{
\Delta; \Phi_a; \Gamma, \Gamma'; \Omega \jtypediff{0}{t}{t}{\tbox{\tau}}
}
 \and
    \inferrule*[left=switch]{
    \Delta; \Phi_a; \overline{\Gamma} \jtype{k_1}{l_1}{t_1}{\grt}\\
    \Delta; \Phi_a; \overline{\Gamma} \jtype{k_2}{l_2}{t_2}{\grt}
    } {\ctx \jtypediff{l_1-k_2}{t_1}{t_2}{\tcho{\grt}}}
 \end{mathpar}
  \end{mdframed}
  \caption{RelCost Unary and Relational Typing (Selected Rules)}\label{fig:relcost} 
\end{figure*}
There are two kinds of typing judgments, unary and relational:
\begin{equation*}
\Delta; \Phi; \Omega \jtype{k}{l}{t}{A}
\qquad
\Delta; \Phi; \Gamma \jtypediff{l}{t_1}{t_2}{\tau}
\end{equation*}
The unary judgment states that the execution cost of $t$ is lower
bounded by $k$ and upper bounded by $l$, and the expression $t$ has
the unary type $A$. The relational judgment states that the relative
cost of $t_1$ with respect to $t_2$ is upper bounded by $l$ and the
two expressions have the relational type $\tau$. Here $\Omega$ is a
unary type environment, $\Gamma$ is a relational type environment,
$\Delta$ is an environment for index variables and $\Phi$ for assumed
constraints over the index terms. Figure~\ref{fig:relcost} shows
selected rules.

To embed RelCost in RHOL, we define a monadic-style cost-instrumented
translation of RelCost types. The translation is given in two-steps:
First, we define an erasure of cost and size information into simple
types and then we define a cost-passing style translation of simple
types with a value-translation and an expression-translation.  The
erasure function is defined as follows:
\begin{mathpar}
\erase{\nat}\defeq \erase{\nat_r}\defeq\nat

\erase{\listt{A}[n]} \defeq \erase{\listt{A}[n]^\alpha} \defeq \listt{\erase{A}}

\erase{UA}\defeq \erase{\tbox A}\defeq\erase{A}


\erase{\tforall{i}{k}{l} A} \defeq \erase{\tforalld{i}{k} A} \defeq \nat \to \erase{A} 

\erase{A \tarrd{k} B}\defeq \erase{A \uarr{k}{l} B}\defeq\erase{A}\to \erase{B}
\end{mathpar}
The cost-passing style translation of \emph{simple} types is 
\begin{mathpar}
\cost{\nat}_v\defeq\nat

\cost{\listt{A}}_v \defeq \listt{\cost{A}_v} 

\cost{A \to  B}_v\defeq\cost{A}_v\to \cost{B}_e

\cost{A}_e \defeq \cost{A}_v\times \nat  
\end{mathpar}
Guided by the translation of types above we can provide a
cost-instrumented translation of simply-typed $\lambda$-expressions
(Figure~\ref{fig:cost-translation}). This translation maps an
expression of the simple type $\tau$ to an expression of type $\tau
\times \nat$, where the second component is the number of reduction
steps under an eager, call-by-value reduction strategy (which is the
semantics of RelCost). It is fairly easy to see that this translation
preserves typability and that it counts steps accurately.

However, this translation forgets the cost and size information in
types. To recover these, we define a HOL formula for every unary
type. But, first, we define axiomatically a predicate ${\sf
  listU}(n,l,P)$ that captures size information about lists:
\begin{mathpar}
\forall l,P.{\sf listU}(0,l,P) \equiv l=[]
\and
\forall n, l, P.{\sf listU}(n+1,l,P) \equiv
  \exists w_1, w_2. l = \cons{w_1}{w_2} \land P(w_1) \land {\sf listU}(n, w_2, P)
\end{mathpar}
We can now define a HOL formula inductively on unary types.
\begin{mathpar}
\costt{\nat}_v(x)\defeq\top

\costt{\listt{A}[n]}_v(x) \defeq {\sf listU}(n, x, \costt{A}_v)


\costt{A \uarr{k}{l} B}_v(x)\defeq\forall y.\costt{A}_v(y)\Rightarrow \costt{B}_e^{k,l}(xy)

\costt{\tforall{i}{k}{l} A}_v(x) \defeq \forall y. \top \Rightarrow \forall
i. \costt{A}_e^{k,l}(xy)

\costt{A}_e^{k,l}(x) \defeq \costt{A}_v(\pi_1 x)\land k\leq \pi_2 x\leq l
\end{mathpar}
The type translation can also be extended to type environments:
$\cost{\erase{x_1:A_1,\ldots,x_n:A_n}}=x_1:\cost{\erase{A_1}}_v,\ldots,x_n:\cost{\erase{A_n}}_v$
Similarly, we can associate to a type environment an HOL context that
we can use to recover the cost and size information:
$\costt{x_1:A_1,\ldots,x_n:A_n}=\costt{A_1}_v(x_1),\ldots,\costt{A_n}_v(x_n).$
\begin{figure*}
\begin{mdframed}
\begin{mathpar}
\cost{x} \defeq (x,0)

\cost{\lambda x.t} \defeq (\lambda x. \cost{t}, 0)

\cost{\Lambda .t} \defeq (\lambda \_. \cost{t}, 0)

\cost{t u} \defeq \mathsf{let}\, x=\cost{t} \, \mathsf{in}\, \mathsf{let}\, y=\cost{u}\, \mathsf{in}\, \mathsf{let}\, z=\pi_1(x)\  \pi_1(y)\, \mathsf{in}\, (\pi_1(z),\pi_2(x) + \pi_2(y) + \pi_2(z) +c_{app})   

\cost{t[]} \defeq \mathsf{let}\, x=\cost{t} \, \mathsf{in}\, \mathsf{let}\, y = \pi_1(x)\ 0\, \mathsf{in}\, (\pi_1(y),\pi_2(x) + \pi_2(y) + c_{iapp})     

\cost{\mathsf{nil}} \defeq (\mathsf{nil},0)

\cost{\mathsf{cons}(t_1,t_2)} \defeq \mathsf{let}\, x=\cost{t_1} \, \mathsf{in}\, \mathsf{let}\, y=\cost{t_2}\, \mathsf{in}\, (\cons{\pi_1(x)}{\pi_1(y)} , \pi_2(x) + \pi_2(y))

\cost{\ecase t' \eof \enil \rightarrow t_1' ~|~ h::tl \rightarrow
    t_2'} \defeq \left \{ \begin{array}{l}
\mathsf{let}\, x=\cost{t'} \, \mathsf{in}\, 
\ecase \pi_1(x) \eof\\
 \enil \rightarrow 
\mathsf{let}\, y=\cost{t_1'} \, \mathsf{in}\, 
(\pi_1(y), \pi_2(x)+\pi_2(y)+c_{case})\\  ~|~ h::tl \rightarrow
\mathsf{let}\, y=\cost{t_2'} \, \mathsf{in}\,
(\pi_1(y), \pi_2(x)+\pi_2(y)+c_{case})
\end{array} \right .

\end{mathpar}
\end{mdframed}
\caption{Cost-instrumented translation of expressions.}
\label{fig:cost-translation}
\end{figure*}
Now we can provide a cost-instrumented translation of unary judgments.
\begin{thm}\label{thm:unaryembed}
If $\Delta; \Phi; \Omega \jtype{k}{l}{t}{A}$, then: $
\juhol{\cost{\erase{\Omega}},\Delta}{\Phi,
  \costt{\Omega}}{\cost{t}}{\cost{\erase{A}}_e}{\costt{A}_e^{k,l}(\res)}
$
\end{thm}






For the embedding of cost and size information in the relational case
we first define a predicate  ${\sf listR}(n,l_1,l_2,a, P)$ in HOL
axiomatically:
\begin{mathpar}
\forall l_1,l_2,a,P.{\sf listR}(0,l_1,l_2,a,P) \equiv l_1=l_2=[]
\and
\forall n, l_1,l_2,a,P.{\sf listR}(n+1,l_1,l_2,a,P) \equiv
\begin{array}{c}\exists w_1,z_1,w_2,z_2. 
l_1=w_1::w_2\land l_2=z_1::z_2\land P(w_1,z_1)\land \\
( ( (w_1=z_1)\land {\sf listR}(n,w_2,z_2,a,P) )\lor\\
(a > 0 \land \exists b.\, a = b + 1 \land {\sf listR}(n,w_2,z_2,b,P))
)
\end{array}
\end{mathpar}

Let $\overline{\tau}$ denote RelCost's erasure of the binary type
$\tau$ to a unary type.\footnote{In RelCost, this erasure is written
  $|\tau|$. We use a different notation to avoid confusion with our
  own erasure function from RelCost's types to simple types.} This
erasure maps $\listt{\tau}[n]^\alpha$ to $\listt{\overline{\tau}}[n]$,
$\tau \tarrd{l} \sigma$ to $\overline{\tau} \uarr{0}{\infty}
\overline{\sigma}$, etc. Next, we define HOL formulas for the binary
types.
\begin{mathpar}
\rcostt{\nat}_v(x,y)\defeq x=y
\and
\rcostt{UA}_v(x,y)\defeq \costt{A}_v(x)\land \costt{A}_v(y)
\and
\rcostt{\tbox \tau}_v(x,y)\defeq (x=y) \land (\rcostt{\tau}_v(x,y))
\and
\rcostt{\tau \tarrd{l} \sigma}_v(x,y)\defeq
\begin{array}{l}
  \costt{\overline{\tau} \uarr{0}{\infty} \overline{\sigma}}_v (x) \land
  \costt{\overline{\tau} \uarr{0}{\infty} \overline{\sigma}}_v (y) \land 
  (\forall z_1,z_2.\rcostt{\tau}_v(z_1,z_2)\Rightarrow \rcostt{\sigma}_e^l(xz_1,yz_2))
\end{array}
\and
\rcostt{\tforalld{i}{l} \tau}_v(x,y) \defeq
\begin{array}{l}
  \costt{\tforall{i}{0}{\infty}\overline{\tau}}_v(x) \land
  \costt{\tforall{i}{0}{\infty}\overline{\tau}}_v(y) \land (\forall
  z_1 z_2. \top \Rightarrow \forall i. \rcostt{\tau}_e^{l}(x z_1,y z_2))
\end{array}
\and
\rcostt{\listt{\tau}[n]^\alpha}_v(x,y) \defeq
{\sf listR}(n,x,y,\alpha,\rcostt{\tau}_v)
\and
\rcostt{\tau}_e^{l}(x,y) \defeq \rcostt{\tau}_v(\pi_1 x,\pi_1 y)\land (\pi_2 x-\pi_2 y \leq l)
\end{mathpar}
The type translation can also be extended to relational type
environments pointwise: $\rcost{x^1:\tau_1,\ldots,x^n:\tau_n} \defeq
x^1\ltag:\cost{\erase
  {\tau_1}}_v,x^1\rtag:\cost{\erase{\tau_1}}_v,\ldots,x^n\ltag:\cost{\erase{\tau_n}}_v,x^n\rtag:\cost{\erase{\tau_n}}_v$
We also need to derive from a type relational environment an HOL
context that remembers the cost and size information:
$\rcostt{x^1:\tau_1,\ldots,x^n:\tau_n} {\defeq}
\rcostt{\tau_1}_v(x^1\ltag,x^1\rtag),\ldots,\rcostt{\tau_n}_v(x^n\ltag,x^n\rtag).$
%
Now we can provide the translation of relational judgments. 
\begin{thm}\label{thm:relationalembed}
If $\Delta; \Phi; \Gamma \jtypediff{l}{t_1}{t_2}{\tau}$, then: $
\jrhol{\rcost{\Gamma},\Delta}{\Phi,\rcostt{\Gamma}}{\cost{t_1}\ltag}{\cost{\erase{\tau}}_e}{\cost{t_2}\rtag}{\cost{\erase{\tau}}_e}{
  \rcostt{\tau}_e^l(\res\ltag,\res\rtag)} $, where $\cost{t_i}_j$ is a
copy of $t_i$ where each variable $x$ is replaced by a variable $x_j$
for $j\in\{1,2\}$.
\end{thm}

RelCost's type-soundness theorem can be derived from
Theorem~\ref{thm:relationalembed} and the soundness of RHOL in set
theory.

\section{Examples}
\label{sec:ex}

We present some illustrative examples to show how RHOL's rules work in
practice. Our first example shows the functional equivalence of two
recursive functions that are synchronous---they perform the same
number of recursive calls. The second example shows the equivalence of
two asynchronous recursive functions. Our third example illustrates
reasoning about the relative cost of two programs, using an encoding
similar to that of RelCost, but the example cannot be verified in
RelCost itself.


\subsection{First example: factorial}

We show that the two following standard implementations of factorial,
with and without an accumulator, are functionally equivalent:
\begin{align*}
{\rm fact}_1 \:&\defeq\: {\rm letrec}\ f_1 \ n_1 = \casenat{n_1}{1}{\lambda x_1. (S\ x_1)*(f_1\ x_1)} \\
{\rm fact}_2 \:&\defeq\: {\rm letrec}\ f_2\ n_2 = \lambda acc. \casenat{n_2}{acc}{\lambda x_2. f_2\ x_2\ ( (S\ x_2) * acc)}
\end{align*}
Our goal is to prove that:
\[
\jrhol{\emptyset}{\emptyset}{{\rm fact_1}}{\nat\to\nat}{{\rm fact_2}}{\nat\to\nat\to\nat}{\forall n_1 n_2. n_1 = n_2 \Rightarrow \forall acc. (\res\ltag\ n_1)* acc = \res\rtag\ n_2\ acc}
\]
%
%
The proof starts by applying \rname{LETREC} rule, which has its main premise:
\[\begin{array}{c}
     \sjrhol{\Gamma,  n_1,n_2 : \nat, f_1 : \nat\to\nat, f_2 : \nat\to\nat\to\nat}
            {\Psi}
            {\casenat{n_1}{1}{\lambda x_1. (S\ x_1)*(f_1\ x_1)}}{\nat\to\nat\to\nat}
            {\lambda acc. \casenat{n_2}{acc}{\lambda x_2. f_2\ x_2\ ( (S\ x_2) * acc)}}{\nat\to\nat\to\nat\to\nat}
            {\\ \forall acc. \res\ltag * acc = \res\rtag\ acc}
\end{array}\]
where $\Psi \defeq n_1 = n_2, \forall y_1 y_2. (y_1,y_2)<(n_1,n_2) \Rightarrow y_1 = y_2 \Rightarrow \forall a. (f_1\ y_1)* a = f_2\ y_2\ a$.

To prove this, we start by applying the one-sided \rname{ABS-R} rule,
with a trivial condition on $acc$.  Then we can apply a two-sided
\rname{CASE} rule, which has 3 premises:
\begin{itemize}
  \item $\Psi \vdash n_1 = 0 \Leftrightarrow n_2 = 0$
  \item $\sjrhol{\Gamma'}
             {\Psi, n_1 = 0, n_2 = 0}
             {1}{\nat\to\nat\to\nat}
             {acc}{\nat\to\nat\to\nat\to\nat}
             {\res\ltag * acc = \res\rtag}$

  \item $\sjrhol{\Gamma'}
             {\Psi}
             {\lambda x_1. (S\ x_1)*(f_1\ x_1)}{\nat\to\nat\to\nat}
             {\lambda x_2. f_2\ x_2\ ( (S\ x_2) * acc)}{\nat\to\nat\to\nat\to\nat}
             {\forall x_1 x_2. n_1 = S\ x_1 \Rightarrow n_2 = S\ x_2 \Rightarrow (\res\ltag\ x_1)* acc = \res\rtag\ x_2}$
\end{itemize}

Premise 1 is a direct consequence of the assertion $n_1=n_2$ in $\Psi$.
Premise 2 is a trivial arithmetic identity which can be proven in HOL
(using rule SUB or by invoking Theorem~\ref{thm:equivhol}).
To prove premise 3, we first apply the (two-sided) \rname{ABS} rule,
which leaves the following proof obligation:
\[\sjrhol{\Gamma'}
         {\Psi, n_1 = S\ x_1, n_2 = S\ x_2}
         {(S\ x_1)*(f_1\ x_1)}{\nat\to\nat\to\nat} {f_2\ x_2\ (
           (S\ x_2) * acc)}{\nat\to\nat\to\nat\to\nat} {\res\ltag *
           acc = \res\rtag}\]
This is proven in HOL by instantiating the inductive hypothesis in
$\Psi$ with $y_1\mapsto x_1, y_2\mapsto x_2, a \mapsto (S\ x_1) * acc$.

\subsection{Second example: take and map}

This example establishes the equivalence of two programs that compute
the same result, but using different number of recursive
calls. Consider the following function $take$ that takes a list $l$
and a natural number $n$ and returns the first $n$ elements of the
list (or the whole list if its length is less than $n$).
%
%
\[\begin{array}{rrl}
     take \;\defeq\; \letrec{f_1}{l_1}{\lambda n_1.~ \caselist{l_1}{[]&&\\}
                                                               {&\lambda h_1 t_1.~\casenat{n_1}{&[]\\&}{&\lambda y_1. \cons{h_1}{(f_1\ t_1\ y_1)}}}}
\end{array}
\]
Next, consider the standard function $map$ that applies a function $g$
to every element of a list $l$ pointwise.
\[\begin{array}{rl}
     map \;\defeq\; \letrec{f_2}{l_2}{\lambda g_2.~ \caselist{l_2}{&[]\\}
                                                               {&\lambda h_2 t_2. \cons{(g_2\ h_2)}{(f_2\ t_2\ g_2)}}}
\end{array}
\]
Intuitively, it should be clear that for all $g, n, l$,
$map\ (take\ l\ n)\ g = take\ (map\ l\ g)\ n$ (mapping $g$ over the
first $n$ elements of the list is the same as mapping over the whole
list and then taking the first $n$ elements). However, the
computations on the two sides of the equality are very different: For
a list $l$ of length more than $n$, $map\ (take\ l\ n)\ g$ only
examines the first $n$ elements, whereas $take\ (map\ l\ g)\ n$
traverses the whole list. In the following we formalize this property
in RHOL (Theorem~\ref{thm:take-map}) and outline the high-level idea
of the proof. The full proof is in the appendix.

\begin{thm}\label{thm:take-map}
$\begin{array}[t]{c}
    \jrhol{l_1,l_2 : \listt{\nat}, n_1,n_2 : \nat, g_1,g_2 : \nat\to\nat}{l_1 = l_2, n_1 = n_2, g_1 = g_2}
          {\\ map\ (take\ l_1\ n_1)\ g_1}{\listt{\nat}}{take\ (map\ l_2\ g_2)\ n_2}{\listt{\nat}}{\res\ltag=\res\rtag}
\end{array}$
\end{thm}
\begin{proof}[Proof idea]
Since the two sides make an unequal number of recursive calls, we need
to reason asynchronously on the two sides (specifically, we use the
rule \rname{LLCASE-A}). However, equality cannot be established
inductively with asynchronous reasoning: If two function applications
are to be shown equal, and a recursion step is taken in only one of
them, then the induction hypothesis cannot be applied. So, we
strengthen the induction hypothesis, replacing the assertion
$\res\ltag = \res\rtag$ in the theorem statement with $ \res\ltag
\sqsubseteq \res\rtag \wedge |\res\ltag|={\sf min}(n_1,|l_1|) \wedge
|\res\rtag|={\sf min}(n_2,|l_2|)$ where $\sqsubseteq$ denotes the
prefix ordering on lists and $|\cdot|$ is the list length
function. This assertion implies $\res\ltag = \res\rtag$ and can be
established inductively. The full proof is in the appendix, but at a
high-level, the proof requires proving two judgments, one for the
inner map-take pair and another for the outer one:
\begin{itemize}
\item
  $\sjrhol{\Gamma}{\Psi}{take\ l_1\ n_1}{\listt{\sigma}}{map\ l_2\ g_2}{\listt{\sigma}}{
  \res\ltag \sqsubseteq_{g_2} \res\rtag }$ 
\item $\sjrhol{\Gamma}{\Psi}{map}{\listt{\sigma}\to\listt{\sigma}}{take}{\listt{\sigma}\to\nat\to\listt{\sigma}}
             {\forall m_1 m_2. m_1 \sqsubseteq_{g_2} m_2 \Rightarrow (\forall g_1. g_1 = g_2 \Rightarrow \forall x_2. x_2 \geq |m_1| \Rightarrow  (\res\ltag\ m_1\ g_1) \sqsubseteq (\res\rtag\ m_2\ x_2))}$
\
\end{itemize}
where $m_1 \sqsubseteq_g m_2$ is an axiomatically defined predicate
equivalent to $({\rm map}\ m_1\ g) \sqsubseteq m_2$ and $\Psi$ are the
assumptions in the statement of the theorem (in particular, $l_1 =
l_2$).  The proof of the first premise proceeds by an analysis of
$map$ using synchronous rules. For the second premise, after applying
\rname{LETREC} we apply the asynchronous \rname{LLCASE-A} rule, and then
prove the following premises:
\begin{enumerate}
\item $\sjrhol{\Gamma}
              {\Psi, \Phi, x_2 \geq |m_1|, g_1 = g_2, m_1 = [], m_2 = []}
              {[]}{.}
              {[]}{.}
              {\res\ltag \sqsubseteq \res\rtag}$
\item $\sjrhol{\Gamma}
              {\Psi, \Phi, x_2 \geq |m_1|, g_1 = g_2, m_1 = []}
              {[]}{.}
              {\lambda h_2 t_2.
                      \casenat{x_2}{[]}
                              {\lambda y_2. \cons{h_2}{f_2\ t_2\ y_2}}}{.}
              {\\ \forall h_2 t_2. m_2 = \cons{h_2}{t_2} \Rightarrow
                  \res\ltag \sqsubseteq (\res\rtag\ h_2\ t_2)}$
\item $\sjrhol{\Gamma}
              {\Psi, \Phi, x_2 \geq |m_1|, g_1 = g_2, m_2 = []}
              {\lambda h_1 t_1. \cons{(g_1\ h_1)}{(f_1\ t_1\ g_1)}}{.}
              {[]}{.}
              {\forall h_1 t_1. m_1 = \cons{h_1}{t_1} \Rightarrow
                  (\res\ltag\ h_1\ t_1) \sqsubseteq \res\rtag}$
\item $\sjrhol{\Gamma}
              {\Psi, \Phi, x_2 \geq |m_1|, g_1 = g_2}
              {\lambda h_1 t_1. \cons{(g_1\ h_1)}{(f_1\ t_1\ g_1)}}{.}
              {\lambda h_2 t_2. \casenat{x_2}{[]}
                                  {\lambda y_2. \cons{h_2}{f_2\ t_2\ y_2}}}{.}
              {\\ \forall h_1 t_1 h_2 t_2. m_1 = \cons{h_1}{t_1} \Rightarrow
                       m_2 = \cons{h_1}{t_1} \Rightarrow
                           (\res\ltag\ h_1\ t_1) \sqsubseteq (\res\rtag\ h_2\ t_2)}$
\end{enumerate}
where $\Phi$ is the inductive hypothesis obtained from the
\rname{LETREC} application. The first two premises follow directly
from the definition of $\sqsubseteq$, while the third one follows from
the contradictory assumptions $m_1 \sqsubseteq_g m_2$, $m_1 =
\cons{h_1}{t_1}$ and $m_2 = \nil$. The last premise is proved by first
applying the \rname{NATCASE-R} rule and then applying the induction
hypothesis.
\end{proof}




\subsection{Third example: insertion sort}

Insertion sort is a standard sorting algorithm that sorts a list
$\cons{h}{t}$ by sorting the tail $t$ recursively and then inserting
$h$ at the appropriate position in the sorted tail. Consider the
following implementations of the insertion function, $\sf insert$, and
the insertion sort function, $\sf isort$, each returning a pair,
whose first element is the usual output list (inserted list for $\sf
insert$ and sorted list for $\sf isort$) and whose second element is
the \emph{number of comparisons} made during the execution (assuming
an eager, call-by-value evaluation strategy).


$\begin{array}{@{}r@{}l@{}l@{}l@{}}
{\sf insert} \defeq \lambda x.\,{\rm letrec}\ insert\ l = {\rm case}\ l\ {\rm of}\  []\; & \mapsto ([x], 0); \\
\_ :: \_\; & \mapsto \lambda h\, t.\, {\rm case}\ & x \leq h\ {\rm of}\\
& & \tbool \mapsto (\cons{x}{l}, 1); \\
& & \fbool \mapsto {\rm let}\ s = insert\ t\ {\rm in}\ (\cons{h}{(\pi_1\ s)}, 1 + (\pi_2\ s))
\end{array}$

$\begin{array}{@{}r@{}l@{}l@{}}
{\sf isort} \defeq {\rm letrec}\ isort\ l = {\rm case}\ l\ {\rm of}\ []\; & \mapsto ([], 0); \\
\_ :: \_ \; & \mapsto \lambda h\, t. \, & {\rm let}\ s = isort \ t \\
& & {\rm let}\ s' = {\sf insert}\ h\ (\pi_1\ s) \ {\rm in}\\
& & (\pi_1 \ s', (\pi_2\ s) + (\pi_2\ s'))
\end{array}$

\noindent 
Using this implementation, we prove the following interesting fact
about insertion sort: Among all lists of the same length, insertion
sort computes the fastest (with fewest comparisons) on lists that are
already sorted. This is a property about the relational cost of
insertion sort (on two different inputs), which cannot be established
in RelCost. To state the property in RHOL, we define a list predicate
${\sf sorted}(l)$ in HOL axiomatically:
\begin{mathpar}
{\sf sorted}([]) \equiv \top
\and
\forall h\, t. \, {\sf sorted}(\cons{h}{t}) \equiv ({\sf sorted}(t) \wedge h \leq {\sf lmin}(t))
\end{mathpar}
where the function ${\sf lmin}(l)$ returns the minimum element of $l$:
\[
{\sf lmin} \defeq {\rm letrec}\ f\ l = \caselist{l}{\infty}{\lambda h\, t.\, {\sf min}(h, f\ t)}
\]  
As in the previous example, let $|\cdot|$ be the standard list length
function. The property of insertion sort mentioned above is formalized
in the following theorem. In words, the theorem says that if ${\sf
  isort}$ is executed on lists $x_1$ and $x_2$ of the same length and
$x_1$ is sorted, then the number of comparisons made during the
sorting of $x_1$ is no more than the number of comparisons made during
the sorting of $x_2$.
\begin{thm}\label{thm:insertion-sort}
  Let $\tau \defeq \listt{\nat} \rightarrow \listt{\nat}$. Then,
  $\jrhol{\bullet}{\bullet}{{\sf isort}}{\tau}{{\sf isort}}{\tau}{\forall
    x_1\, x_2.\, ({\sf sorted}(x_1) \wedge |x_1| = |x_2|) \Rightarrow
    \pi_2(\res\ltag \ x_1) \leq \pi_2(\res\rtag\ x_2)}$.
\end{thm}
A full proof is shown in the appendix. The proof proceeds mostly
synchronously in the two sides. Following the structure of $\sf
isort$, we apply the rules \rname{LETREC}, \rname{LISTCASE} and
\rname{APP} + \rname{ABS} (for the ${\rm let}$ binding, which, as
usual, is defined as a function application), followed by an
application of the inductive hypothesis for the recursive call to
$isort$. Eventually, we expose the call to $\sf insert$ on both
sides. At this point, the observation is that since $x_1$ is already
sorted, its head element must be no greater than all elements in its
tail, so ${\sf insert}$ must return immediately with at most $1$
comparison on the $x_1$ side. Formally, this last proof step can be
completed by switching to either UHOL or HOL and using subject
conversion; we switch to HOL in the appendix.

\section{Conclusion}
We have developed Relational Higher-Order Logic, a new formalism to
reason about relational properties of (pure) higher-order programs
written in a simply typed $\lambda$-calculus with inductive types and
recursive definitions. The system is expressive, has solid foundations
via an equivalence with Higher-Order Logic, and yet retains the (nice)
\lq\lq feel\rq\rq\ of relational refinement type systems. An important
direction for future work is to extend Relational Higher-Order Logic
to effectful programs. Natural directions include integrating the
state monad, and the Giry monad for probability sub-distributions.
One particularly exciting perspective is to broaden the scope of
relational cost analysis to probabilistic programs, and to prove
relational costs for different data structures. There are also many
potential applications to security, differential privacy, machine
learning, and probabilistic programming.

For practical purposes, it will also be interesting to build prototype
implementations of Relational Higher-Order Logic. We believe that much
of the technology developed for (relational) refinement types, and in
particular the automated generation of verification conditions (maybe
with user hints to switch between unary and binary modes of reasoning)
and the connection to SMT-solvers can be lifted without significant
hurdle to Relational Higher-Order Logic.



\bibliographystyle{ACM-Reference-Format}

\bibliography{main}

\newpage
\onecolumn
\appendix

\section{Semantics}

\subsection*{Semantics of HOL}

\subsubsection*{Types}

The interpretation for the types corresponds directly to the usual representation of pairs, lists and functions in set theory.

\begin{align*}
\sem{\bool} &\defeq \{\fbool, \tbool\} \\
\sem{\nat} &\defeq \nat \\
\sem{\listt{\tau}} &\defeq \listt{\sem{\tau}} \\
\sem{\tau_1\times\tau_2} &\defeq \sem{\tau_1}\times\sem{\tau_2}\\
\sem{\tau_1\to\tau_2} &\defeq \sem{\tau_1}\to\sem{\tau_2}
\end{align*}

\subsubsection*{Terms}

The terms are given an interpretation with respect to a valuation $\rho$ which is a partial function mapping variables to elements in the
interpretation of their type. Given $\rho$, we use the notation $\rho\subst{x}{v}$ to denote the unique extension of $\rho$ such
that if $y = x$ then $\rho\subst{x}{v}(y) = v$ and, otherwise, $\rho\subst{x}{v}(y) = \rho(y)$.

\begin{mathpar}
\semexpv{x}{\rho} \defeq \rho(x) \and \semexpv{\pair{t}{u}}{\rho} := \pair{\semexpv{t}{\rho}}{\semexpv{u}{\rho}} \and
\semexpv{\pi_i\ t}{\rho} \defeq \pi_i(\semexpv{t}{\rho}) \and
\semexpv{\lambda x : \tau. t} \defeq \lambda v : \sem{\tau}. \semexpv{x}{\rho\subst{v}{\semexpv{v}{\rho}}} \and
\semexpv{c}{\rho} \defeq c \and \semexpv{S~t}{\rho} \defeq S~\semexpv{t}{\rho} \and
\semexpv{\cons{t}{u}}{\rho} \defeq \cons{\semexpv{t}{\rho}}{\semexpv{u}{\rho}} \and
\semexpv{\caselist{t}{u}{v}}{\rho} \defeq  \caselist{\semexpv{t}{\rho}}{\semexpv{u}{\rho}}{\semexpv{v}{\rho}} \and
\semexpv{\letrec{f}{x}{t}}{\rho} \defeq {\rm fix}\ f\ (\semexpv{\lambda x. t}{\rho})\; \text{ where }\; {\rm fix}\ f\ v := v\subst{f}{({\rm fix}\ f\ v)}
\end{mathpar}

\subsubsection*{Formulas}

We assume that for predicate $P$ of arity $\tau_1\times\dots\times\tau_n$,
we have an interpretation $\llbracket P \rrbracket \in \sem{\tau_1}\times\dots\times\sem{\tau_n}$
that satisfies the axioms for P. The interpretation of a formula is defined as follows:

\begin{align*}
\semexpv{P(t_1,\dots,t_n)}{\rho}\quad &\defeq\quad (\semv{t_1}{\rho},\dots,\semv{t_n}{\rho}) \in \sem{P}\\
\semexpv{\top}{\rho} \quad &\defeq\quad \tilde\top \\
\semexpv{\bot}{\rho} \quad &\defeq\quad \tilde\bot \\
\semexpv{\phi_1 \wedge \phi_2}{\rho} \quad &\defeq\quad \semexpv{\phi_1}{\rho} \:\tilde\wedge\: \semexpv{\phi_2}{\rho, \mathcal{M}} \\
\semexpv{\phi_1 \Rightarrow \phi_2}{\rho} \quad &\defeq\quad \semexpv{\phi_1}{\rho} \:\tilde{\Rightarrow}\: \semexpv{\phi_2}{\rho} \\
\semexpv{\forall x:\tau. \phi}{\rho} \quad &\defeq\quad \tilde{\forall} v. v \in \sem{\tau} \:\tilde{\Rightarrow}\: \semexpv{\phi}{\rho\subst{x}{v}}
\end{align*}

where we use the tilde ($\sim$) to distinguish between the (R)HOL connectives and the meta-level connectives.

\subsubsection*{Soundness} We have the following result:

\begin{theorem}[Soundness of set-theoretical semantics]
If $\jhol{\Gamma}{\Psi}{\phi}$, then for every valuation
$\rho\models\Gamma$, $\bigwedge_{\psi\in\Psi} \semexpv{\psi}{\rho}$
implies $\semexpv{\phi}{\rho}$.
\end{theorem}

\begin{proof}
By induction on the lenght of the derivation of $\jhol{\Gamma}{\Psi}{\phi}$.
\end{proof}

\subsection*{Semantics of UHOL}

The intended meaning of a UHOL judgment $\juhol{\Gamma}{\Psi}{t}{\tau}{\phi}$ is:

\[
\begin{array}{c}
  \text{for all } \mathcal{M}, \rho. \text{ s.t. } \rho \models \Gamma \text{ and } \mathcal{M} \models \Psi,\phi \\
  \semexpv{\bigwedge \Psi}{\rho, \mathcal{M}} \:\text{ implies }\: \semexpv{\phi}{\rho\subst{\res}{\semexpv{t}{\rho}},\mathcal{M}}
\end{array}
\]

We have the following result:

\begin{thm}[Set-theoretical soundness and consistency of UHOL]~
If $\juhol{\Gamma}{\Psi}{t}{\sigma}{\phi}$, then for every valuation
$\rho\models\Gamma$, $\bigwedge_{\psi\in\Psi} \semexpv{\psi}{\rho}$
implies $\semexpv{\phi}{\rho\subst{\res}{\semexpv{t}{\rho}}}$. In
particular, there is no proof of
$\juhol{\Gamma}{\emptyset}{t}{\sigma}{\bot}$ in UHOL.
\end{thm}

\begin{proof}
It is a direct consequence of the embedding from UHOL into HOL and the soundness of HOL.
\end{proof}

\subsection*{Semantics of RHOL}

The intended meaning of a RHOL judgment $\jrhol{\Gamma}{\Psi}{t_1}{\tau_1}{t_2}{\tau_2}{\phi}$ is:

\[
\begin{array}{c}
  \text{for all} \mathcal{M}, \rho. \text{ s.t. } \rho \models \Gamma \text{ and } \mathcal{M} \models \Psi,\phi \\
  \semexpv{\bigwedge \Psi}{\rho, \mathcal{M}} \:\text{ implies }\: \semexpv{\phi}{\rho\subst{\res\ltag}{\semexpv{t_1}{\rho}}\subst{\res\rtag}{\semexpv{t_2}{\rho}},\mathcal{M}}
\end{array}
\]

We have the following result:

\begin{thm}[Set-theoretical soundness and consistency of RHOL]~
If $\jrhol{\Gamma}{\Psi}{t_1}{\sigma_1}{t_2}{\sigma_2}{\phi}$, then
for every valuation $\rho\models\Gamma$, $\bigwedge_{\psi\in\Psi}
\semexpv{\psi}{\rho}$ implies
$\semexpv{\phi}{\rho\subst{\res\ltag}{\semexpv{t_1}{\rho}},
  \subst{\res\rtag}{\semexpv{t_2}{\rho}}}$. In particular, there is no
proof of $\jrhol{\Gamma}{\emptyset}{t_1}{\sigma_1}{
  t_2}{\sigma_2}{\bot}$ for any $\Gamma$.
\end{thm}

\begin{proof}
It is a direct consequence of the embedding of RHOL into HOL and the soundness of HOL.
\end{proof}

\section{Additional rules}

For reasons of space, we have omited some derivable and admissible rules in HOL, UHOL and RHOL. 
These are useful to prove some theorems and examples.
We now discuss the most interesting among them:

\subsection*{HOL}

The following rules are derivable in HOL:

\begin{itemize}

\item A cut rule can be derived from \rname{$\Rightarrow_I$} and \rname{$\Rightarrow_E$}:
\[
\infer[\sf CUT]
  {\jhol{\Gamma}{\Psi}{\phi}}
  {\jhol{\Gamma}{\Psi,\phi'}{\phi} &
   \jhol{\Gamma}{\Psi}{\phi'}}\]

\item A rule for case analysis can be derived from \rname{LIST}:

\[
\infer[\sf DESTR-LIST]
  {\jhol{\Gamma}{\Psi}{\phi}}
  {\jlc{\Gamma}{l}{\listt{\tau}}
    &
   \jhol{\Gamma}{\Psi, l=[]}{\phi}
    &
   \jhol{\Gamma, h:\tau, t:\listt{\tau}}{\Psi, l=\cons{h}{t}}{\phi}}\]\\
  
\item A rule for strong induction can be derived from \rname{LIST}:

\[
\infer[\sf S-LIST]
  {\jhol{\Gamma}{\Psi}{\forall t : \listt{\tau}. \phi}}
  {\jhol{\Gamma}{\Psi}{\phi\subst{t}{\nil} &
   \jhol{\Gamma, h:\tau, t:\listt{\tau}}{\Psi, \forall u : \listt{\tau}. |u| \leq |t| \Rightarrow \phi\subst{t}{u}}{\phi\subst{t}{\cons{h}{t}}}}}
\]

\item A rule for (weak) double induction can be derived by applying \rname{LIST} twice:
\[\infer[\sf D-LIST]
  {\jhol{\Gamma}{\Psi}{\forall l_1 l_2. \phi}}
  {\begin{array}{c}
      \jhol{\Gamma}{\Psi}{\phi\subst{l_1}{\nil}\subst{l_2}{\nil}} \\
      \jhol{\Gamma, h_1:\tau_1, t_1:\listt{\tau_1}}{\Psi, \phi\subst{l_1}{t_1}\subst{l_2}{\nil}}{\phi\subst{l_1}{\cons{h_1}{t_1}}\subst{l_2}{\nil}} \\
      \jhol{\Gamma, h_2:\tau_2, t_2:\listt{\tau_2}}{\Psi, \phi\subst{l_1}{\nil}\subst{l_2}{t_2}}{\phi\subst{l_1}{\nil}\subst{l_2}{\cons{h_2}{t_2}}} \\
      \jhol{\Gamma, h_1:\tau_1, t_2:\listt{\tau_2}, h_2:\tau_2, t_2:\listt{\tau_2}}{\Psi, \phi\subst{l_1}{t_1}\subst{l_2}{t_2}}{\phi\subst{l_1}{\cons{h_1}{t_1}}\subst{l_2}{\cons{h_2}{t_2}}}
   \end{array}}
\]

\item A rule for strong double induction can be derived from \rname{D-LIST}:
\[\infer[\sf S-D-LIST]
  {\jhol{\Gamma}{\Psi}{\forall l_1 l_2. \phi}}
  {\begin{array}{c}
      \jhol{\Gamma}{\Psi}{\phi\subst{l_1}{\nil}\subst{l_2}{\nil}} \\
      \jhol{\Gamma, h_1:\tau_1, t_1:\listt{\tau_1}}{\Psi, \forall m_1. |m_1|\leq|t_1| \Rightarrow \phi\subst{l_1}{m_1}\subst{l_2}{\nil}}{\phi\subst{l_1}{\cons{h_1}{t_1}}\subst{l_2}{\nil}} \\
      \jhol{\Gamma, h_2:\tau_2, t_2:\listt{\tau_2}}{\Psi, \forall m_2. |m_2|\leq|t_2| \Rightarrow \phi\subst{l_1}{\nil}\subst{l_2}{m_2}}{\phi\subst{l_1}{\nil}\subst{l_2}{\cons{h_2}{t_2}}} \\
      \jhol{\Gamma, h_1:\tau_1, t_1:\listt{\tau_1}, h_2:\tau_2, t_2:\listt{\tau_2}}
           {\\ \Psi, \forall m_1 m_2. (|m_1|,|m_2|) < (|\cons{h_1}{t_1}|,|\cons{h_2}{t_2}|) \Rightarrow
                 \phi\subst{l_1}{m_1}\subst{l_2}{m_2}}{\phi\subst{l_1}{\cons{h_1}{t_1}}\subst{l_2}{\cons{h_2}{t_2}}}
   \end{array}}
\]

\end{itemize}

\subsection*{RHOL}

The following version of the case rule is admissible:

\[
\infer[\sf NATCASE*]
  {\jrhol{\Gamma}{\Psi}{\casenat{t_1}{u_1}{v_1}}{\sigma_1}{\casenat{t_2}{u_2}{v_2}}{\sigma_2}{\phi}}
  {\begin{array}{c}
      \jrhol{\Gamma}{\Psi}{t_1}{\listt{\tau_1}}{t_2}{\listt{\tau_2}}{\phi' \wedge (\res\ltag = 0 \Leftrightarrow \res\rtag = 0)} \\
      \jrhol{\Gamma}{\Psi, \phi'\defsubst{0}{0}}{u_1}{\sigma_1}{u_2}{\sigma_2}{\phi} \\
      \jrhol{\Gamma}{\Psi}{v_1}{\nat\to\sigma_1}{v_2}{\nat\to\sigma_2}
            {\forall x_1 x_2. \phi'\defsubst{S x_1}{S x_2} \Rightarrow \phi\defsubst{\res\ltag\ x_1}{\res\rtag\ x_2}}
   \end{array}}
\]
and the one sided version:
\[
\infer[\sf NATCASE*-L]
  {\jrhol{\Gamma}{\Psi}{\casenat{t_1}{u_1}{v_1}}{\sigma_1}{t_2}{\sigma_2}{\phi}}
  {\begin{array}{c}
      \jrhol{\Gamma}{\Psi}{t_1}{\listt{\tau_1}}{t_2}{\sigma_2}{\phi'} \\
      \jrhol{\Gamma}{\Psi, \phi'\defsubst{0}{t_2}}{u_1}{\sigma_1}{t_2}{\sigma_2}{\phi} \\
      \jrhol{\Gamma}{\Psi}{v_1}{\nat\to\sigma_1}{t_2}{\sigma_2}
            {\forall x_1.\phi'\subst{\res\ltag}{S x_1} \Rightarrow \phi\subst{\res\ltag}{\res\ltag\ x_1}}
   \end{array}}
\]

Notice that we can always recover the initial version of the rule by instantiating $\phi'$ as $t_1 = \res\ltag \wedge t_2 = \res\rtag$.\\

\section{Proofs}

\subsection*{Proof of Theorem~\ref{thm:ref-types-sub}}

We will use without proof the following results:

\begin{lem}
  If $\Gamma \vdash \tau \preceq \sigma$ in refinement types, then $\erfor{\tau} \equiv \erfor{\sigma}$.
\end{lem}

\begin{proof}
  By induction on the derivation.
\end{proof}

\begin{lem}
  For every type $\tau$ and expression $e$ and variable $x\not\in FV(\tau,e)$, $\ertype{\tau}{e} = \ertype{\tau}{x}\subst{x}{e}$
\end{lem}

\begin{proof}
  By structural induction.
\end{proof}

\noindent
Now we proceed with the proof of the theorem. We do it by induction on the derivation:\\

\noindent \textbf{Case.} $\inferrule{\Gamma \vdash \tau}{\Gamma \vdash \tau \preceq \tau}$\\
To show: $\erfor{\Gamma}, x:\erfor{\tau} \mid \ertype{\tau}{x} \vdash \ertype{\tau}{x}$. Directly by \rname{AX}.\\

\noindent \textbf{Case.} $\inferrule{\Gamma \vdash \tau_1 \preceq \tau_2 \\ \Gamma \vdash \tau_2 \preceq \tau_3}{\Gamma \vdash \tau_1 \preceq \tau_3}$\\
To show: $\erfor{\Gamma}, x:\erfor{\tau_1} \mid \ertype{\Gamma}{}, \ertype{\tau_1}{x} \vdash \ertype{\tau_3}{x}$.\\
By induction hypothesis on the premises,\\
$\erfor{\Gamma}, x:\erfor{\tau_1} \mid \ertype{\Gamma}{}, \ertype{\tau_1}{x} \vdash \ertype{\tau_2}{x}$\\
and\\
$\erfor{\Gamma}, x:\erfor{\tau_2} \mid \ertype{\Gamma}{}, \ertype{\tau_2}{x} \vdash \ertype{\tau_3}{x}$.\\
We complete the proof by \rname{CUT}. Notice that $\erfor{\tau_1} \equiv \erfor{\tau_2} \equiv \erfor{\tau_3}$.\\

\noindent \textbf{Case.} $\inferrule{\Gamma \vdash \tau_1  \preceq \tau_2}{\Gamma \vdash \listt{\tau_1} \preceq \listt{\tau_2}}$\\
To show: $\erfor{\Gamma}, x:\erfor{\listt{\tau_1}} \mid \ertype{\Gamma}{}, \ertype{\listt{\tau_1}}{x} \vdash \ertype{\listt{\tau_2}}{\res}$\\
Expanding the definitions: $\erfor{\Gamma}, x:\listt{\erfor{\tau_1}} \mid \ertype{\Gamma}{}, \top \vdash \top$,\\
which is trivial.\\

\noindent \textbf{Case.} $\inferrule{\Gamma \vdash \reftype{}{\tau}{\phi}}{\Gamma \vdash \reftype{}{\tau}{\phi} \preceq \tau}$\\
To show: $\erfor{\Gamma}, x:\erfor{\reftype{}{\tau}{\phi}} \mid \ertype{\reftype{}{\tau}{\phi}}{x} \vdash \ertype{\tau}{x}$.\\
Expanding the definitions: $\erfor{\Gamma}, x:\erfor{\reftype{}{\tau}{\phi}} \mid \ertype{\tau}{x} \wedge \emfor{\phi}{x} \vdash \ertype{\tau}{x}$\\
and now the proof is completed trivially by \rname{$\wedge_E$} and \rname{AX}.\\

\noindent \textbf{Case.} $\inferrule{\Gamma \vdash \tau \preceq \sigma \\ \Gamma, \res:\tau \vdash \phi}{\Gamma \vdash \tau \preceq \reftype{}{\sigma}{\phi}}$\\
To show: $\erfor{\Gamma}, \res:\erfor{\tau} \vdash \ertype{\Gamma}{}, \ertype{\tau}{\res} \vdash \ertype{\reftype{}{\sigma}{\phi}}{\res}$\\
Expanding the definition: $\erfor{\Gamma}, \res:\erfor{\tau} \mid \ertype{\Gamma}{}, \ertype{\tau}{\res} \vdash \ertype{\sigma}{\res} \wedge \phi$\\
By induction hypothesis on the premises we have:\\
$\erfor{\Gamma}, \res:\erfor{\tau} \mid \ertype{\Gamma}{}, \ertype{\tau}{\res} \vdash \ertype{\sigma}{\res}$\\
and:\\
$\erfor{\Gamma}, \res:\erfor{\tau} \mid \ertype{\Gamma}{}, \ertype{\tau}{\res} \vdash \phi$\\
We complete the proof by applying the \rname{$\wedge_I$} rule.\\

\noindent \textbf{Case.} $\inferrule{\Gamma \vdash \sigma_2 \preceq \sigma_1 \\ \Gamma, x:\sigma_2 \vdash \tau_1 \preceq \tau_2}
                                    {\Gamma \vdash \Pi (x:\sigma_1).\tau_1 \preceq \Pi (x:\sigma_2).\tau_2}$\\
To show: $\erfor{\Gamma}, f:\erfor{\Pi (x:\sigma_1).\tau_1} \mid \ertype{\Gamma}{},\ertype{\Pi (x:\sigma_1).\tau_1}{f}
             \vdash \ertype{\Pi (x:\sigma_2).\tau_2}{f}$\\
Expanding the definitions:\\
$\erfor{\Gamma}, f:\erfor{\Pi (x:\sigma_1).\tau_1} \mid \ertype{\Gamma}{}, \forall x. \ertype{\sigma_1}{x} \Rightarrow \ertype{\tau_1}{f x}
     \vdash \forall x. \ertype{\sigma_2}{x} \Rightarrow \ertype{\tau_2}{f x}$\\
By the rules \rname{$\forall_I$} and \rname{$\Rightarrow_I$} it suffices to prove:\\
$\erfor{\Gamma}, f:\erfor{\Pi (x:\sigma_1).\tau_1}, x:\erfor{\sigma_2} \mid \ertype{\Gamma}{},\forall x. \ertype{\sigma_1}{x} \Rightarrow \ertype{\tau_1}{f x}, \ertype{\sigma_2}{x}
     \vdash \ertype{\tau_2}{f x}$\hfill{}(1)\\
On the other hand, by induction hypothesis on the premises:\\
$\erfor{\Gamma}, x:\erfor{\sigma_2} \mid \ertype{\Gamma}{},\ertype{\sigma_2}{x} \vdash \ertype{\sigma_1}{x}$\hfill{}(2)\\
and\\
$\erfor{\Gamma}, x:\erfor{\sigma_2}, y:\erfor{\tau_1} \mid \ertype{\Gamma}{},\ertype{\sigma_2}{x}, \ertype{\tau_1}{y} \vdash \ertype{\tau_2}{y}$\hfill{}(3)\\
which we can weaken respectively to:\\
$\erfor{\Gamma}, x:\erfor{\sigma_2}, f:\erfor{\Pi (x:\sigma_1).\tau_1} \mid \erfor{\Gamma}{},\ertype{\sigma_2}{x}, \forall x. \ertype{\sigma_1}{x} \Rightarrow \ertype{\tau_1}{f x} \vdash \ertype{\sigma_1}{x}$\hfill{}(4)\\
and\\
$\erfor{\Gamma}, x:\erfor{\sigma_2},  y:\erfor{\tau_1}, f:\erfor{\Pi (x:\sigma_1).\tau_1} \mid \erfor{\Gamma}{},\ertype{\sigma_2}{x}, \ertype{\tau_1}{y}, \forall x. \ertype{\sigma_1}{x} \Rightarrow \ertype{\tau_1}{f x}
  \vdash \ertype{\tau_2}{y}$\hfill{}(5)\\
From (4), by doing a cut with its own premise $\forall x. \ertype{\sigma_1}{x} \Rightarrow \ertype{\tau_1}{f x}$, we derive:\\
$\erfor{\Gamma}, x:\erfor{\sigma_2}, f:\erfor{\Pi (x:\sigma_1).\tau_1} \mid \ertype{\Gamma}{},\ertype{\sigma_2}{x}, \forall x. \ertype{\sigma_1}{x} \Rightarrow \ertype{\tau_1}{f x} \vdash
\ertype{\tau_1}{f x}$\hfill{}(6)\\
From (5), by \rname{$\Rightarrow_I$} and \rname{$\forall_I$} we can derive:\\
$\erfor{\Gamma}, x:\erfor{\sigma_2}, f:\erfor{\Pi (x:\sigma_1).\tau_1} \mid \ertype{\Gamma}{},\ertype{\sigma_2}{x}, , \forall x. \ertype{\sigma_1}{x} \Rightarrow \ertype{\tau_1}{f x}
  \vdash \forall y. \ertype{\tau_1}{y} \Rightarrow \ertype{\tau_2}{y}$\hfill{}\\
And by \rname{$\forall_E$} \\
$\erfor{\Gamma}, x:\erfor{\sigma_2}, f:\erfor{\Pi (x:\sigma_1).\tau_1} \mid \ertype{\Gamma}{},\ertype{\sigma_2}{x}, , \forall x. \ertype{\sigma_1}{x} \Rightarrow \ertype{\tau_1}{f x}
  \vdash \ertype{\tau_1}{fx} \Rightarrow \ertype{\tau_2}{fx}$\hfill{}(7)\\
Finally, from (6) and (7) by \rname{$\Rightarrow_E$} we get:\\
$\erfor{\Gamma}, x:\erfor{\sigma_2}, f:\erfor{\Pi (x:\sigma_1).\tau_1} \mid \ertype{\Gamma}{},\ertype{\sigma_2}{x}, \forall x. \ertype{\sigma_1}{x} \Rightarrow \ertype{\tau_1}{f x} \vdash \ertype{\tau_2}{f x}$\\
and by one last application of \rname{$\Rightarrow_I$} we get what we wanted to prove.

\subsection*{Proof of Theorem~\ref{thm:ref-types-typ}}

By induction on the derivation:\\

\noindent \textbf{Case.} $\inferrule{}{x: \tau, \Gamma \vdash x : \tau}$\\
To prove : $x: \erfor{\tau} , \erfor{\Gamma} \vdash \ertype{\tau}{x}, \ertype{\Gamma}{} \vdash x : \erfor{\tau} \mid \ertype{\tau}{\res}$.
Directly by \rname{VAR}.\\

\noindent \textbf{Case.} $\inferrule{\Gamma, x:\tau \vdash t : \sigma}{\Gamma \vdash \lambda x. t : \Pi(x:\tau).\sigma}$\\
To prove: $\emtyj{\Gamma}{\lambda x.t}{\Pi(x:\tau).\sigma}$.\\
Expanding the definitions:\\
$\erfor{\Gamma} \mid \ertype{\Gamma}{} \vdash \lambda x. \emterm{t}{} : \erfor{\tau}\to\erfor{\sigma} \mid \forall x. \ertype{\tau}{x} \Rightarrow \ertype{\sigma}{\res x}$\\
By induction hypothesis on the premise:\\
$\erfor{\Gamma}, x:\erfor{\tau} \mid \ertype{\Gamma}{}, \ertype{\tau}{x} \vdash \emterm{t}{} : \erfor{\sigma} \mid \ertype{\sigma}{\res}$\\
Directly by \rname{ABS}.\\

\noindent \textbf{Case.} $\inferrule{\jreft{\Gamma}{t}{\pitype{x}{\tau}{\sigma}} \\ \jreft{\Gamma}{u}{\tau}}{\jreft{\Gamma}{t\ u}{\sigma\subst{x}{u}}}$\\
To prove: $\emtyj{\Gamma}{t\ u}{\sigma\subst{x}{u}}$.\\
Expanding the definitions:\\
$\erfor{\Gamma} \mid \ertype{\Gamma}{} \vdash \emterm{t}{}\ \emterm{e_2}{} : \erfor{\sigma} \mid \ertype{\sigma}{\res}\subst{x}{u}$\\
By induction hypothesis on the premise:\\
$\erfor{\Gamma} \mid \ertype{\Gamma}{} \vdash \emterm{t}{} : \erfor{\tau}\to\erfor{\sigma} \mid \forall x. \ertype{\tau}{x} \Rightarrow \ertype{\sigma}{\res x}$\\
and\\
$\erfor{\Gamma} \mid \ertype{\Gamma}{} \vdash \emterm{u}{} : \erfor{\tau} \mid \ertype{\tau}{\res}$\\
We get the result directly by \rname{APP}.\\

\noindent \textbf{Case.} $\inferrule{\jreft{\Gamma}{t}{\listt{\tau}} \\ \jreft{\Gamma}{u}{\sigma} \\ \jreft{\Gamma}{v}{\tau\to\listt{\tau}\to\sigma}}
                                    {\jreft{\Gamma}{{\rm case}\ t\ {\rm of}\ \nil \mapsto u ; \cons{\_}{\_} \mapsto v}{\sigma}}$\\
To prove: $\emtyj{\Gamma}{{\rm case}\ t\ {\rm of}\ \nil \mapsto u ; \cons{\_}{\_} \mapsto v}{\sigma}$\\
By induction hypothesis on the premises:\\
$\erfor{\Gamma} \mid \ertype{\Gamma}{} \vdash \emterm{t}{} : \erfor{\listt{\tau}} \mid \ertype{\listt{\tau}}{\res}$,\hfill{}(1)\\
$\erfor{\Gamma} \mid \ertype{\Gamma}{} \vdash \emterm{u}{} : \erfor{\sigma} \mid \ertype{\sigma}{\res}$,\hfill{}(2)\\
and \\
$\erfor{\Gamma} \mid \ertype{\Gamma}{} \vdash
    \emterm{v}{} : \erfor{\tau\to\listt{\tau}\to\sigma} \mid \ertype{\tau\to\listt{\tau}\to\sigma}{\res}$\hfill{}(3)\\
Expanding the definitions on (3) we get:\\
$\erfor{\Gamma} \mid \ertype{\Gamma}{} \vdash
    \emterm{v}{} : \erfor{\tau}\to\erfor{\listt{\tau}}\to\erfor{\sigma} \mid
       \forall x. \ertype{\tau}{x} \Rightarrow \forall y. \ertype{\listt{\tau}}{y} \Rightarrow \ertype{\sigma}{\res\ x\ y} $\hfill{}(4)\\
And from (1), (2) and (4) we apply \rname{LISTCASE*} and we get the result. Notice that (2) and (4) are stronger than the premises of the rule,
so we will first need to apply \rname{SUB} to weaken them\\

\noindent \textbf{Case.} $\inferrule{\Gamma \vdash \tau}{\Gamma \vdash \nil : \listt{\tau}}$ \\
To prove: $\emtyj{\Gamma}{\nil}{\listt{\tau}}$\\
Expanding the definitions: $\erfor{\Gamma} \mid \ertype{\Gamma}{} \vdash \nil : \listt{\erfor{\tau}} \mid {\rm All}(\res, x, \ertype{\tau}{x})$\\
And by the definition of ${\rm All}$ for the empty case, trivially ${\rm All}(\nil, x, \ertype{\tau}{x})$, so we apply the
\rname{NIL} rule and we get the result.\\ 

\noindent \textbf{Case.} $\inferrule{\Gamma \vdash h : \tau \\ \Gamma \vdash t : \listt{\tau}}{\Gamma \vdash \cons{h}{t} : \listt{\tau}}$ \\
To prove: $\emtyj{\Gamma}{\cons{h}{t}}{\listt{\tau}}$. \\
Expanding the definitions: $\erfor{\Gamma} \mid \ertype{\Gamma}{} \vdash \cons{h}{t} : \listt{\erfor{\tau}} \mid {\rm All}(\res, \lambda x. \ertype{\tau}{x})$.\\
By induction hypothesis on the premises, we have:\\
$\erfor{\Gamma} \mid \ertype{\Gamma}{} \vdash h : \erfor{\tau} \mid \ertype{\tau}{\res}$\\
and\\
$\erfor{\Gamma} \mid \ertype{\Gamma}{} \vdash t : \listt{\erfor{\tau}} \mid {\rm All}(\res, \lambda x. \ertype{\tau}{x})$.\\
We complete the proof by the \rname{CONS} rule and the definition of ${\rm All}$ in the inductive case.\\

\noindent \textbf{Case.} $\inferrule{\Gamma \vdash \tau \preceq \sigma \\ \jreft{\Gamma}{t}{\tau}}{\jreft{\Gamma}{t}{\sigma}}$\\
To prove: $\emtyj{\Gamma}{t}{\sigma}$\\
and, since $\erfor{\sigma} \equiv \erfor{\tau}$, it is the same as writing\\
$\emtyj{\Gamma}{t}{\tau}$\\
By induction hypothesis on the premises:\\
$\erfor{\Gamma}, x:\erfor{\tau} \mid \ertype{\Gamma}{}, \ertype{\tau}{x} \vdash \ertype{\sigma}{x}$\\
and\\
$\emtyj{\Gamma}{t}{\tau}$\\
The proof is completed by applying \rname{$\Rightarrow_I$} to the first premise, and then \rname{SUB}.\\

\noindent \textbf{Case.} $\inferrule{\jreft{\Gamma, x: \tau, f:\pitype{y}{\reftype{}{\tau}{y<x}}{\sigma\subst{x}{y}}}{t}{\sigma} \\ \wdef{f}{x}{t}}
                                    {\jreft{\Gamma}{\letrec{f}{x}{t}}{\pitype{x}{\tau}{\sigma}}}$\\
To prove: $\emtyj{\Gamma}{\letrec{f}{x}{t}}{\pitype{x}{\tau}{\sigma}}$\\
By induction hypothesis on the premise:\\
$\erfor{\Gamma},x:\erfor{\tau}, f:\erfor{\tau}\to\erfor{\sigma} \mid \ertype{\Gamma}{}, \ertype{\tau}{x}, \forall y. \ertype{\tau}{y} \wedge y < x \Rightarrow \ertype{\sigma\subst{x}{y}}{f y}
    \vdash \emterm{t}{} : \erfor{\sigma} \mid \ertype{\sigma}{\res}$\\
Directly by \rname{LETREC}.

\subsection*{Proof of Theorem~\ref{thm:emb-rrt}}

We can recover the lemma from the unary case:

\begin{lem}
  For every type $\tau$, expressions $t_1,t_2$ and variables $x_1,x_2\not\in FV(\tau,t_1,t_2)$, $\errtype{\tau}{t_1,t_2} = \errtype{\tau}{x_1,x_2}\subst{x_1}{t_1}\subst{x_2}{t_2}$
\end{lem}

Most cases are very similar to the unary case, so we will only show the most
interesting ones:

\noindent \textbf{Case.} $\inferrule{\Gamma \vdash T}
                                    {\Gamma \vdash \nil \sim \nil :: \listt{T}}$ \\
To show: $\jrhol{\erfor{\Gamma}}{\errtype{\Gamma}{}}{\nil}{\erfor{\listt{T}}}{\nil}{\erfor{\listt{T}}}{\errtype{\listt{T}}{\res\ltag,\res\rtag}}$.\\
There are two options. If $T$ is a unary type, we have to prove:\\
$\jrhol{\erfor{\Gamma}}{\errtype{\Gamma}{}}{\nil}{\erfor{\listt{T}}}{\nil}{\erfor{\listt{T}}}
       {\bigwedge_{i\in\{1,2\}} {\rm All}(\res_i, \lambda x.\ertype{\tau}{x})}$\\
And by the definition of ${\rm All}$ we can directly prove:\\
$\jhol{\emptyset}{\emptyset}{{\rm All}(\nil, \lambda x.\ertype{\tau}{x}) \wedge {\rm All}(\nil, \lambda x.\ertype{\tau}{x})}$\\
If $T$ is a relational type, we have to prove:\\
$\jrhol{\erfor{\Gamma}}{\errtype{\Gamma}{}}{\nil}{\erfor{\listt{T}}}{\nil}{\erfor{\listt{T}}}
       {{\rm All2}(\res\ltag, \res\rtag, \lambda x_1.\lambda x_2.\errtype{T}{x_1, x_2})}$\\
And by the definition of ${\rm All2}$ we can directly prove:\\
$\jhol{\emptyset}{\emptyset}{{\rm All2}(\nil, \nil, \lambda x_1.\lambda x_2.\errtype{T}{x_1,x_2})}$\\

\noindent \textbf{Case.}$\inferrule{\Gamma \vdash h_1 \sim h_2 :: T \\ \Gamma \vdash t_1 \sim t_2 :: \listt{T}}
                                   {\Gamma \vdash \cons{h_1}{t_1} \sim \cons{h_2}{t_2} :: \listt{T}}$\\
To show: $\jrhol{\erfor{\Gamma}}{\errtype{\Gamma}{}}{\cons{h_1}{t_2}}{\erfor{\listt{T}}}{\cons{h_2}{t_2}}{\erfor{\listt{T}}}{\listt{T}}$.\\
There are two options. If $T$ is a unary type, we have to prove:\\
$\jrhol{\erfor{\Gamma}}{\errtype{\Gamma}{}}{\cons{h_1}{t_1}}{\erfor{\listt{T}}}{\cons{h_2}{t_2}}{\erfor{\listt{T}}}
       {\bigwedge_{i\in\{1,2\}} {\rm All}(\res_i, \lambda x.\ertype{T}{x})}$\\
By induction hypothesis we have:\\
$\jrhol{\erfor{\Gamma}}{\errtype{\Gamma}{}}{h_1}{\erfor{T}}{\cons{h_2}{t_2}}{\erfor{T}}
       {\bigwedge_{i\in\{1,2\}} \ertype{T}{\res_i}}$\\
and\\
$\jrhol{\erfor{\Gamma}}{\errtype{\Gamma}{}}{t_1}{\erfor{\listt{T}}}{t_2}{\erfor{\listt{T}}}
       {\bigwedge_{i\in\{1,2\}} {\rm All}(\res_i, \lambda x.\ertype{T}{x})}$\\
And by the definition of ${\rm All}$ we can directly prove:\\
$\bigwedge_{i\in\{1,2\}} \ertype{T}{h_i} \Rightarrow
    {\bigwedge_{i\in\{1,2\}} {\rm All}(t_i, \lambda x.\ertype{T}{x})} \Rightarrow
       {\bigwedge_{i\in\{1,2\}} {\rm All}(\cons{h_i}{t_i}, \lambda x.\ertype{T}{x})}$\\
So by the \rname{CONS} rule, we prove the result.
If $T$ is a relational type, we have to prove:\\
$\jrhol{\erfor{\Gamma}}{\errtype{\Gamma}{}}{\cons{h_1}{t_1}}{\erfor{\listt{T}}}{\cons{h_2}{t_2}}{\erfor{\listt{T}}}
       {{\rm All2}(\res\ltag, \res\rtag, \lambda x_1. \lambda x_2.\errtype{T}{x_1, x_2})}$\\
By induction hypothesis we have:\\
$\jrhol{\erfor{\Gamma}}{\errtype{\Gamma}{}}{h_1}{\erfor{T}}{\cons{h_2}{t_2}}{\erfor{T}}
       {\errtype{T}{\res\ltag, \res\rtag}}$\\
and\\
$\jrhol{\erfor{\Gamma}}{\errtype{\Gamma}{}}{t_1}{\erfor{\listt{T}}}{t_2}{\erfor{\listt{T}}}
       {{\rm All2}(\res\ltag, \res\rtag, \lambda x_1.\lambda x_2.\errtype{T}{x_1,x_2})}$\\
And by the definition of ${\rm All2}$ we can directly prove:\\
$\errtype{T}{h_1, h_2} \Rightarrow
    {{\rm All2}(t_1, t_2, \lambda x_1.\lambda x_2.\errtype{T}{x_1, x_2})} \Rightarrow
       {\rm All}(\cons{h_1}{t_1}, \cons{h_1}{h_2}, \lambda x_1.\lambda x_2.\errtype{T}{x_1,x_2})$\\
So by the \rname{CONS} rule, we prove the result.\\

\noindent \textbf{Case.}$\inferrule{\Gamma \vdash t_1 \sim t_2 :: \listt{T} \\ \Gamma \vdash t_1 = \nil \Leftrightarrow t_2 = \nil \\
                                    \Gamma \vdash u_1 \sim u_2 :: U \\ \Gamma \vdash v_1 \sim v_2 :: \rtprod{h}{T}{\rtprod{t}{\listt{T}}{U}}}
                                   {\Gamma \vdash {\rm case}\ t_1\ {\rm of}\ \nil \mapsto u_1 ; \cons{\_}{\_} \mapsto v_1 \sim
                                                    {\rm case}\ t_2\ {\rm of}\ \nil \mapsto u_2 ; \cons{\_}{\_} \mapsto v_2 :: U}$\\
To show:\\
$\jrhol{\erfor{\Gamma}}{\errtype{\Gamma}{}}{\caselist{t_1}{u_1}{v_1}}{\erfor{U}}{\caselist{t_2}{u_2}{r_2}}{\erfor{U}}{\errtype{U}{\res\ltag,\res\rtag}}$\\
By induction hypothesis we have:\\
$\jhol{\erfor{\Gamma}}{\errtype{\Gamma}{}}{t_1 = \nil \Leftrightarrow t_2 = \nil}$,\\
$\jrhol{\erfor{\Gamma}}{\errtype{\Gamma}{}}{u_1}{\erfor{U}}{u_2}{\erfor{U}}{\errtype{U}{\res\ltag,\res\rtag}}$ \\
and\\
$\jrhol{\erfor{\Gamma}}{\errtype{\Gamma}{}}{v_1}
       {T\to\listt{T}\to U}{v_2}{T\to\listt{T}\to U}
       {\forall h_1 h_2. \errtype{T}{h_1,h_2} \Rightarrow \forall t_1 t_2. \errtype{\listt{T}}{t_1,t_2} \Rightarrow \errtype{U}{\res\ltag h_1 t_1,\ h_2 t_2 \res\rtag}}$ \\
By applying the \rname{LISTCASE*} rule to the three premises we get the result.\\

\noindent \textbf{Case.}$\inferrule{ \Gamma, x :: T, f :: \rtprod{y}{\rtref{y}{T}{(y_1,y_2) < (x_1,x_2)}}{U\subst{x}{y}} \vdash t_1 \sim t_2 :: U\\\\
                                      \Gamma \vdash \rtprod{x}{T}{U} \\ \wdef{f_1}{x_1}{t_1} \\ \wdef{f_2}{x_2}{t_2}}
                                   {\Gamma \vdash {\rm letrec}\ f_1\ x_1 = t_1 \sim {\rm letrec}\ f_2\ x_2 = t_2  :: \rtprod{x}{T}{U}}$\\
To show:\\
$\jrhol{\erfor{\Gamma}}{\errtype{\Gamma}{}}{\letrec{f_1}{x_1}{t_1}}{\erfor{\rtprod{x}{T}{U}}}
       {\letrec{f_2}{x_2}{t_2}}{\erfor{\rtprod{x}{T}{U}}}{\errtype{\rtprod{x}{T}{U}}{\res\ltag,\res\rtag}}$\\
Expanding the definitions:\\
$\jrhol{\erfor{\Gamma}}{\errtype{\Gamma}{}}{\letrec{f_1}{x_1}{t_1}}{\erfor{T}\to\erfor{U}}
       {\letrec{f_2}{x_2}{t_2}}{\erfor{T}\to\erfor{U}}{\forall x_1 x_2. \errtype{T}{x_1,x_2} \Rightarrow \errtype{U}{\res\ltag x_1,\ \res\rtag x_2}}$\\
By induction hypothesis on the premise:\\
$\jrhol{\erfor{\Gamma}, x_1,x_2 : \erfor{T}, f_1,f_2 : \erfor{T}\to\erfor{U}}
       {\errtype{\Gamma}{}, \errtype{T}{x_1,x_2}, \forall y_1,y_2. (\errtype{T}{y_1,y_2} \wedge (y_1,y_2) < (x_1, x_2)) \Rightarrow \errtype{U}{f_1 x_1,\ f_2 x_2}}
       {t_1}{\erfor{U}}{t_2}{\erfor{U}}{\errtype{U}{\res\ltag, \res\rtag}}$\\
And we apply the \rname{LETREC} rule to get the result.

\subsection*{Proof of Theorem~\ref{thm:equivhol}}

The easier direction is the reverse implication. To prove it, one just
notices that we can trivially apply \rname{SUB} instantiating $\phi'$
as a tautology that matches the structure of the types. For instance,
for a base type $\nat$ we would use $\top$, for an arrow type
$\nat\to\nat$ we would use $\forall x. \bot \Rightarrow \top$, and so
on.

We now prove the direct implication by induction on the derivation of
$\jrhol{\Gamma}{\Psi}{t_1}{\sigma_1}{t_2}{\sigma_2}{\phi}$. Suppose
the last rule is:

\noindent \textbf{Case.} \rname{VAR} (similarly, \rname{NIL} and \rname{PROJ})

\noindent
The premise of the rule is already the judgment we want to prove.\\

\noindent \textbf{Case.} \rname{ABS} $ \inferrule
      {\jrhol{\Gamma,x_1:\tau_1,x_2:\tau_2}{\Psi,\phi'}{t_1}{\sigma_1}{t_2}{\sigma_2}{\phi}}
      {\jrhol{\Gamma}{\Psi}{\lambda x_1. t_1}{\tau_1 \to \sigma_1}{\lambda x_2. t_2}{\tau_2\to \sigma_2}
             {\forall x_1,x_2. \phi' \Rightarrow \phi\subst{\res\ltag}{\res\ltag\ x_1}\subst{\res\rtag}{\res\rtag\ x_2}}}$

\noindent
By applying the induction hypothesis on the premise:
\\
$\jhol{\Gamma, x_1:\tau_{1}, x_2:\tau_{2}}{\Psi, \phi'}{\phi\defsubst{t_1}{t_2}}$ \hfill{} (1)
\\
By applying \rname{$\Rightarrow_I$} on (1):
\\
$\jhol{\Gamma, x_1:\tau_{1}, x_2:\tau_{2}}{\Psi}{\phi' \Rightarrow \phi\defsubst{t_1}{t_2}}$
\\
By applying \rname{$\forall_I$} twice on (2):
\\
$\jhol{\Gamma}{\Psi}{\forall x_1 x_2. \phi' \Rightarrow \phi\defsubst{t_1}{t_2}}$ \hfill{} (3)
\\
Finally, by applying CONV on (3):
\\
$\jhol{\Gamma}{\Psi}{\forall x_1 x_2. \phi' \Rightarrow \phi\defsubst{(\lambda x_1. t_1)\ x_1}{(\lambda x_2. t_2)\ x_2}}$
\\
Proof for \rname{ABS-L} (and \rname{ABS-R}) is analogous. \\

\noindent \textbf{Case.} \rname{APP} $\inferrule
{\jrhol{\Gamma}{\Psi}{t_1}{\tau_1\to \sigma_1}{t_2}{\tau_2\to \sigma_2}
               {\forall x_1,x_2. \phi'\subst{\res\ltag}{x_1}\subst{\res\rtag}{x_2}\Rightarrow \phi\subst{\res\ltag}{\res\ltag\ x_1}\subst{\res\rtag}{\res\rtag\ x_2}} \\
         \jrhol{\Gamma}{\Psi}{u_1}{\tau_1}{u_2}{\tau_2}{\phi'}}
{\jrhol{\Gamma}{\Psi}{t_1 u_1}{\sigma_1}{t_2 u_2}{\sigma_2}{\phi\subst{x_1}{u_1}\subst{x_2}{u_2}}}$\\
 
\noindent
By applying the induction hypothesis on the premises we have:
\\
$\jhol{\Gamma}{\Psi}{\forall x_1 x_2. \phi'\defsubst{x_1}{x_2}\Rightarrow 
          \phi\defsubst{t_1\ x_1}{t_2\ x_2}}$ \hfill{} (1)
\\
and
\\
$\jhol{\Gamma}{\Psi}{\phi'\defsubst{u_1}{u_2}}$ \hfill{} (2)
\\
By applying twice \rname{$\forall_E$} to (1) with $u_1,u_2$:
\\
$\jhol{\Gamma}{\Psi}{\phi'\defsubst{u_1}{u_2}\Rightarrow 
          \phi\defsubst{t_1\ u_1}{t_2\ u_2}}$ \hfill{} (3)
\\
and by applying \rname{$\Rightarrow_E$} to (3) and (2):
\\
$\jhol{\Gamma}{\Psi}{\phi\defsubst{t_1\ u_1}{t_2\ u_2}}$
\\
Proof for \rname{APP-L} (and APP-R) is analogous, and it uses the UHOL embedding for the
premise about the argument.
Proofs for \rname{CONS} and \rname{PAIR} are very similar as well.\\

\noindent \textbf{Case.} \rname{SUB}
      $\inferrule
      {\jrhol{\Gamma}{\Psi}{t_1}{\sigma_1}{t_2}{\sigma_2}{\phi'} \\
       \Gamma\mid\Psi\vdash_{\sf HOL}\phi'\defsubst{t_1}{t_2} \Rightarrow \phi\defsubst{t_1}{t_2}}
      {\jrhol{\Gamma}{\Psi}{t_1}{\sigma_1}{t_2}{\sigma_2}{\phi}}$\\

\noindent
Applying the inductive hypothesis on the premises we have:\\
$\jhol{\Gamma}{\Psi}{\phi'\defsubst{t_1}{t_2}}$\\
and\\
$\jhol{\Gamma}{\Psi}{\phi'\defsubst{t_1}{t_2} \Rightarrow \phi\defsubst{t_1}{t_2}}$\\
The proof is simply applying \rname{$\Rightarrow_E$}.\\


\noindent \textbf{Case.} \rname{LETREC} 
 $\inferrule
     {\wdef{f_1}{x_1}{e_1} \:\: \wdef{f_2}{x_2}{e_2} \\
      \begin{array}{c}
        \jrhol{\Gamma, x_1:I_1, x_2:I_2, f_1:I_1\to{\sigma}, f_2:I_2\to{\sigma_2}}
              {\\ \Psi, \phi', \forall m_1 m_2. (|m_1|, |m_2|) < (|x_1|, |x_2|) \Rightarrow \phi'\subst{x_1}{m_1}\subst{x_2}{m_2}\Rightarrow
                  \phi\subst{x_1}{m_1}\subst{x_2}{m_2}\defsubst{f_1\ m_1}{f_2\ m_2}}{\\ e_1}{\sigma_1}{e_2}{\sigma_2}{\phi}
      \end{array}
     }
     {\jrhol{\Gamma}{\Psi}{{\rm letrec}\ f_1\ x_1\ = e_1}{I_1\to{\sigma_2}}{{\rm letrec}\ f_2\ x_2\ = e_2}{I_2\to{\sigma_2}}
            {\forall x_1 x_2. \phi' \Rightarrow \phi\defsubst{\res\ltag\ x_1}{\res\rtag\ x_2}}}$\\

\noindent
As an example, we prove the list and nat case, but for other datatypes the proof
is similar.
Applying the inductive hypothesis on the premise we have:

\[\jhol{\Gamma, l_1, n_2, f_1, f_2}{\Psi, \forall m_1 m_2. (|m_1|,|m_2|) < (|l_1|, |n_2|) \Rightarrow \phi\defsubst{f_1 m_1}{f_2 m_2}}{\phi\defsubst{e_1}{e_2}}\]
By \rname{$\forall_I$} we derive:
\[\jhol{\Gamma}{\Psi}{\forall f_1, f_2, l_1, n_2 .(\forall m_1 m_2. (|m_1|,|m_2|) < (|l_1|, |n_2|) \Rightarrow \phi\defsubst{f_1 m_1}{f_2 m_2}) \Rightarrow \phi\defsubst{e_1}{e_2}}.
  \tag{$\Phi$} \]

We want to prove
\[\jhol{\Gamma}{\Psi}{\forall l_1 n_2. \phi\defsubst{F_1\ l_1}{F_2\ n_2}}\]
where we use the abbreviations
\begin{align*}
F_1 \quad&:=\quad \letrec{f_1}{x_1}{e_1} \\
F_2 \quad&:=\quad \letrec{f_2}{x_2}{e_2}
\end{align*}

We will use strong double induction over natural numbers and lists. We need to prove four premises. Since we can prove $(\Phi)$ from $\Gamma, \Psi$, we can add it to the axioms:

\begin{enumerate}[label=(\Alph*)]
  \item $\jhol{\Gamma}{\Psi, \Phi}{\phi\defsubst{F_1\ \nil}{F_2\ 0}}$
  \item $\jhol{\Gamma, h_1, t_1}{\Psi, \Phi, \forall m_1. |m_1| \leq |t_1| \Rightarrow \phi\defsubst{F_1\ m_1}{F_2\ 0}}{\phi\defsubst{F_1\ (\cons{h_1}{t_1})}{F_2\ 0}}$
  \item $\jhol{\Gamma, x_2}{\Psi, \Phi, \forall m_2. |m_2| \leq |x_2| \Rightarrow \phi\defsubst{F_1\ \nil}{F_2\ m_2}}{\phi\defsubst{F_1\ \nil}{F_2\ (S x_2)}}$
  \item $\jhol{\Gamma, h_1, t_1, x_2}{\Psi, \Phi, \forall m_1 m_2. (|m_1|, |m_2|)<(|\cons{h_1}{t_1}|, |S x_2|) \Rightarrow \\
         \phi\defsubst{F_1\ m_1}{F_2\ m_2}}{\phi\defsubst{F_1\ (\cons{h_1}{t_1})}{F_2\ (S x_2)}}$
\end{enumerate}

To prove them, we will instantiate the quantifiers in $\Phi$ with the appropriate variables.

To prove (A), we instantiate $\Phi$ at $F_1, F_2, \nil,0$:
\[(\forall m_1 m_2. (|m_1|,|m_2|) < (|[]|, |0|) \Rightarrow \phi\defsubst{F_1 m_1}{F_2 m_2}) \Rightarrow \phi\defsubst{e_1}{e_2}\subst{l_1}{\nil}\subst{n_2}{0}\subst{f_1}{F_1}\subst{f_2}{F_2}\] 
and, since $(|m_1|,|m_2|) < (|[]|, |0|)$ is trivially false, then
\[\phi\defsubst{e_1}{e_2}\subst{l_1}{\nil}\subst{n_2}{0}\subst{f_1}{F_1}\subst{f_2}{F_2}\]
and by beta-expansion and \rname{CONV}:
\[\phi\defsubst{F_1\ \nil}{F_2\ 0}\].\\

To prove (B), we instantiate $\Phi$ at $F_1, F_2, \cons{h_1}{t_1}, 0$
\[(\forall m_1 m_2. (|m_1|,|m_2|) < (|\cons{h_1}{t_1}|, |0|) \Rightarrow \phi\defsubst{F_1 m_1}{F_2 m_2}) \Rightarrow
         \phi\defsubst{e_1}{e_2}\subst{l_1}{\cons{h_1}{t_1}}\subst{n_2}{0}\subst{f_1}{F_1}\subst{f_2}{F_2}\]
by beta-expansion:
\[(\forall m_1 m_2. (|m_1|,|m_2|) < (|\cons{h_1}{t_1}|, |0|) \Rightarrow \phi\defsubst{F_1 m_1}{F_2 m_2}) \Rightarrow
         \phi\defsubst{F_1\ \cons{h_1}{t_1}}{F_2\ 0}\]
Since $(|m_1|,|m_2|) < (|\cons{h_1}{t_1}|, |0|)$ is only satisfied if $|m_1|\leq|t_1| \wedge m_2 = 0$, we can write it as:
\[(\forall m_1 m_2. (|m_1| \leq |t_1| \wedge m_2 = 0 ) \Rightarrow \phi\defsubst{F_1 m_1}{F_2 m_2}) \Rightarrow
         \phi\defsubst{F_1\ \cons{h_1}{t_1}}{F_2\ 0}\]
On the other hand, one of the antecedents of (B) is $\forall m_1. |m_1| \leq |t_1| \Rightarrow \phi\defsubst{F_1\ m_1}{F_2\ 0}$,
so by \rname{$\Rightarrow_E$} we prove $\phi\defsubst{F_1\ \cons{h_1}{t_1}}{F_2\ 0}$, which is the consequent of (B).

The proof of (C) is symmetrical to the proof of (B).
 
To prove (D), we instantiate $\Phi$ at $F_1, F_2, \cons{h_1}{t_1}, S x_2$
\[
\begin{array}{c}
   (\forall m_1 m_2. (|m_1|,|m_2|) < (|\cons{h_1}{t_1}|, |S x_2|) \Rightarrow \phi\defsubst{F_1 m_1}{F_2 m_2}) \Rightarrow \\
         \phi\defsubst{e_1}{e_2}\subst{l_1}{\cons{h_1}{t_1}}\subst{n_2}{S x_2}\subst{f_1}{F_1}\subst{f_2}{F_2}
\end{array}
\]
by beta-expansion:
\[(\forall m_1 m_2. (|m_1|,|m_2|) < (|\cons{h_1}{t_1}|, |S x_2|) \Rightarrow \phi\defsubst{F_1 m_1}{F_2 m_2}) \Rightarrow
         \phi\defsubst{F_1\ \cons{h_1}{t_1}}{F_2\ (S x_2)}\]
One of the antecedents of (D) is exactly
$\forall m_1 m_2. (|m_1|, |m_2|)<(|\cons{h_1}{t_1}|, |S x_2|) \Rightarrow \phi\defsubst{F_1\ m_1}{F_2\ m_2}$,
so by \rname{$\Rightarrow_E$} we prove $\phi\defsubst{F_1\ \cons{h_1}{t_1}}{F_2\ (S x_2)}$, which is the 
consequent of (D).

Proof of \rname{LETREC-L} (and \rname{LETREC-R}) is analogous, and uses simple strong induction.\\

\noindent \textbf{Case.} \rname{CASE}
$\inferrule
      {\jrhol{\Gamma}{\Psi}{l_1}{\listt{\tau_1}}{l_2}{\listt{\tau_2}}{\res\ltag = \nil \Leftrightarrow \res\rtag = \nil} \\
      \jrhol{\Gamma}{\Psi, l_1= \nil, l_2= \nil}{u_1}{\sigma_1}{u_2}{\sigma_2}{\phi} \\\\
      \begin{array}{c}
         \jrhol{\Gamma}{\Psi}{v_1}{\tau_1\to\listt{\tau_1}\to\sigma_1}{v_2}{\tau_2\to\listt{\tau_2}\to\sigma_2}
               {\\ \forall h_1 h_2 t_1 t_2. l_1 = \cons{h_1}{t_1} \Rightarrow l_2 = \cons{h_2}{t_2} \Rightarrow 
                      \phi\defsubst{\res\ltag\ h_1\ t_1}{\res\rtag\ h_2\ t_2}}
      \end{array}
      }
  {\jrhol{\Gamma}{\Psi}{\caselist{l_1}{u_1}{v_1}}{\sigma_1}{\caselist{l_2}{u_2}{v_2}}{\sigma_2}{\phi}}$
  
We prove the rule for natural numbers. Applying the induction hypothesis to the premises of the rule, we have:

\begin{enumerate}[label=(\Alph*)]
  \item $\jhol{\Gamma}{\Psi}{t_1 = 0 \Leftrightarrow t_2 = 0}$
  \item $\jhol{\Gamma}{\Psi, t_1 = 0, t_2 = 0}{\phi\defsubst{u_1}{u_2}}$
  \item $\jhol{\Gamma}{\Psi}{\forall x_1, x_2. t_1 = S x_1 \Rightarrow t_2 = S x_2 \Rightarrow \phi\defsubst{v_1\ x_1}{v_2\ x_2}}$
\end{enumerate}

We want to prove:
\[\jhol{\Gamma}{\Psi}{\phi\defsubst{(\casenat{t_1}{u_1}{v_1})}{(\casenat{t_2}{u_2}{v_2})}}\]

By applying \rname{DESTR-NAT} twice, we get four premises:

\begin{enumerate}
\item $\jhol{\Gamma}{\Psi, t_1 = 0, t_2 = 0}{\phi\defsubst{(\casenat{t_1}{u_1}{v_1})}{(\casenat{t_2}{u_2}{v_2})}}$
\item $\jhol{\Gamma, m_2}{\Psi, t_1 = 0, t_2 = S m_2}{\phi\defsubst{(\casenat{t_1}{u_1}{v_1})}{(\casenat{t_2}{u_2}{v_2})}}$
\item $\jhol{\Gamma, m_1}{\Psi, t_1 = S m_1, t_2 = 0}{\phi\defsubst{(\casenat{t_1}{u_1}{v_1})}{(\casenat{t_2}{u_2}{v_2})}}$
\item $\jhol{\Gamma, m_1,m_2}{\Psi, t_1 = S m_1, t_2 = S m_2}{\phi\defsubst{(\casenat{t_1}{u_1}{v_1})}{(\casenat{t_2}{u_2}{v_2})}}$
\end{enumerate}

We can prove (2) and (3) by deriving a contradiction with \rname{NC} and (A). After beta-reducing in (1) and (4) we can easily
derive them from (B) and (C) respectively.

Proof of \rname{CASE-L} (and \rname{CASE-R}) is analogous.

\subsection*{Proof of Lemma~\ref{lem:emb-uhol-rhol}}

By the embedding into HOL, we have:
\begin{itemize}
  \item $\jhol{\Gamma}{\Psi}{\phi\subst{\res}{t_1}}$
  \item $\jhol{\Gamma}{\Psi}{\phi'\subst{\res}{t_2}}$
\end{itemize}
and by the \rname{$\wedge_I$} rule,
\[\jhol{\Gamma}{\Psi}{\phi\subst{\res}{t_1} \wedge \phi'\subst{\res}{t_2}}.\]
Finally, by undoing the embedding:
\[\jrhol{\Gamma}{\Psi}{t_1}{\sigma_1}{t_2}{\sigma_2}{\phi}.\]



\subsection*{Proof of Lemma~\ref{lem:dcc:confinement}}

By induction on the derivation of $\tau \searrow \ell$.\\

\noindent \textbf{Case.} $\inferrule{\ell \sqsubseteq
  \ell'}{\ttmonad{\ell'}(\tau) \searrow \ell}$

Since $\ell \not\sqsubseteq \adversary$ (given) and $\ell \sqsubseteq
\ell'$ (premise), it must be the case that $\ell' \not\sqsubseteq
\adversary$. Hence, by definition,
$\dccform{\adversary}{\ttmonad{\ell'}(\tau)}(x,y) = \top$.\\

\noindent \textbf{Case.} $\inferrule{\tau \searrow
  \ell}{\ttmonad{\ell'}(\tau) \searrow \ell}$

We consider two cases: \\

If $\ell' \not\sqsubseteq \adversary$, then
$\dccform{\adversary}{\ttmonad{\ell'}(\tau)}(x,y) = \top$ by
definition. \\

If $\ell' \sqsubseteq \adversary$, then
$\dccform{\adversary}{\ttmonad{\ell'}(\tau)}(x,y) =
\dccform{\adversary}{\tau}(x,y)$ by definition. By i.h.\ on the
premise, we have $\dccform{\adversary}{\tau}(x,y) \equiv
\top$. Composing, $\dccform{\adversary}{\ttmonad{\ell'}(\tau)}(x,y)
\equiv \top$.\\

\noindent \textbf{Case.} $\inferrule{\tau_1 \searrow \ell \\ \tau_2
  \searrow \ell}{\tau_1 \times \tau_2 \searrow \ell}$

By i.h.\ on the premises, we have $\dccform{\adversary}{\tau_i}(x,y)
\equiv \top$ for $i = 1,2$ and all $x, y$. Therefore,
$\dccform{\adversary}{\tau_1 \times \tau_2}(x,y) \defeq
\dccform{\adversary}{\tau_1}(\pi_1(x), \pi_1(y)) \wedge
\dccform{\adversary}{\tau_2}(\pi_2(x), \pi_2(y)) \equiv \top \wedge
\top \equiv \top$.\\

\noindent \textbf{Case.} $\inferrule{\tau_2 \searrow \ell}{\tau_1
  \rightarrow \tau_2 \searrow \ell}$

By i.h.\ on the premise, we have $\dccform{\adversary}{\tau_2}(x, y)
\equiv \top$ for all $x, y$. Hence, $\dccform{\adversary}{\tau_1
  \rightarrow \tau_2}(x, y) \defeq (\forall v, w.\,
\dccform{\adversary}{\tau_1}(v,w) \Rightarrow
\dccform{\adversary}{\tau_2}(x\ v, y\ w)) \equiv (\forall v, w.\,
\dccform{\adversary}{\tau_1}(v,w) \Rightarrow \top) \equiv \top$.



\subsection*{Proof of Theorem~\ref{thm:dcc:embed}}

By induction on the given derivation of $\Gamma \vdash e: \tau$.\\

\noindent \textbf{Case.} $\inferrule{ }{\Gamma \vdash \tbool: \bool}$

\noindent
To show:
$\jrhol{\erase{\Gamma}}
{\dccform{\adversary}{\Gamma}}
{\tbool}{\bool}
{\tbool}{\bool}
{(\res\ltag = \tbool \wedge \res\rtag = \tbool)
  \vee (\res\ltag = \fbool \wedge \res\rtag = \fbool)}
$.\\
By rule TRUE, it suffices to show $(\tbool = \tbool \wedge \tbool =
\tbool) \vee (\tbool = \fbool \wedge \tbool = \fbool)$ in HOL, which
is trivial. \\

\noindent \textbf{Case.}  $\inferrule{\Gamma \vdash e: \bool \\ \Gamma
  \vdash e_t: \tau \\ \Gamma \vdash e_f: \tau} {\Gamma \vdash
  \casebool{e}{e_t}{e_f} : \tau}$

\noindent
To show:
$\jrhol{\erase{\Gamma}}
{\dccform{\adversary}{\Gamma}}
{(\casebool{\erase{e}_1}{\erase{e_t}_1}{\erase{e_f}_1})}{\erase{\tau}}
{(\casebool{\erase{e}_2}{\erase{e_t}_2}{\erase{e_f}_2})}{\erase{\tau}}
{\dccform{\adversary}{\tau}(\res\ltag, \res\rtag)}$.\\
By i.h.\ on the first premise:\\
$\jrhol{\erase{\Gamma}}
{\dccform{\adversary}{\Gamma}}
{\erase{e}_1}{\bool}
{\erase{e}_2}{\bool}
{(\res\ltag = \tbool \wedge \res\rtag = \tbool)
  \vee (\res\ltag = \fbool \wedge \res\rtag = \fbool)}$\\
By i.h.\ on the second premise:\\
$\jrhol{\erase{\Gamma}}
{\dccform{\adversary}{\Gamma}}
{\erase{e_t}_1}{\erase{\tau}}
{\erase{e_t}_2}{\erase{\tau}}
{\dccform{\adversary}{\tau}(\res\ltag, \res\rtag)}$\\
By i.h.\ on the third premise:\\
$\jrhol{\erase{\Gamma}}
{\dccform{\adversary}{\Gamma}}
{\erase{e_f}_1}{\erase{\tau}}
{\erase{e_f}_2}{\erase{\tau}}
{\dccform{\adversary}{\tau}(\res\ltag, \res\rtag)}$\\
Applying rule BOOLCASE to the past three statements yields the
required result.\\

\noindent \textbf{Case.} $\inferrule{ }{\Gamma, x: \tau \vdash x:
  \tau}$

\noindent
To show: $\jrhol{\erase{\Gamma}, x_1: \erase{\tau}, x_2: \erase{\tau}}
{\dccform{\adversary}{\Gamma}, \dccform{\adversary}{\tau}(x_1, x_2)}
{x_1}{\erase{\tau}}{x_2}{\erase{\tau}}{\dccform{\adversary}{\tau}(\res\ltag,
  \res\rtag)}$.\\
This follows immediately from rule VAR.\\

\noindent \textbf{Case.} $\inferrule{\Gamma, x: \tau_1 \vdash e:
  \tau_2}{\Gamma \vdash \lambda x.e: \tau_1 \rightarrow \tau_2}$

\noindent
To show:
$\jrhol{\erase{\Gamma}}
{\dccform{\adversary}{\Gamma}}
{\lambda x_1. \erase{e}\ltag}{\erase{\tau_1} \rightarrow \erase{\tau_2}}
{\lambda x_2. \erase{e}\rtag}{\erase{\tau_1} \rightarrow \erase{\tau_2}}
{\forall x_1, x_2.\, \dccform{\adversary}{\tau_1}(x_1, x_2) \Rightarrow \dccform{\adversary}{\tau_2}(\res\ltag\ x_1, \res\rtag\ x_2)}$.\\
By i.h.\ on the premise:
$\jrhol{\erase{\Gamma}, x_1: \erase{\tau_1}, x_2: \erase{\tau_2}}
{\dccform{\adversary}{\Gamma}, \dccform{\adversary}{\tau_1}(x_1, x_2)}
{\erase{e}\ltag}{\erase{\tau_2}}
{\erase{e}\rtag}{\erase{\tau_2}}
{\dccform{\adversary}{\tau_2}(\res\ltag, \res\rtag)}$.\\
Applying rule ABS immediately yields the required result.\\ 
  
\noindent \textbf{Case.}
$\inferrule{\Gamma \vdash e: \tau_1 \rightarrow \tau_2 \\ \Gamma
  \vdash e': \tau_1}{\Gamma \vdash e\ e': \tau_2}$

\noindent
To show:
$\jrhol{\erase{\Gamma}}
{\dccform{\adversary}{\Gamma}}
{\erase{e}_1 \ \erase{e'}_1}{\erase{\tau_2}}
{\erase{e}_2 \ \erase{e'}_2}{\erase{\tau_2}}
{\dccform{\adversary}{\tau_2}(\res\ltag, \res\rtag)}
$.\\
By i.h.\ on the first premise:\\
$\jrhol{\erase{\Gamma}}
{\dccform{\adversary}{\Gamma}}
{\erase{e}_1}{\erase{\tau_1} \rightarrow \erase{\tau_2}}
{\erase{e}_2}{\erase{\tau_1} \rightarrow \erase{\tau_2}}
{\forall x_1, x_2.\, \dccform{\adversary}{\tau_1}(x_1,x_2) \Rightarrow \dccform{\adversary}{\tau_2}(\res\ltag\ x_1, \res\rtag\ x_2)}
$\\
By i.h.\ on the second premise:\\
$\jrhol{\erase{\Gamma}}
{\dccform{\adversary}{\Gamma}}
{\erase{e'}_1}{\erase{\tau_1}}
{\erase{e'}_2}{\erase{\tau_1}}
{\dccform{\adversary}{\tau_1}(\res\ltag, \res\rtag)}
$\\
Applying rule APP immediately yields the required result.\\

\noindent \textbf{Case.}  $\inferrule{\Gamma \vdash e: \tau
  \\ \Gamma \vdash e': \tau'} {\Gamma \vdash \pair{e}{e'}: \tau
  \times \tau'}$

\noindent
To show:
$\jrhol{\erase{\Gamma}}
{\dccform{\adversary}{\Gamma}}
{\pair{\erase{e}_1}{\erase{e'}_1}}{\erase{\tau} \times \erase{\tau'}}
{\pair{\erase{e}_2}{\erase{e'}_2}}{\erase{\tau} \times \erase{\tau'}}
{\dccform{\adversary}{\tau}(\pi_1(\res\ltag), \pi_1(\res\rtag)) \wedge \dccform{\adversary}{\tau'}(\pi_2(\res\ltag), \pi_2(\res\rtag))}
$.\\
By i.h.\ on the first premise:\\
$\jrhol{\erase{\Gamma}}
{\dccform{\adversary}{\Gamma}}
{\erase{e}_1}{\erase{\tau}}
{\erase{e}_2}{\erase{\tau}}
{\dccform{\adversary}{\tau}(\res\ltag, \res\rtag)}
$\\
By i.h.\ on the second premise:\\
$\jrhol{\erase{\Gamma}}
{\dccform{\adversary}{\Gamma}}
{\erase{e'}_1}{\erase{\tau'}}
{\erase{e'}_2}{\erase{\tau'}}
{\dccform{\adversary}{\tau'}(\res\ltag, \res\rtag)}
$\\
The required result follows from the rule PAIR. We only need to show
the third premise of the rule, i.e., the following in HOL:
\[\forall x_1 x_2 y_1 y_2. \dccform{\adversary}{\tau}(x_1, x_2) \Rightarrow 
\dccform{\adversary}{\tau'}(y_1, y_2) \Rightarrow
(\dccform{\adversary}{\tau}(\pi_1 \pair{x_1}{y_1},
\pi_1\pair{x_2}{y_2}) \wedge
\dccform{\adversary}{\tau'}(\pi_2 \pair{x_1}{y_1},
\pi_2 \pair{x_2}{y_2}))\]
Since $\pi_1 \pair{x_1}{y_1} = x_1$, etc.,
this implication simplifies to:
\[\forall x_1 x_2 y_1 y_2. \dccform{\adversary}{\tau}(x_1, x_2) \Rightarrow 
\dccform{\adversary}{\tau'}(y_1, y_2) \Rightarrow
(\dccform{\adversary}{\tau}(x_1, x_2) \wedge
\dccform{\adversary}{\tau'}(y_1, y_2))\]
which is an obvious tautology.\\

\noindent \textbf{Case.}  $\inferrule{\Gamma \vdash e: \tau \times
  \tau'}{\Gamma \vdash \pi_1(e): \tau}$

\noindent
To show:
$\jrhol{\erase{\Gamma}}
{\dccform{\adversary}{\Gamma}}
{\pi_1(\erase{e}_1)}{\erase{\tau}}
{\pi_1(\erase{e}_2)}{\erase{\tau}}
{\dccform{\adversary}{\tau}(\res\ltag, \res\rtag)}
$.\\
By i.h.\ on the premise:\\
$\jrhol{\erase{\Gamma}}
{\dccform{\adversary}{\Gamma}}
{\erase{e}_1}{\erase{\tau} \times \erase{\tau'}}
{\erase{e}_2}{\erase{\tau} \times \erase{\tau'}}
{\dccform{\adversary}{\tau}(\pi_1(\res\ltag), \pi_1(\res\rtag)) \wedge
\dccform{\adversary}{\tau'}(\pi_2(\res\ltag), \pi_2(\res\rtag))}
$\\
By rule SUB:\\
$\jrhol{\erase{\Gamma}}
{\dccform{\adversary}{\Gamma}}
{\erase{e}_1}{\erase{\tau} \times \erase{\tau'}}
{\erase{e}_2}{\erase{\tau} \times \erase{\tau'}}
{\dccform{\adversary}{\tau}(\pi_1(\res\ltag), \pi_1(\res\rtag))}
$\\
By rule PROJ$_1$, we get the required result.\\

\noindent \textbf{Case.}  $\inferrule{\Gamma \vdash e: \tau} {\Gamma
  \vdash \eret{\ell}(e): \ttmonad{\ell}(\tau)}$

\noindent
To show:
$\jrhol{\erase{\Gamma}}
{\dccform{\adversary}{\Gamma}}
{\erase{e}_1}{\erase{\tau}}
{\erase{e}_2}{\erase{\tau}}
{\dccform{\adversary}{\ttmonad{\ell}(\tau)}(\res\ltag, \res\rtag)}
$.\\
By i.h.\ on the premise:
$\jrhol{\erase{\Gamma}}
{\dccform{\adversary}{\Gamma}}
{\erase{e}_1}{\erase{\tau}}
{\erase{e}_2}{\erase{\tau}}
{\dccform{\adversary}{\tau}(\res\ltag, \res\rtag)}
$\hfill{(1)} \\
If $\ell \sqsubseteq \adversary$, then
$\dccform{\adversary}{\ttmonad{\ell}(\tau)}(\res\ltag, \res\rtag)
\defeq \dccform{\adversary}{\tau}(\res\ltag, \res\rtag)$, so the
required result is the same as (1). \\
If $\ell \not\sqsubseteq \adversary$, then
$\dccform{\adversary}{\ttmonad{\ell}(\tau)}(\res\ltag, \res\rtag)
\defeq \top$ and the required result follows from rule SUB applied to
(1). \\

\noindent \textbf{Case.}  $\inferrule{\Gamma \vdash e:
  \ttmonad{\ell}(\tau) \\ \Gamma, x: \tau \vdash e': \tau' \\ \tau'
  \searrow \ell} {\Gamma \vdash \ebind(e, x.e'): \tau'}$

\noindent
To show:
$\jrhol{\erase{\Gamma}}
{\dccform{\adversary}{\Gamma}}
{(\lambda x. \erase{e'}_1)\ \erase{e}_1}{\erase{\tau'}}
{(\lambda x. \erase{e'}_2)\ \erase{e}_2}{\erase{\tau'}}
{\dccform{\adversary}{\tau'}(\res\ltag, \res\rtag)}
$.\\
By i.h.\ on the first premise:\\
$\jrhol{\erase{\Gamma}}
{\dccform{\adversary}{\Gamma}}
{\erase{e}_1}{\erase{\tau}}
{\erase{e}_2}{\erase{\tau}}
{\dccform{\adversary}{\ttmonad{\ell}(\tau)}(\res\ltag, \res\rtag)}
$ \hfill (1)\\
By i.h.\ on the second premise:\\
$\jrhol{\erase{\Gamma}, x_1: \erase{\tau}, x_2: \erase{\tau}}
{\dccform{\adversary}{\Gamma}, \dccform{\adversary}{\tau}(x_1, x_2)}
{\erase{e'}_1}{\erase{\tau'}}
{\erase{e'}_2}{\erase{\tau'}}
{\dccform{\adversary}{\tau'}(\res\ltag, \res\rtag)}
$ \hfill (2)\\
We consider two cases:\\
\textbf{Subcase.} $\ell \sqsubseteq \adversary$. Here,
$\dccform{\adversary}{\ttmonad{\ell}(\tau)}(\res\ltag, \res\rtag)
\defeq \dccform{\adversary}{\tau}(\res\ltag, \res\rtag)$, so (1) can
be rewritten to:\\
$\jrhol{\erase{\Gamma}}
{\dccform{\adversary}{\Gamma}}
{\erase{e}_1}{\erase{\tau}}
{\erase{e}_2}{\erase{\tau}}
{\dccform{\adversary}{\tau}(\res\ltag, \res\rtag)}
$ \hfill (3)\\
Applying rule ABS to (2) yields:\\
$\jrhol{\erase{\Gamma}}
{\dccform{\adversary}{\Gamma}}
{\lambda x_1. \erase{e'}_1}{\erase{\tau} \rightarrow \erase{\tau'}}
{\lambda x_2. \erase{e'}_2}{\erase{\tau} \rightarrow \erase{\tau'}}
{\forall x_1 x_2. \dccform{\adversary}{\tau}(x_1, x_2) \Rightarrow
  \dccform{\adversary}{\tau'}(\res\ltag \ x_1, \res\rtag\ x_2)}
$ \hfill (4)\\
Applying rule APP to (4) and (3) yields:\\
$\jrhol{\erase{\Gamma}}
{\dccform{\adversary}{\Gamma}}
{(\lambda x_1. \erase{e'}_1) \ \erase{e}_1}{\erase{\tau'}}
{(\lambda x_2. \erase{e'}_2) \ \erase{e}_2}{\erase{\tau'}}
{\dccform{\adversary}{\tau'}(\res\ltag, \res\rtag)}
$ \\
which is what we wanted to prove.\\
\textbf{Subcase.} $\ell \not\sqsubseteq \adversary$. Here,
$\dccform{\adversary}{\ttmonad{\ell}(\tau)}(\res\ltag, \res\rtag)
\defeq \dccform{\adversary}{\tau}(\res\ltag, \res\rtag)$, so (1) can
be rewritten to:\\
$\jrhol{\erase{\Gamma}}
{\dccform{\adversary}{\Gamma}}
{\erase{e}_1}{\erase{\tau}}
{\erase{e}_2}{\erase{\tau}}
{\top}
$ \hfill (5)\\
Applying rule ABS to (2) yields:\\
$\jrhol{\erase{\Gamma}}
{\dccform{\adversary}{\Gamma}}
{\lambda x_1. \erase{e'}_1}{\erase{\tau} \rightarrow \erase{\tau'}}
{\lambda x_2. \erase{e'}_2}{\erase{\tau} \rightarrow \erase{\tau'}}
{\forall x_1 x_2. \dccform{\adversary}{\tau}(x_1, x_2) \Rightarrow
  \dccform{\adversary}{\tau'}(\res\ltag \ x_1, \res\rtag\ x_2)}
$ \\
By Lemma~\ref{lem:dcc:confinement} applied to the subcase
assumption $\ell \not\sqsubseteq \adversary$ and the premise $\tau'
\searrow \ell$, we have $\dccform{\adversary}{\tau'}(\res\ltag \ x_1,
\res\rtag\ x_2) \equiv \top$. So, by rule SUB:\\
$\jrhol{\erase{\Gamma}}
{\dccform{\adversary}{\Gamma}}
{\lambda x_1. \erase{e'}_1}{\erase{\tau} \rightarrow \erase{\tau'}}
{\lambda x_2. \erase{e'}_2}{\erase{\tau} \rightarrow \erase{\tau'}}
{\forall x_1 x_2. \dccform{\adversary}{\tau}(x_1, x_2) \Rightarrow \top}
$ \\
Since $(\forall x_1 x_2. \dccform{\adversary}{\tau}(x_1, x_2)
\Rightarrow \top) \equiv \top \equiv (\forall x_1, x_2. \top
\Rightarrow \top)$, we can use SUB again to get:\\
$\jrhol{\erase{\Gamma}}
{\dccform{\adversary}{\Gamma}}
{\lambda x_1. \erase{e'}_1}{\erase{\tau} \rightarrow \erase{\tau'}}
{\lambda x_2. \erase{e'}_2}{\erase{\tau} \rightarrow \erase{\tau'}}
{\forall x_1, x_2. \top \Rightarrow \top}
$ \hfill (6) \\
Applying rule APP to (6) and (5) yields:\\
$\jrhol{\erase{\Gamma}}
{\dccform{\adversary}{\Gamma}}
{(\lambda x_1. \erase{e'}_1) \ \erase{e}_1}{\erase{\tau'}}
{(\lambda x_2. \erase{e'}_2) \ \erase{e}_2}{\erase{\tau'}}
{\top}
$ \\
which is the same as our goal since
$\dccform{\adversary}{\tau'}(\res\ltag, \res\rtag) \equiv \top$.\\



\subsection*{Proof of Theorem~\ref{thm:unaryembed}}
By induction on the derivation of $\Delta; \Phi; \Omega
\jtype{k}{l}{t}{A}$. We will show few cases.

\noindent \textbf{Case.}
 $\inferrule
 {\ }{\octx, x:\grt \jtype{0}{0}{x}{\grt}}
$

\noindent We can conclude by the following derivation: 
\begin{mathpar}
\inferrule*[right=PAIR-L]
{
\inferrule*[right=VAR]
{\ }
{\juhol{\cost{\erase{\Omega}},x:\cost{\erase{\grt}}_v,\Delta}{\Phi_a,\costt{\Omega}, \costt{\grt}_v(x)}{x}{\cost{\erase{\grt}}_v}{\costt{\grt}_v(\res)}
}
\and
\inferrule*[right=Nat]{\ }
{\juhol{\cost{\erase{\Omega}},x:\cost{\erase{\grt}}_v,\Delta}{\Phi_a,\costt{\Omega}, \costt{\grt}_v(x)}{0}{\nat}{ 0\leq \res\leq 0}}
}
{\juhol{\cost{\erase{\Omega}},x:\cost{\erase{\grt}}_v,\Delta}{\Phi_a,\costt{\Omega}, \costt{\grt}_v(x)}{(x,0)}{\cost{\erase{\grt}}_v\times \nat}{\costt{\grt}_v(\pi_1 \res)\land 0\leq\pi_2\res\leq 0 }}  
\end{mathpar}
where the additional proof conditions that is needed for the [PAIR-L] rule can be easily proved in HOL.
\medskip

\noindent \textbf{Case.}
$
 \inferrule
  {
  }
  {
    \octx \jtype{0}{0}{\econst}{\tint}
  }
$

\noindent Then we can conclude by the following derivation:
\begin{mathpar}
\inferrule*[right=PAIR-L]
{
\inferrule*[right=Nat]
{\ }
{\juhol{\cost{\erase{\Omega}},\Delta}{\Phi_a,\costt{\Omega}}{\econst}{\nat}{\top}
}
\and
\inferrule*[right=Nat]{\ }{\juhol{\cost{\erase{\Omega}},\Delta}{\Phi_a,\costt{\Omega}}{0}{\nat}{ 0\leq \res\leq 0}}
}
{\juhol{\cost{\erase{\Omega}},\Delta}{\Phi_a,\costt{\Omega}}{(\econst,0)}{\nat\times\nat}{0\leq\pi_2\res\leq 0 }}  
\end{mathpar}
where the additional proof conditions that is needed for the [PAIR-L] rule can be easily proved in HOL.
\medskip

\noindent\textbf{Case.}
$
  \inferrule
    {\Delta; \Phi_a;  x: \grt_1, \Omega  \jtype{k}{l}{t}{\grt_2} }
    {
      \octx
      \jtype{0}{0}{ \lambda x.t }{\grt_1 \uarr{k}{l}
        \grt_2}
    }
$

\noindent By induction hypothesis we have 
$\juhol{\cost{\erase{\Omega}},x:\cost{\erase{\grt_1}}_v,\Delta}{\Phi,\costt{\Omega},\costt{\grt_1}_v(x) }{\cost{t}}{\cost{\erase{\grt_2}}_e}{\costt{A}^{k,l}_e(\res)}$
and we can conclude by the following derivation:
\begin{mathpar}
\inferrule*[right=PAIR-L]
{\inferrule*[right=ABS]
{\juhol{\cost{\erase{\Omega}},x:\cost{\erase{\grt_1}}_v,\Delta}{\Phi,\costt{\Omega},\costt{\grt_1}_v(x) }{\cost{t}}{\cost{\erase{\grt_2}}_e}{\costt{\grt_2}^{k,l}_e(\res)} 
}
{\juhol{\cost{\erase{\Omega}},\Delta}{\Phi,\costt{\Omega}}{\lambda x. \cost{t}}{\cost{\erase{\grt_1}}_v\to \cost{\erase{\grt_2}}_e}{\forall x. \costt{\grt_1}_v(x) \Rightarrow \costt{\grt_2}^{k,l}_e(\res x)}} 
\and
\juhol{\cost{\erase{\Omega}},\Delta}{\Phi,\costt{\Omega}}{0}{\nat}{ 0\leq \res\leq 0}
}
{\juhol{\cost{\erase{\Omega}},\Delta}{\Phi,\costt{\Omega}}{(\lambda x. \cost{t},0)}{(\cost{\erase{\grt_1}}_v\to \cost{\erase{\grt_2}}_e) \times \nat}{\forall x. \costt{\grt_1}_v(x) \Rightarrow \costt{\grt_2}^{k,l}_e((\pi_1 \res) x)} \land 0\leq\pi_2 \res\leq 0}
\end{mathpar}
where the additional proof conditions that is needed for the [PAIR-L] rule can be easily proved in HOL.
\medskip

\noindent\textbf{Case}
$
 \inferrule{
   \octx \jtype{k_1}{l_1}{t_1}{\grt_1\uarr{k}{l} \grt_2} \\
   \octx \jtype{k_2}{l_2}{t_2}{\grt_1}}
 {
   \octx \jtype{k_1+k_2+k+\kapp}{l_1+l_2+l+\kapp}{ t_1 \, t_2}{
     \grt_2}}
$

\noindent By induction hypothesis and unfolding some some definitions we have 
\begin{multline*}
\juhol{\cost{\erase{\Omega}},\Delta}{\Phi_a,\costt{\Omega}}{\cost{t_1}}{(\cost{\erase{\grt_1}}_v\to (\cost{\erase{\grt_2}}_v\times \nat))\times \nat}{}
\\
\forall h. \costt{\grt_1}_v(h)\Rightarrow ( \costt{\grt_2}_v(\pi_1 ((\pi_1 (\res)) h)) \land k\leq  \pi_2 ((\pi_1 (\res)) h) \leq l) \land k_1\leq \pi_2(\res) \leq l_1
\end{multline*}
and
$\juhol{\cost{\erase{\Omega}},\Delta}{\Phi_a,\costt{\Omega}}{\cost{t_2}}{\cost{\erase{\grt_1}}_v\times \nat }{\costt{\grt_1}_v(\pi_1(\res))\land k_2\leq \pi_2(\res) \leq l_2 }$. So, we can prove:
\begin{multline*}
\juhol{\cost{\erase{\Omega}},\Delta}{\Phi_a,\costt{\Omega}}{\mathsf{let}\, x=\cost{t_1} \, \mathsf{in}\, \mathsf{let}\, y=\cost{t_2}\, \mathsf{in}\, \pi_1(x)\,  \pi_1(y)}{\cost{\erase{\grt_2}}_v\times \nat }{}\\
\costt{\grt_2}_v(\pi_1 (\res))
\land k\leq \pi_2(\res) \leq l 
\land k_1\leq \pi_2(x) \leq l_1 
\land k_2\leq \pi_2(y)\res \leq l_2
\end{multline*}

This combined with  the definition of the cost-passing translation  $\cost{t_1\,  t_2} \defeq \mathsf{let}\, x=\cost{t_1} \, \mathsf{in}\, \mathsf{let}\, y=\cost{t_2}\, \mathsf{in}\, \mathsf{let}\, z=\pi_1(x)\  \pi_1(y)\, \mathsf{in}\, (\pi_1(z),\pi_2(x) + \pi_2(y) + \pi_2(z) +c_{app})  $ allows us to prove as required the following:
\[
\juhol{\cost{\erase{\Omega}},\Delta}{\Phi_a,\costt{\Omega}}{\cost{t_1\, t_2}}{\cost{\erase{\grt_2}}_v\times \nat }{}\\
\costt{\grt_2}_v(\pi_1 (\res))
\land k+k_1+k_2+c_{app} \leq \pi_2(\res) \leq l +l_1+l_2+c_{app}.
\]



\subsection*{Proof of Theorem~\ref{thm:relationalembed}}

To prove Theorem~\ref{thm:relationalembed}, we need three lemmas.

\begin{lemma}\label{lem:relcost:rel-imp-unary}
Suppose $\Delta; \Phi \vdash \tau \mathrel{\sf wf}$.\footnote{This
  judgment simply means that $\tau$ is well-formed in the context
  $\Delta; \Phi$. It is defined in the original RelCost
  paper~\citep{CBGGH17}.} Then, the following hold:
\begin{enumerate}
  \item
    $\jhol{\Delta}{\Phi}{\forall x y. \, \rcostt{\tau}_v(x,y)
    \Rightarrow \costt{\overline{\tau}}_v(x) \wedge
    \costt{\overline{\tau}}_v(y)}$
  \item
    $\jhol{\Delta}{\Phi}{\forall x y. \, \rcostt{\tau}_e^t(x,y)
    \Rightarrow \costt{\overline{\tau}}_e^{0,\infty}(x) \wedge
    \costt{\overline{\tau}}_e^{0,\infty}(y)}$
\end{enumerate}
Also, (3) ${\rcostt{\Gamma}} \mathrel{\Rightarrow}
\costt{\overline{\Gamma}_1} \wedge \costt{\overline{\Gamma}_2}$ where
$\overline{\Gamma}_1$ and $\overline{\Gamma}_2$ are obtained by
replacing each variable $x$ in $\overline{\Gamma}$ with $x_1$ and
$x_2$, respectively.
\end{lemma}
\begin{proof}
(1) and (2) follow by a simultaneous induction on the given
  judgment. (3) follows immediately from~(1).
\end{proof}


\begin{lemma}\label{lem:relcost:rel-imp-unary-judgment}
If $\ctx \jtypediff{t}{e_1}{e_2}{\tau}$ in RelCost, then $\Delta;
\Phi; \overline{\Gamma} \jtype{0}{\infty}{e_i}{\overline{\tau}}$ for
$i \in \{1,2\}$ in RelCost.
\end{lemma}
\begin{proof}
By induction on the given derivation.
\end{proof}


\begin{lemma}\label{lem:relcost:subtyping}
  If $\Delta; \Phi \jsubtype{\tau_1}{\tau_2}$, then $\Delta; \Phi
  \vdash \forall xy.\, \rcostt{\tau_1}_v(x,y) \Rightarrow
         {\rcostt{\tau_2}_v(x,y)}$.
\end{lemma}
\begin{proof}
By induction on the given derivation of $\Delta; \Phi
\jsubtype{\tau_1}{\tau_2}$.
\end{proof}


\begin{proof}[Proof of Theorem~\ref{thm:relationalembed}]
The proof is by induction on the given derivation of $\Delta; \Phi;
\Gamma \jtypediff{k}{t_1}{t_2}{\tau}$. We show only a few
representative cases here.\\

\noindent \textbf{Case:}
$\inferrule*[right=r-iLam]{ i::S, \ctx \jtypediff{t}{e}{e'}{\tau}
  \\ i \not \in \fiv{\Phi_a; \Gamma} } { \ctx \jtypediff{0}{\eLam
    e}{\eLam e'}{\tforalld{i}{t} \tau}} $

\noindent
To show: $\jrhol{\rcost{\Gamma}, \Delta}{\Phi_a,
  \rcostt{\Gamma}}{(\lambda \_. \cost{e}\ltag, 0)}{(\nat \to
  \cost{\erase{\tau}}_e) \times \nat} {(\lambda\_. \cost{e'}\rtag,
  0)}{(\nat \to \cost{\erase{\tau}}_e) \times
  \nat}{\rcostt{\tforalld{i}{t} \tau}_e^0(\res\ltag,\res\rtag)}$. \\
Expand $\rcostt{\tforalld{i}{t} \tau}_e^0(\res\ltag,\res\rtag)$ to
$\rcostt{\tforalld{i}{t} \tau}_v(\pi_1 \ \res\ltag,\pi_1 \ \res\rtag)
\wedge \pi_2 \res\ltag - \pi_2\ \res\rtag \leq 0$, and apply the rule
\rname{PAIR} to reduce to two proof obligations:\\
(A) $\jrhol{\rcost{\Gamma}, \Delta}{\Phi_a, \rcostt{\Gamma}}{\lambda
  \_. \cost{e}\ltag}{\nat \to \cost{\erase{\tau}}_e}
      {\lambda\_. \cost{e'}\rtag}{\nat \to
        \cost{\erase{\tau}}_e}{\rcostt{\tforalld{i}{t}
          \tau}_v(\res\ltag,\res\rtag)}$ \\
(B) $\jrhol{\rcost{\Gamma}, \Delta}{\Phi_a, \rcostt{\Gamma}}{0}{\nat}
      {0}{\nat}{\res\ltag - \res\rtag \leq 0}$

(B) follows immediately by rule \rname{ZERO}. To prove (A), expand
$\rcostt{\tforalld{i}{t} \tau}_v(\res\ltag,\res\rtag)$ and apply rule
\rname{$\sf \wedge_I$}. We get three proof obligations.\\
(C) $\jrhol{\rcost{\Gamma}, \Delta}{\Phi_a, \rcostt{\Gamma}}{\lambda
  \_. \cost{e}\ltag}{\nat \to \cost{\erase{\tau}}_e}
      {\lambda\_. \cost{e'}\rtag}{\nat \to
        \cost{\erase{\tau}}_e}{\costt{\tforall{i}{0}{\infty}\overline{\tau}}_v(\res\ltag)}$\\
(D) $\jrhol{\rcost{\Gamma}, \Delta}{\Phi_a, \rcostt{\Gamma}}{\lambda
  \_. \cost{e}\ltag}{\nat \to \cost{\erase{\tau}}_e}
      {\lambda\_. \cost{e'}\rtag}{\nat \to
        \cost{\erase{\tau}}_e}{\costt{\tforall{i}{0}{\infty}\overline{\tau}}_v(\res\rtag)}$\\
(E) $\jrhol{\rcost{\Gamma}, \Delta}{\Phi_a, \rcostt{\Gamma}}{\lambda
  \_. \cost{e}\ltag}{\nat \to \cost{\erase{\tau}}_e}
      {\lambda\_. \cost{e'}\rtag}{\nat \to
        \cost{\erase{\tau}}_e}{\forall
        z_1 z_2. \top \Rightarrow \forall i. \rcostt{\tau}_e^{t}(\res\ltag\ z_1,\res\rtag\ z_2)}$ 

To prove (C), apply Lemma~\ref{lem:relcost:rel-imp-unary-judgment} to
the given derivation (not just the premise), obtaining a RelCost
derivation for $\Delta; \Phi_a; \overline{\Gamma}
\jtype{0}{\infty}{\eLam e}{(\tforall{i}{0}{\infty}
  \overline{\tau})}$. Applying Theorem~\ref{thm:unaryembed} to this
yields $\juhol{\cost{\overline{\Gamma}},\Delta}{\Phi_a,
  \costt{\overline{\Gamma}}}{(\lambda\_. \cost{e}, 0)}{(\nat \to
  \cost{\erase{\overline{\tau}}}_e) \times
  \nat}{\costt{\tforall{i}{0}{\infty}\overline{\tau}}_e^{0,\infty}(\res)}$
in UHOL, which is the same as
$\juhol{\cost{\overline{\Gamma}},\Delta}{\Phi_a,
  \costt{\overline{\Gamma}}}{(\lambda\_. \cost{e}, 0)}{(\nat \to
  \cost{\erase{\overline{\tau}}}_e) \times
  \nat}{\costt{\tforall{i}{0}{\infty}\overline{\tau}}_v(\pi_1\ \res)
  \wedge 0 \leq \pi_2\ \res \leq \infty}$. Applying rule
\rname{PROJ$_1$}, we get
$\juhol{\cost{\overline{\Gamma}},\Delta}{\Phi_a,
  \costt{\overline{\Gamma}}}{\pi_1(\lambda\_. \cost{e}, 0)}{\nat \to
  \cost{\erase{\overline{\tau}}}_e}{\costt{\tforall{i}{0}{\infty}\overline{\tau}}_v(\res)
}$. By subject conversion,
$\juhol{\cost{\overline{\Gamma}},\Delta}{\Phi_a,
  \costt{\overline{\Gamma}}}{\lambda\_. \cost{e}}{\nat \to
  \cost{\erase{\overline{\tau}}}_e}{\costt{\tforall{i}{0}{\infty}\overline{\tau}}_v(\res)
}$. Renaming variables, we get
$\juhol{\cost{\overline{\Gamma}}_1,\Delta}{\Phi_a,
  \costt{\overline{\Gamma}_1}}{\lambda\_. \cost{e}_1}{\nat \to
  \cost{\erase{\overline{\tau}}}_e}{\costt{\tforall{i}{0}{\infty}\overline{\tau}}_v(\res)
}$.

Now note that by definition, $\rcost{\Gamma} \supseteq
\cost{\overline{\Gamma}}_1$ and by
Lemma~\ref{lem:relcost:rel-imp-unary}(3), ${\rcostt{\Gamma}}
\mathrel{\Rightarrow} \costt{\overline{\Gamma}_1}$. Hence, we also get
$\juhol{\rcost{\Gamma},\Delta}{\Phi_a,
  \rcostt{\Gamma}}{\lambda\_. \cost{e}_1}{\nat \to
  \cost{\erase{\overline{\tau}}}_e}{\costt{\tforall{i}{0}{\infty}\overline{\tau}}_v(\res)
}$. (C) follows immediately by rule \rname{UHOL-L}.

(D) has a similar proof.

To prove (E), apply the rule \rname{ABS}, getting the obligation:\\
$\jrhol{\rcost{\Gamma}, \Delta, z_1,z_2: \nat}{\Phi_a,
  \rcostt{\Gamma}}{\cost{e}\ltag}{\cost{\erase{\tau}}_e}
{\cost{e'}\rtag}{\cost{\erase{\tau}}_e}{\forall
  i. \rcostt{\tau}_e^{t}(\res\ltag,\res\rtag)}$\\
Since $z_1,z_2$ do not appear anywhere else, we can strengthen the
context to remove them, thus reducing to:
$\jrhol{\rcost{\Gamma}, \Delta}{\Phi_a,
  \rcostt{\Gamma}}{\cost{e}\ltag}{\cost{\erase{\tau}}_e}
{\cost{e'}\rtag}{\cost{\erase{\tau}}_e}{\forall
  i. \rcostt{\tau}_e^{t}(\res\ltag,\res\rtag)}$\\
Next, we transpose to HOL using Theorem~\ref{thm:equivhol}. We get the
obligation:\\
$\jhol{\rcost{\Gamma}, \Delta}{\Phi_a,
  \rcostt{\Gamma}}{\forall
  i. \rcostt{\tau}_e^{t}(\cost{e}\ltag,\cost{e'}\rtag)}$\\
This is equivalent to: \\
$\jhol{\rcost{\Gamma}, \Delta, i:S}{\Phi_a,
  \rcostt{\Gamma}}{\rcostt{\tau}_e^{t}(\cost{e}\ltag,\cost{e'}\rtag)}$\\
The last statement follows immediately from i.h.\ on the premise,
followed by transposition to HOL using Theorem~\ref{thm:equivhol}.\\

\noindent \textbf{Case:}
$  \inferrule*[right=nochange]
{
\Delta; \Phi_a; \Gamma \jtypediff{t}{e}{e}{\tau}\\
\forall x \in dom(\Gamma).~~
\Delta; \Phi_a  \jsubtype{\Gamma(x)}{\tbox{\Gamma(x)}}
}
{
\Delta; \Phi_a; \Gamma, \Gamma'; \Omega \jtypediff{0}{e}{e}{\tbox{\tau}}
}
$

\noindent To show: $\jrhol{\rcost{\Gamma}, \Delta}{\Phi_a,
  \rcostt{\Gamma}}{\cost{e}_1}{\cost{\erase{\tau}}_e}
          {\cost{e}_2}{\cost{\erase{\tau}}_e}
          {\rcostt{\tbox{\tau}}_e^0(\res\ltag,\res\rtag)}$. \\
Expanding the definition of $\rcostt{\tbox{\tau}}_e^0$, this is
  equivalent to:\\
$\jrhol{\rcost{\Gamma}, \Delta}{\Phi_a,
    \rcostt{\Gamma}}{\cost{e}_1}{\cost{\erase{\tau}}_e}
  {\cost{e}_2}{\cost{\erase{\tau}}_e} {\rcostt{\tau}_v(\pi_1
    \ \res\ltag,\pi_2\ \res\rtag) \wedge (\pi_1\ \res\ltag =
    \pi_1\ \res\rtag) \wedge (\pi_2\ \res\ltag - \pi_2\ \res\rtag \leq
    0)}$\\ 
Using rule \rname{$\sf \wedge_I$}, we reduce this to two obligations:\\
(A) $\jrhol{\rcost{\Gamma}, \Delta}{\Phi_a,
    \rcostt{\Gamma}}{\cost{e}_1}{\cost{\erase{\tau}}_e}
  {\cost{e}_2}{\cost{\erase{\tau}}_e} {\rcostt{\tau}_v(\pi_1
    \ \res\ltag,\pi_2\ \res\rtag)}$\\
(B) $\jrhol{\rcost{\Gamma}, \Delta}{\Phi_a,
    \rcostt{\Gamma}}{\cost{e}_1}{\cost{\erase{\tau}}_e}
  {\cost{e}_2}{\cost{\erase{\tau}}_e} {(\pi_1\ \res\ltag =
    \pi_1\ \res\rtag) \wedge (\pi_2\ \res\ltag - \pi_2\ \res\rtag \leq
    0)}$\\

\noindent  
By i.h.\ on the first premise,   \\
$\jrhol{\rcost{\Gamma}, \Delta}{\Phi_a,
  \rcostt{\Gamma}}{\cost{e}_1}{\cost{\erase{\tau}}_e}
{\cost{e}_2}{\cost{\erase{\tau}}_e} {\rcostt{\tau}_v(\pi_1
  \ \res\ltag,\pi_2\ \res\rtag) \wedge (\pi_2\ \res\ltag -
  \pi_2\ \res\rtag \leq t)}$\\
By rule \rname{SUB},\\
$\jrhol{\rcost{\Gamma}, \Delta}{\Phi_a,
  \rcostt{\Gamma}}{\cost{e}_1}{\cost{\erase{\tau}}_e}
{\cost{e}_2}{\cost{\erase{\tau}}_e} {\rcostt{\tau}_v(\pi_1
  \ \res\ltag,\pi_2\ \res\rtag)}$\\
which is the same as (A).\\

\noindent
To prove (B), apply Lemma~\ref{lem:relcost:subtyping} to the second
premise to get for every $x\in dom(\Gamma)$ that
$\jhol{\Delta}{\Phi_a}{\rcostt{\Gamma(x)}_v(x_1,x_2) \Rightarrow
  \rcostt{\tbox{\Gamma(x)}}_v(x_1,x_2)}$. Since
$\rcostt{\tbox{\Gamma(x)}}_v(x_1,x_2) \Rightarrow x_1 = x_2$ and from
$\rcostt{\Gamma}$ we know that $\rcostt{\Gamma(x)}_v(x_1,x_2)$, it
follows that $\jhol{\rcost{\Gamma}, \Delta}{\Phi_a,
  \rcostt{\Gamma}}{x_1 = x_2}$. Since this holds for every $x \in
dom(\Gamma)$, it follows immediately that $\jhol{\rcost{\Gamma},
  \Delta}{\Phi_a, \rcostt{\Gamma}}{\cost{e}_1 = \cost{e}_2}$. By
Theorem~\ref{thm:equivhol}, $\jrhol{\rcost{\Gamma}, \Delta}{\Phi_a,
  \rcostt{\Gamma}}{\cost{e}_1}{\cost{\erase{\tau}}_e}
{\cost{e}_2}{\cost{\erase{\tau}}_e}{\res\ltag = \res\rtag}$. (B)
follows immediately by rule \rname{SUB}.

\end{proof}



\section{Examples}

\subsection*{Factorial}

This example shows that the two following implementations of factorial, with and without accumulator, are equivalent:
\begin{align*}
{\rm fact}_1 \:&\defeq\: {\rm letrec}\ f_1 \ n_1 = \casenat{n_1}{1}{\lambda x_1. S x_1*(f_1\ x_1)} \\
{\rm fact}_2 \:&\defeq\: {\rm letrec}\ f_2\ n_2 = \lambda acc. \casenat{n_2}{acc}{\lambda x_2. f_2\ x_2\ ( S x_2 * acc)}
\end{align*}
Our goal is to prove that:
\[
\jrhol{\emptyset}{\emptyset}{{\rm fact_1}}{\nat\to\nat}{{\rm fact_2}}{\nat\to\nat\to\nat}{\forall n_1 n_2. n_1 = n_2 \Rightarrow \forall acc. (\res\ltag\ n_1)* acc = \res\rtag\ n_2\ acc}
\]
Since both programs do the same number of iterations, we can do
synchronous reasoning for the recursion at the head of the
programs. However, the bodies of the functions have different types
since ${\rm fact}_2$ receives an extra argument, the
accumulator. Therefore, we will need a one-sided application of
\rname{ABS-R}, before we can go back to reasoning synchronously. We
will then apply the \rname{CASE} rule, knowing that both terms reduce
to the same branch, since $n_1 = n_2$. On the zero branch, we will
need to prove the trivial equality $1*acc=acc$.  On the succesor
branch, we will need to prove that $Sx * ({\rm fact}\ x) * acc = {\rm
  fact}_2\ x_2\ (Sx_2 * acc)$, knowing by induction hypothesis that
such a property holds for every $m$ less that $x$.

Now we will expand on the details. We start the proof applying the
\rname{LETREC} rule, which has 2 premises:
\begin{enumerate}
\item Both functions are well-defined
\item $\sjrhol{\Gamma,  n_1,n_2 : \nat, f_1 : \nat\to\nat, f_2 : \nat\to\nat\to\nat}
              {n_1 = n_2, \forall y_1 y_2. (y_1,y_2)<(n_1,n_2) \Rightarrow y_1 = y_2 \Rightarrow \forall acc. (f_1\ y_1)* acc = f_2\ y_2\ acc  }
              {\casenat{n_1}{1}{\lambda x_1. S x_1*(f_1\ x_1)}}{\nat\to\nat\to\nat}
              {\lambda acc. \casenat{n_2}{acc}{\lambda x_2. f_2\ x_2\ ( S x_2 * acc)}}{\nat\to\nat\to\nat\to\nat}
              {n_1 = n_2 \Rightarrow \forall acc. \res\ltag * acc = \res\rtag\ acc}$
\end{enumerate}

We assume that the first premise is provable.

To prove the second premise, we start by applying ABS-R, which leaves the following proof obligation:
\[\begin{array}{c}
      \sjrhol{\Gamma,  n_1,n_2 : \nat, f_1 : \nat\to\nat, f_2 : \nat\to\nat\to\nat, acc: \nat}
             {n_1 = n_2, \forall y_1 y_2. (y_1,y_2)<(n_1,n_2) \Rightarrow y_1 = y_2 \Rightarrow \forall acc. (f_1\ y_1)* acc = f_2\ y_2\ acc, n_1 = n_2}
             {\\ \casenat{n_1}{1}{\lambda x_1. S x_1*(f_1\ x_1)}}{\nat\to\nat\to\nat}
             {\casenat{n_2}{acc}{\lambda x_2. f_2\ x_2\ ( S x_2 * acc)}}{\nat\to\nat\to\nat\to\nat}
             {\res\ltag * acc = \res\rtag}
  \end{array}\]

Now we can apply \rname{CASE}, and we have 3 premises, where $\Psi$ denotes the axioms of the previous judgment:
\begin{itemize}
  \item $\sjrhol{\Gamma}{\Psi}{n_1}{}{n_2}{}{\res\ltag = 0 \Leftrightarrow \res\rtag = 0}$
  \item $\sjrhol{\Gamma'}
             {\Psi, n_1 = 0, n_2 = 0}
             {1}{\nat\to\nat\to\nat}
             {acc}{\nat\to\nat\to\nat\to\nat}
             {\res\ltag * acc = \res\rtag}$

  \item $\sjrhol{\Gamma'}
             {\Psi}
             {\lambda x_1. S x_1*(f_1\ x_1)}{\nat\to\nat\to\nat}
             {\lambda x_2. f_2\ x_2\ ( S x_2 * acc)}{\nat\to\nat\to\nat\to\nat}
             {\forall x_1 x_2. n_1 = S x_1 \Rightarrow n_2 = S x_2 \Rightarrow (\res\ltag\ x_1)* acc = \res\rtag\ x_2}$
\end{itemize}

Premise 1 is a direct consequence of $n_1=n_2$.
Premise 2 is a trivial arithmetic identity.
To prove premise 3, we first apply the ABS rule:
\[\sjrhol{\Gamma'}
         {\Psi, n_1 = S x_1, n_2 = S x_2}
         {S x_1*(f_1\ x_1)}{\nat\to\nat\to\nat}
         {f_2\ x_2\ ( S x_2 * acc)}{\nat\to\nat\to\nat\to\nat}
         {\res\ltag * acc = \res\rtag}\]

and then by Theorem~\ref{thm:equivhol} we can finish the proof in HOL by deriving.
\[\Psi, n_1 = S x_1, n_2 = S x_2 \vdash S x_1*(f_1\ x_1) * acc = f_2\ x_2\ ( S x_2 * acc)\]

From the premises we can first prove that $(x_1,x_2)<(n_1,n_2)$ so by the inductive
hypothesis from the \rname{LETREC} rule, and the \rname{$\Rightarrow_E$} rule, 
we get
\[\forall acc. (f_1\ x_1)* acc = f_2\ x_2\ acc,\]
which we then instantiate with $S x_1 * acc$  to get
\[(f_1\ x_1)* S x_1 * acc = f_2\ x_2\ (S x_1 * acc).\]
On the other hand, from the hypotheses we also have $x_1 = x_2$, so by
\rname{CONV} we finally prove
\[(f_1\ x_1)* S x_1 * acc = f_2\ x_2\ (S x_2 * acc)\].

\subsection*{List reversal}

A related example for lists is the equivalence of reversal with and without accumulator. The structure of the proof is the same as in the factorial
example, but we will briefly show it to illustrate how the LISTCASE rule is used.
The functions are written:
\[
\begin{array}{rl}
  {\rm rev}_1 \:&\defeq\: \letrec{f_1}{l_1}{\caselist{l_1}{[]}{\lambda h_1. \lambda t_1. (f_1\ t_1) +\!+ [x_1]}} \\
  {\rm rev}_2 \:&\defeq\: \letrec{f_2}{l_2}{\lambda acc. \caselist{l_2}{acc}{ \lambda h_2. \lambda t_2. f_2\ t_2 \ (h_2 :: acc)}}
\end{array}
\]
We want to prove they are related by the following judgment:
\[
\jrhol{\emptyset}{\emptyset}{{\rm rev}_1}{\listt{\tau}\to\listt{\tau}}{{\rm rev}_2}{\listt{\tau}\to\listt{\tau}}
      {\forall l_1, l_2. l_1 = l_2 \Rightarrow \forall acc .\ (\res\ltag\ l_1) ++ acc = \res\rtag\ l_2\ acc}
\]
By the \rname{LETREC} rule, we have to prove 2 premises: 
\begin{enumerate}
\item Both functions are well-defined.
\item $\sjrhol{\Gamma, \dots}{l_1 = l_2, \forall m_1 m_2. (|m_1|, |m_2|) < (|l_1|, |l_2|) \Rightarrow m_1 = m_2 \Rightarrow \forall acc .(f_1\ m_1) ++ acc = f_2\ m_2\ acc}
              {\caselist{l_1}{[]}{\lambda h_1. \lambda t_1. (f_1\ t_1) +\!+ [x_1]}}
              {\tau\to{\sf list}_{\tau}\to{\sf list}_\tau\to{\sf list}_\tau}
              {\lambda acc. \caselist{l_2}{acc}{ \lambda h_2. \lambda t_2. f_2\ t_2 \ (h_2 :: acc)}}
              {\tau\to{\sf list}_{\tau}\to{\sf list}_{\tau}\to{\sf list}_{\tau}\to{\sf list}_{\tau}}
              {\forall acc .\ \res\ltag ++ acc = \res\rtag\ acc}$
\end{enumerate}
For the second premise, similarly as in factorial, we apply ABS-R. We have the following premise, where $\Psi$ denotes the axioms in the previous judgment:
\[\begin{array}{c}
              \sjrhol{\Gamma, \dots}{\Psi}
              {\caselist{l_1}{[]}{\lambda h_1. \lambda t_1. (f_1\ t_1) +\!+ [x_1]}}
              {\tau\to{\sf list}_{\tau}\to{\sf list}_\tau\to{\sf list}_\tau}
              {\caselist{t_2}{acc}{ \lambda h_2. \lambda t_2. f_2\ t_2 \ (h_2 :: acc)}}
              {\tau\to{\sf list}_{\tau}\to{\sf list}_{\tau}\to{\sf list}_{\tau}\to{\sf list}_{\tau}}
              {\\ \res\ltag ++ acc = \res\rtag}
  \end{array}\]
and then LISTCASE, which has three premises:
\begin{itemize}
  \item $\sjrhol{\Gamma, \dots}{\Psi}
              {l_1}
              {{\sf list}_\tau}
              {l_2}
              {{\sf list}_{\tau}}
              {\res\ltag = \nil \Leftrightarrow \res\rtag = \nil}$

  \item $\sjrhol{\Gamma, \dots}{\Psi, l_1 = [], l_2 = []}
              {[]}
              {{\sf list}_\tau}
              {acc}
              {{\sf list}_{\tau}}
              {\res\ltag ++ acc = \res\rtag}$

  \item $\sjrhol{\Gamma, \dots}{\Psi}
              {\lambda h_1. \lambda t_1. (f_1\ t_1) +\!+ [x_1]}
              {\tau\to{\sf list}_\tau\to{\sf list}_\tau}
              {\lambda h_2. \lambda t_2. f_2\ t_2 \ (h_2 :: acc)}
              {\tau\to{\sf list}_{\tau}\to{\sf list}_{\tau}}
              {\\ \forall h_1 t_1 h_2 t_2. l_1 = \cons{h_1}{t_1} \Rightarrow l_2 = \cons{h_2}{t_2} \Rightarrow \res\ltag ++ acc = \res\rtag}$

\end{itemize}
We complete the proof in a similar way as in the factorial example.

\subsection*{Proof of Theorem~\ref{thm:take-map}}

We will use without proof two unary lemmas:

\begin{lem}\label{lem:take-length}
  $\juhol{\bullet}{\bullet}{take}{\listt{\nat}\to\nat\to\listt{\nat}}{\forall l n. |r\ l\ n| = min(n,|l|)}$
\end{lem}

\begin{lem}\label{lem:map-length}
  $\juhol{\bullet}{\bullet}{map}{\listt{\nat}\to(\nat\to\nat)\to\listt{\nat}}{\forall l f. |r\ l\ f| = |l|}$
\end{lem}
We want to prove
\[
\sjrhol{\Gamma}{l_1 = l_2, n_1 = n_2, g_1 = g_2}{map\ (take\ l_1\ n_1)\ g_1 }{\listt{\sigma}}
                                    {take\ (map\ l_2\ g_2)\ n_2}{\listt{\sigma}}
                                    {\res\ltag \sqsubseteq \res\rtag \wedge |\res\ltag|={\sf min}(n_1,|l_1|) \wedge |\res\rtag|={\sf min}(n_2,|l_2|)}
\]
where $\res\ltag \sqsubseteq \res\rtag$ is the prefix ordering and is
defined as an inductive predicate:
\begin{mathpar}
  \forall l. [] \sqsubseteq l \and \forall h l_1 l_2. l_1 \sqsubseteq l_2 \Rightarrow \cons{h}{l_1} \sqsubseteq \cons{h}{l_2}
\end{mathpar}
By the helping lemmas and Lemma \ref{lem:emb-uhol-rhol}, it suffices to prove just the first conjunct:
\[
\sjrhol{\Gamma}{l_1 = l_2, n_1 = n_2, g_1 = g_2}{map\ (take\ l_1\ n_1)\ g_1 }{\listt{\sigma}}
                                    {take\ (map\ l_2\ g_2)\ n_2}{\listt{\sigma}}
                                    {\res\ltag \sqsubseteq \res\rtag}
\]
The derivation begins by applying the APP{-}R rule. We get the
following judgment on $n_2$:
\[l_1 = l_2, n_1 = n_2, g_1 = g_2 \vdash n_2 \mid \res \geq |take\ l_1\ n_1|\tag{1}\]
and a main premise:
\[
\sjrhol{\Gamma }{l_1 = l_2, n_1 = n_2, g_1 = g_2}{map\ (take\ l_1\ n_1)\ g_1}{\listt{\sigma}}{take\ (map\ l_2\ g_2)}{\nat\to\listt{\sigma}}
       { \forall x_2. x_2 \geq |take\ l_1\ n_1| \Rightarrow \res\ltag \sqsubseteq (\res\rtag\ x_2)}
\tag{2}\]
Notice that we have chosen the premise $x_2 \geq |take\ l_1\ n_1|$
because we are trying to prove $\res\ltag \sqsubseteq (\res\rtag\ x_2)$,
which is only true if we take a larger prefix on the right than on the
left.  The judgment (1) is easily proven from the fact that
$|take\ l_1\ n_1| = min(n_1, |l_1|) \leq n_1 = n_2$, which we get from
the lemmas.
To prove (2) we first apply APP{-}L with a trivial condition $g_1 =
g_2$ on $g_1$. Then we apply APP and we have two premises:
\begin{enumerate}[label=(\Alph*)]
\item $\sjrhol{\Gamma}{\Psi}{take\ l_1\ n_1}{\listt{\sigma}}{map\ l_2\ g_2}{\listt{\sigma}}{ \res\ltag \sqsubseteq_{g_2} \res\rtag }$
\item $\sjrhol{\Gamma}{\Psi}{map}{\listt{\sigma}\to\listt{\sigma}}{take}{\listt{\sigma}\to\nat\to\listt{\sigma}}
             {\forall m_1 m_2. m_1 \sqsubseteq_{g_2} m_2 \Rightarrow (\forall g_1. g_1 = g_2 \Rightarrow \forall x_2. x_2 \geq |m_1| \Rightarrow  (\res\ltag\ m_1\ g_1) \sqsubseteq (\res\rtag\ m_2\ x_2))}$
\end{enumerate}
where $\sqsubseteq_g$ is defined as an inductive predicate parametrized by $g$:
\begin{mathpar}
  \forall l. [] \sqsubseteq_g l \and \forall h l_1 l_2. l_1 \sqsubseteq_g l_2 \Rightarrow \cons{h}{l_1} \sqsubseteq_g \cons{(g h)}{l_2}
\end{mathpar}

We first show how to prove (A).
We start by applying APP with a trivial condition for the arguments
to get:
\[
\sjrhol{\Gamma}{\Psi}{take\ l_1}{\nat\to\listt{\sigma}}{map\ l_2}{\listt{\sigma}}
      {\forall x_1 g_2. (\res\ltag\ x_1) \sqsubseteq_{g_2} (\res\rtag\ g_2) }
\]
We then apply APP, which has two premises, one of them
equating $l_1$ and $l_2$. The other one is:
\[
\sjrhol{\Gamma}{\Psi}{take}{\listt{\sigma}\to\nat\to\listt{\sigma}}{map}{\listt{\sigma}\to\listt{\sigma}}
       {\forall m_1 m_2. m_1 = m_2 \Rightarrow 
          \forall x_1 g_2. (\res\ltag\ m_1\ x_1) \sqsubseteq_{g_2} (\res\rtag\ m_2\ g_2) }
\]

To complete this branch of the proof, we apply LETREC. We need to
prove the following premise:
\[\sjrhol{\Gamma}{\Psi, m_1 = m_2, \forall k_1 k_2. (k_1 , k_2) < (m_1, m_2) \Rightarrow k_1 = k_2 \Rightarrow 
          \forall x_1 g_2. (f_1\ k_1\ x_1) \sqsubseteq_{g_2} (f_2\ k_2\ g_2)}{\lambda n_1. e_1}{.}{\lambda g_2. e_2}{.}
          {\forall x_1 g_2. (\res\ltag\ x_1) \sqsubseteq_{g_2} (\res\rtag\ g_2)}\]

Where $e_1,e_2$ abbreviate the bodies of the functions:
\[\begin{array}{rrl}
     e_1 \;\defeq\; \caselist{m_1}{[]&&\\}
                                  {&\lambda h_1 t_1. \casenat{x_1}{&[]\\&}
                                                                 {&\lambda y_1. \cons{h_1}{f_1\ t_1\ y_1}}}
\end{array}
\]
\[\begin{array}{rl}
    e_2 \;\defeq\; \caselist{m_2}{&[]\\}
                                {&\lambda h_2 t_2. \cons{(g_2\ h_2)}{(f_2\ t_2\ g_2)}}
\end{array}
\]

If we apply ABS we get a premise: 
\[\sjrhol{\Gamma}{\Psi, m_1 = m_2, \forall k_1 k_2. (k_1 , k_2) < (m_1, m_2) \Rightarrow k_1 = k_2 \Rightarrow 
          \forall x_1 g_2. (f_1\ k_1\ x_1) \sqsubseteq_{g_2} (f_2\ k_2\ g_2)}{e_1}{.}{e_2}{.}
          {\res\ltag \sqsubseteq_f \res\rtag}\]

And now we can apply a synchronous CASE rule, since we have a premise
$m_1 = m_2$. This yields 3 proof obligations, where $\Psi'$ is the set of axioms in the previous judgment:

\begin{enumerate}[label=(A.\arabic*)]
  \item $\sjrhol{\Gamma}{\Psi'}{m_1}{}{m_2}{}{\res\ltag = \nil \Leftrightarrow \res\rtag = \nil}$
  \item $\sjrhol{\Gamma'}{\Psi'}{[]}{.}{[]}{.}{\res\ltag \sqsubseteq_f \res\rtag}$
  \item $\sjrhol{\Gamma'}{\Psi'}
          {\lambda h_1 t_1. \casenat{x_1}{[]}{\lambda y_1. \cons{h_1}{f_1\ t_1\ y_1}}}{.}
          {\\ \lambda h_2 t_2. \cons{(g_2\ h_2)}{(f_2\ t_2\ g_2)}}{.}
          {\forall h_1 t_1 h_2 t_2. m_1 = \cons{h_1}{t_1} \Rightarrow m_2 = \cons{h_2}{t_2} \Rightarrow (\res\ltag\ h_1\ t_1) \sqsubseteq_{g_2} (\res\rtag\ h_2\ t_2)}$

\end{enumerate}    

Premises (A.1) and (A.2) are trivial. To prove (A.3) we first apply ABS twice:
\[\begin{array}{c}
    \sjrhol{\Gamma'}{\Psi', m_1 = \cons{h_1}{t_1} , m_2 = \cons{h_2}{t_2}}
          {\casenat{n_1}{[]}{\lambda y_1. \cons{h_1}{f_1\ t_1\ y_1}}}{.}
          {\cons{(g_2\ h_2)}{(f_2\ t_2\ g_2)}}{.}
          {\res\ltag \sqsubseteq_{g_2} \res\rtag}
\end{array}\]

Next, we apply CASE{-}L, which has the following two premises:
\begin{enumerate}[label=(A.3.\roman*)]
  \item $\sjrhol{\Gamma'}{\Psi', m_1 = \cons{h_1}{t_1} , m_2 = \cons{h_2}{t_2}, n_1 = 0}
          {[]}{.}
          {\cons{(g_2\ h_2)}{(f_2\ t_2\ g_2)}}{.}
          {\res\ltag \sqsubseteq_{g_2} \res\rtag}$
  \item $\sjrhol{\Gamma'}{\Psi', m_1 = \cons{h_1}{t_1} , m_2 = \cons{h_2}{t_2}}
          {\lambda y_1. \cons{h_1}{f_1\ t_1\ y_1}}{.}
          {\cons{(g_2\ h_2)}{(f_2\ t_2\ g_2)}}{.}
          {\forall y_1. n_1 = S y_1 \Rightarrow (\res\ltag\ y_1) \sqsubseteq_{g_2} \res\rtag}$
\end{enumerate}

Premise (A.3.i) can be directly derived in HOL from the definition of
$\sqsubseteq_{g_2}$. To prove (A.3.ii) we need to make use of our
inductive hypothesis:
\[\forall k_1 k_2. (k_1 , k_2) < (m_1, m_2) \Rightarrow k_1 = k_2 \Rightarrow 
          \forall x_1 g_2. (f_1\ k_1\ x_1) \sqsubseteq_{g_2} (f_2\ k_2\ g_2)\]
In particular, from the premises $m_1 = \cons{h_1}{t_1}$ and $m_2 =\cons{h_2}{t_2}$
we can deduce $(t_1 , t_2) < (m_1,m_2)$.
Aditionally, from the premise $m_1 = m_2$ we prove $t_1 =t_2$.
Therefore, from the inductive hypothesis we derive $\forall
x_1 g_2. (f_1\ t_1\ x_1) \sqsubseteq_{g_2} (f_2\ t_2\ g_2)$, and by
definition of $\sqsubseteq_{g_2}$, and the fact that $h_1 = h_2$,
for every $y$ we can prove 
$\cons{h_1}{(f_1\ t_1\ y)} \:\sqsubseteq_{g_2}\:
\cons{(g_2\ h_2)}{f_2\ t_2}$. By Theorem~\ref{thm:equivhol}, we can prove (A.3.ii).
\\

We will now show how to prove (B) :
\[\sjrhol{\Gamma}{\Psi}{map}{\listt{\sigma}\to\listt{\sigma}}{take}{\listt{\sigma}\to\nat\to\listt{\sigma}}
             {\forall m_1 m_2. m_1 \sqsubseteq_{g_2} m_2 \Rightarrow (\forall g_1. g_1 = g_2 \Rightarrow \forall x_2. x_2 \geq |m_1| \Rightarrow (\res\ltag\ m_1\ g_1) \sqsubseteq (\res\rtag\ m_2\ x_2))}\]
On this branch we will also use LETREC. We have to prove a premise:

\[\sjrhol{\Gamma}{\Psi, \Phi}{\lambda g_1. e_2}{.}{\lambda x_2. e_1}{.}
         {\forall g_1. g_1 = g_2 \Rightarrow \forall x_2. x_2 \geq |m_1| \Rightarrow (\res\ltag\ g_1) \sqsubseteq (\res\rtag\ x_2)}\]
where
\[ \Phi \:\defeq\:  \begin{array}{c}
                      m_1 \sqsubseteq_{g_2} m_2,\\
                      \forall k_1 k_2. (k_1,k_2) < (m_1, m_2) \Rightarrow k_1 \sqsubseteq_{g_2} k_2 \Rightarrow
                         (\forall g_1. g_1 = g_2 \Rightarrow \forall x_2. x_2 \geq |k_1| \Rightarrow (\res\ltag\ k_1\ g_1) \sqsubseteq (\res\rtag\ k_2\ x_2))
                    \end{array}\]

We start by applying ABS. Our goal is to prove:
\[\sjrhol{\Gamma}{\Psi, \Phi, x_2 \geq |m_1|, g_1 = g_2}
                      {\begin{array}{l}
                          \caselist{m_1}{[] \\}
                                   {\lambda h_1 t_1. \cons{(g_1\ h_1)}{(f_1\ t_1\ g_1)}}
                       \end{array}}{.}
                      {\begin{array}{l}
                          \caselist{m_2}{[] \\}
                          {\lambda h_2 t_2. \casenat{x_2}{[] \\}{\lambda y_2. \cons{h_2}{f_2\ t_2\ y_2}}}
                       \end{array}}{.}
         {\res\ltag \sqsubseteq \res\rtag}\]

Notice that we have $\alpha$-renamed the variables to have the
appropriate subscript.
Now we want to apply a CASE rule, but the lists over which we are
matching are not necessarily of the same length. Therefore, we use the
asynchronous LISTCASE{-}A rule. We have to prove four premises:
\begin{enumerate}[label = (B.\arabic*)]
\item $\sjrhol{\Gamma}
              {\Psi, \Phi, x_2 \geq |m_1|, g_1 = g_2, m_1 = [], m_2 = []}
              {[]}{.}
              {[]}{.}
              {\res\ltag \sqsubseteq \res\rtag}$
\item $\sjrhol{\Gamma}
              {\Psi, \Phi, x_2 \geq |m_1|, g_1 = g_2, m_1 = []}
              {[]}{.}
              {\\ \lambda h_2 t_2.
                      \casenat{x_2}{[]}
                              {\lambda y_2. \cons{h_2}{f_2\ t_2\ y_2}}}{.}
              {\forall h_2 t_2. m_2 = \cons{h_2}{t_2} \Rightarrow
                  \res\ltag \sqsubseteq (\res\rtag\ h_2\ t_2)}$
\item $\sjrhol{\Gamma}
              {\Psi, \Phi, x_2 \geq |m_1|, g_1 = g_2, m_2 = []}
              {\lambda h_1 t_1. \cons{(g_1\ h_1)}{(f_1\ t_1\ g_1)}}{.}
              {[]}{.}
              {\forall h_1 t_1. m_1 = \cons{h_1}{t_1} \Rightarrow
                  (\res\ltag\ h_1\ t_1) \sqsubseteq \res\rtag}$
\item $\sjrhol{\Gamma}
              {\Psi, \Phi, x_2 \geq |m_1|, g_1 = g_2}
              {\lambda h_1 t_1. \cons{(g_1\ h_1)}{(f_1\ t_1\ g_1)}}{.}
              {\\ \lambda h_2 t_2. \casenat{x_2}{[]}
                                  {\lambda y_2. \cons{h_2}{f_2\ t_2\ y_2}}}{.}
              {\\ \forall h_1 t_1 h_2 t_2. m_1 = \cons{h_1}{t_1} \Rightarrow
                       m_2 = \cons{h_1}{t_1} \Rightarrow
                           (\res\ltag\ h_1\ t_1) \sqsubseteq (\res\rtag\ h_2\ t_2)}$
\end{enumerate}

Premises (B.1) and (B.2) are trivially derived from the definition of
the $\sqsubseteq$ predicate. To prove premise (B.3) we see that we
have premises $m_1 \sqsubseteq_{g_2} m_2$, $m_2 = []$, and $m_1 =
h_1::t_2$, from which we can derive a contradiction.

It remains to prove (B.4). To do so, we apply ABS twice and then
NATCASE{-}R, which has two premises:
\begin{enumerate}[label = (B.4.\roman*)]
  \item $\sjrhol{\Gamma}{\Psi, \Phi, x_2 \geq |m_1|, g_1 = g_2, m_1 = \cons{h_1}{t_1}, m_2 = \cons{h_1}{t_1}, x_2 = 0}
                      {\cons{(g_1\ h_1)}{(f_1\ t_1\ g_1)}}{.}
                      {[]}{.}
         {\res\ltag \sqsubseteq \res\rtag}$
  \item $\sjrhol{\Gamma}{\Psi, \Phi, x_2 \geq |m_1|, g_1 = g_2, m_1 = \cons{h_1}{t_1}, m_2 = \cons{h_1}{t_1}}
                      {\cons{(g_1\ h_1)}{(f_1\ t_1\ g_1)}}{.}
                      {\lambda y_2. \cons{h_2}{f_2\ t_2\ y_2}}{.}
         {\\ \forall y_2. x_2 = S y_2 \Rightarrow \res\ltag \sqsubseteq (\res\rtag\ y_2)}$
\end{enumerate}

To prove (B.4.i) we derive a contradiction between the premises. From $x_2 = 0$ and the 
premise $x_2 \geq |m_1|$ we derive
$m_1 = []$ and, together with $m_1 = \cons{h_1}{t_1}$ we arrive at a
contradiction by applying NC. 

To prove (B.4.ii) we need to use the
induction hypothesis. From $m_1 = \cons{h_1}{t_1}, m_2 = \cons{h_1}{t_1}$
we can prove that $|t_1|<|m_1|$
and $|t_2|<|m_2|$, so we can do a CUT with the i.h. and derive:
\[t_1 \sqsubseteq_{g_2} t_2 \Rightarrow
                         (\forall g_1. g_1 = g_2 \Rightarrow \forall x_2. x_2 \geq |t_1| \Rightarrow (f_1\ t_1\ g_1) \sqsubseteq (f_2\ t_2\ x_2))\]
By assumption, $m_1 \sqsubseteq_{g_2} m_2$, so $t_1 \sqsubseteq_{g_2}
t_2$.  Moreover, also by assumption $g_1 = g_2$, and $S y_2 = x_2 \geq
|m_1| = S |t_1|$, so $y_2 \geq |t_1|$.  So if we instantiate the
i.h. with $g_1$ and $y_2$, and apply CUT again, we can prove:
\[(f_1\ t_1\ g_1) \sqsubseteq (f_2\ t_2\ y_2)\]
On the other hand, since $\cons{h_1}{t_1} \sqsubseteq_{g_2}
\cons{h_2}{t_2}$, then (by elimination of $\sqsubseteq_{g_2}$) we can
derive $g_1 h_1 = h_2$ and by definition of $\sqsubseteq$,
$\cons{(g_1\ h_1)}{(f_1\ t_1\ g_1)} \sqsubseteq
\cons{h_2}{(f_2\ t_2\ y_2)}$. So we can apply Theorem~\ref{thm:equivhol} and
prove (B.4.ii). This ends the proof.\qed



\subsection*{Proof of Theorem~\ref{thm:insertion-sort}}

We need two straightforward lemmas in UHOL. The lemmas state that
sorting preserves the length and minimum element of a list.

\begin{lem}\label{lem:isort:length}
Let $\tau \defeq \listt{\nat} \rightarrow \listt{\nat}$. Then,
(1)~$\juhol{\bullet}{\bullet}{{\sf insert}}{\nat \rightarrow
  \tau}{\forall x\, l.\, |\pi_1(\res\ x\ l)| = 1 + |l|}$, and
(2)~$\juhol{\bullet}{\bullet}{{\sf isort}}{\tau}{\forall x.\, |\pi_1
  (\res\ x)| = |x|}$.
\end{lem}

\begin{lem}\label{lem:isort:lmin}
Let $\tau \defeq \listt{\nat} \rightarrow \listt{\nat}$. Then,
(1)~$\juhol{\bullet}{\bullet}{{\sf insert}}{\nat \rightarrow
  \tau}{\forall x\, l.\, {\sf lmin}(\pi_1 (\res\ x\ l)) = {\sf min}(x,
  {\sf lmin}(l))}$, and (2)~$\juhol{\bullet}{\bullet}{{\sf
    isort}}{\tau}{\forall x.\, {\sf lmin}(\pi_1 (\res\ x)) = {\sf
    lmin}(x)}$.
\end{lem}

\begin{proof}[Proof of Theorem~\ref{thm:insertion-sort}]
We prove the theorem using LETREC. We actually show the following
stronger theorem, which yields a stronger induction hypothesis in the
proof.
\[
\jrhol{\bullet}{\bullet}{{\sf isort}}{\tau}{{\sf isort}}{\tau}{\forall
  x_1\, x_2.\, ({\sf sorted}(x_1) \wedge |x_1| = |x_2|) \Rightarrow
  (\pi_2(\res\ltag \ x_1) \leq \pi_2(\res\rtag\ x_2)) \wedge
  \underline{(\res\ltag \ x_1 = {\sf isort}\ x_1) \wedge (\res\rtag
    \ x_2 = {\sf isort}\ x_2)}}
\]
Let $\iota$ denote the inductive hypothesis:
\[
\begin{array}{@{}l@{}l@{}}
\iota \defeq \forall m_1\, m_2.\,(|m_1|, |m_2|) < (|x_1|, |x_2|)
\; & \Rightarrow ({\sf sorted}(m_1) \wedge |m_1| = |m_2|) \\
& \Rightarrow \pi_2(isort_1 \ m_1) \leq \pi_2(isort_2\ m_2) \wedge
(isort_1\ m_1 = {\sf isort}\ m_1) \wedge (isort_2\ m_2 = {\sf
  isort}\ m_2)
\end{array}
\]
and $e$ denote the body of the function ${\sf isort}$:
\[e \defeq \begin{array}[t]{@{}r@{}l@{}l@{}}
{\rm case}\ l\ {\rm of}\ []\; & \mapsto ([], 0); \\
\_ :: \_ \; & \mapsto \lambda h\, t. \, & {\rm let}\ s = isort \ t \\
& & {\rm let}\ s' = {\sf insert}\ h\ (\pi_1\ s) \ {\rm in}\\
& & (\pi_1 \ s', (\pi_2\ s) + (\pi_2\ s'))
\end{array}\]

By LETREC, it suffices to prove the following (we omit simple types
for easier reading; they play no essential role in the proof).
\[
\jrholu{isort_1, isort_2, x_1, x_2}{{\sf sorted}(x_1), |x_1| = |x_2|, \iota}
       {e[isort_1/isort][x_1/l]}{e[isort_2/isort][x_2/l]}
{
\left(  
\begin{array}{l}
  \pi_2\ \res\ltag \leq \pi_2\ \res\rtag \\
  \wedge\; \res\ltag = {\sf isort}\ x_1  \\
  \wedge\; \res\rtag = {\sf isort}\ x_2
\end{array}
\right)
}
\]

Following the structure of $e$, we next apply the rule LISTCASE. This
yields the following two main proof obligations, corresponding to the
two case branches (the third proof obligation, $x_1 = []
\Leftrightarrow x_2 = []$ follows immediately from the assumption
$|x_1| = |x_2|$).

\begin{equation}\label{goal:isort:nil}
 \jrholu{isort_1, isort_2, x_1, x_2}{{\sf sorted}(x_1), |x_1| = |x_2|,
   \iota, x_1= x_2 = []} {([], 0)}{([], 0)}{(\pi_2\ \res\ltag \leq
   \pi_2\ \res\rtag) \wedge (\res\ltag = {\sf isort}\ x_1) \wedge
   (\res\rtag = {\sf isort}\ x_2)}
\end{equation}

\begin{equation}\label{goal:isort:cons}
\begin{array}{l}
isort_1, isort_2, x_1, x_2,h_1,t_1,h_2,t_2 \mid\\
{\sf sorted}(x_1), |x_1| = |x_2|, \\
\iota, \underline{x_1 = \cons{h_1}{t_1}, x_2 = \cons{h_2}{t_2}}
\end{array}
\vdash
\begin{array}{l}
{\rm let}\ s = isort_1 \ t_1 \\
{\rm let}\ s' = {\sf insert}\ h_1\ (\pi_1\ s) \ {\rm in}\\
(\pi_1 \ s', (\pi_2\ s) + (\pi_2\ s'))
\end{array}
\sim
\begin{array}{l}
{\rm let}\ s = isort_2 \ t_2 \\
{\rm let}\ s' = {\sf insert}\ h_2\ (\pi_1\ s) \ {\rm in}\\
(\pi_1 \ s', (\pi_2\ s) + (\pi_2\ s'))
\end{array}
\Bigg |
\begin{array}{l}
  \pi_2 \ \res\ltag \leq \pi_2 \ \res\rtag \\
  \wedge\; \res\ltag = {\sf isort}\ x_1 \\
  \wedge\; \res\rtag = {\sf isort}\ x_2
\end{array}
\end{equation}

(\ref{goal:isort:nil}) is immediate: By Theorem~\ref{thm:equivhol}, it
suffices to show that $(\pi_2 ([], 0) \leq \pi_2 ([], 0)) \wedge (([],
0) = {\sf isort}\ x_1) \wedge (([], 0) = {\sf isort}\ x_2)$. Since
$x_1 = x_2 = []$ by assumption here, this is equivalent to $(\pi_2
([], 0) \leq \pi_2 ([], 0)) \wedge (([], 0) = {\sf isort}\ []) \wedge
(([], 0) = {\sf isort}\ [])$, which is trivial by direct computation.

To prove (\ref{goal:isort:cons}), we expand the outermost occurences
of ${\rm let}$ in both to function applications using the definition
${\rm let}\ x = e_1\ {\rm in}\ e_2 \defeq (\lambda x.e_2)
\ e_1$. Applying the rules APP and ABS, it suffices to prove the
following for any $\phi$ of our choice.

\begin{equation}\label{goal:isort:cons:1}
isort_1, isort_2, x_1, x_2,h_1,t_1,h_2,t_2\ \Bigg|
\begin{array}{l}
{\sf sorted}(x_1), |x_1| = |x_2|, \\
\iota, x_1 = \cons{h_1}{t_1}, x_2 = \cons{h_2}{t_2}
\end{array}
\vdash
\begin{array}{l}
isort_1\ t_1
\end{array}
\sim
\begin{array}{l}
isort_2\ t_2
\end{array}
\Bigg|
\ \phi
\end{equation}

\begin{equation}\label{goal:isort:cons:2}
\begin{array}{l}
isort_1, isort_2, x_1, x_2,h_1,t_1,h_2,t_2,s_1,s_2 \mid\\
{\sf sorted}(x_1), |x_1| = |x_2|, \\
\iota, x_1 = \cons{h_1}{t_1}, x_2 = \cons{h_2}{t_2}\\
\underline{\phi[s_1/\res\ltag][s_2/\res\rtag]}
\end{array}
\vdash
\begin{array}{l}
{\rm let}\ s' = {\sf insert}\ h_1\ (\pi_1\ s_1) \ {\rm in}\\
(\pi_1 \ s', (\pi_2\ s_1) + (\pi_2\ s'))
\end{array}
\sim
\begin{array}{l}
{\rm let}\ s' = {\sf insert}\ h_2\ (\pi_1\ s_2) \ {\rm in}\\
(\pi_1 \ s', (\pi_2\ s_2) + (\pi_2\ s'))
\end{array}
\Bigg|
\begin{array}{l}
  \pi_2 \ \res\ltag \leq \pi_2 \ \res\rtag \\
  \wedge\; \res\ltag = {\sf isort}\ x_1 \\
  \wedge\; \res\rtag = {\sf isort}\ x_2
\end{array}
\end{equation}

We choose $\phi$ as follows:
\[ \phi \defeq \pi_2 \ \res\ltag \leq \pi_2\ \res\rtag \wedge
\res\ltag = {\sf isort}(t_1) \wedge \res\rtag = {\sf isort}(t_2)
\wedge |\pi_1 \ \res\ltag| = |\pi_1 \ \res\rtag| \wedge {\sf
  lmin}(t_1) = {\sf lmin}(\pi_1\ \res\ltag)\]

\noindent
\underline{Proof of (\ref{goal:isort:cons:1})}: By
Theorem~\ref{thm:equivhol}, it suffices to prove the following five
statements in HOL under the context of
(\ref{goal:isort:cons:1}). These statements correspond to the five
conjuncts of $\phi$.

\begin{equation}\label{goal:isort:cons:3}
  \pi_2 (isort_1\ t_1) \leq \pi_2(isort_2\ t_2)
\end{equation}

\begin{equation}\label{goal:isort:cons:4}
  isort_1\ t_1 = {\sf isort}\ t_1
\end{equation}

\begin{equation}\label{goal:isort:cons:5}
  isort_1\ t_2 = {\sf isort}\ t_2
\end{equation}

\begin{equation}\label{goal:isort:cons:6}
  |\pi_1(isort_1\ t_1)| = |\pi_1(isort_2\ t_2)|
\end{equation}

\begin{equation}\label{goal:isort:cons:7}
  {\sf lmin}(t_1) = {\sf lmin}(\pi_1 (isort_1\ t_1))
\end{equation}

(\ref{goal:isort:cons:3})--(\ref{goal:isort:cons:5}) follow from the
induction hypothesis $\iota$ instantiated with $m_1 := t_1, m_2 :=
t_2$. Note that because $x_1 = \cons{h_1}{t_1}$ and $x_2 =
\cons{h_2}{t_2}$, we can prove (in HOL) that $(|t_1|, |t_2|) < (|x_1|,
|x_2|)$. Since, $|x_1| = |x_2|$, $x_1 = \cons{h_1}{t_1}$ and $x_2 =
\cons{h_2}{t_2}$, we can also prove that $|t_1| = |t_2|$. Finally,
from the axiomatic definition of ${\sf sorted}$ and the assumption
${\sf sorted}(x_1)$ it follows that ${\sf sorted}(t_1)$. These
together allow us to instantiate the i.h.\ $\iota$ and immediately
derive (\ref{goal:isort:cons:3})--(\ref{goal:isort:cons:5}).

To prove (\ref{goal:isort:cons:6}), we use (\ref{goal:isort:cons:4})
and (\ref{goal:isort:cons:5}), which reduces (\ref{goal:isort:cons:6})
to $|\pi_1 ({\sf isort}\ t_1)| = |\pi_1 ({\sf isort}\ t_2)|$. To prove
this, we apply Theorem~\ref{thm:equiv-uhol-hol} to
Lemma~\ref{lem:isort:length}, yielding $\forall x.\, |\pi_1({\sf
  isort}\ x)| = |x|$. Hence, we can further reduce our goal to proving
$|t_1| = |t_2|$, which we already did above.

To prove (\ref{goal:isort:cons:7}), we use (\ref{goal:isort:cons:4}),
which reduces (\ref{goal:isort:cons:7}) to ${\sf lmin}(t_1) = {\sf
  lmin}(\pi_1 ({\sf isort}\ t_1))$. This follows immediately from
Theorem~\ref{thm:equiv-uhol-hol} applied to
Lemma~\ref{lem:isort:lmin}.

\noindent
This proves (\ref{goal:isort:cons:1}).\\

\noindent
\underline{Proof of (\ref{goal:isort:cons:2})}: We expand the
definition of ${\rm let}$ and apply the rules APP and ABS to reduce
(\ref{goal:isort:cons:2}) to proving the following for any $\phi'$.

\begin{equation}\label{goal:isort:let2:1}
\begin{array}{l}
  isort_1, isort_2, x_1, x_2,\\
  h_1,t_1,h_2,t_2,s_1,s_2\
\end{array}
\Bigg|
\begin{array}{l}
{\sf sorted}(x_1), |x_1| = |x_2|, \\
\iota, x_1 = \cons{h_1}{t_1}, x_2 = \cons{h_2}{t_2},\\
\phi\defsubst{s_1}{s_2}
\end{array}
\vdash
\begin{array}{l}
{\sf insert}\ h_1\ (\pi_1\ s_1)
\end{array}
\sim
\begin{array}{l}
{\sf insert}\ h_2\ (\pi_1\ s_2)
\end{array}
\Bigg|
\ \phi'
\end{equation}

\begin{equation}\label{goal:isort:let2:2}
\begin{array}{l}
isort_1, isort_2, x_1, x_2,h_1,t_1,h_2,t_2,s_1,s_2,s_1',s_2' \mid\\
{\sf sorted}(x_1), |x_1| = |x_2|, \\
\iota, x_1 = \cons{h_1}{t_1}, x_2 = \cons{h_2}{t_2}\\
\phi[s_1/\res\ltag][s_2/\res\rtag], \phi'\defsubst{s_1'}{s_2'}
\end{array}
\vdash
\begin{array}{l}
(\pi_1 \ s_1', (\pi_2\ s_1) + (\pi_2\ s_1'))
\end{array}
\sim
\begin{array}{l}
(\pi_1 \ s_2', (\pi_2\ s_2) + (\pi_2\ s_2'))
\end{array}
\Bigg|
\begin{array}{l}
  \pi_2 \ \res\ltag \leq \pi_2 \ \res\rtag \\
  \wedge\; \res\ltag = {\sf isort}\ x_1 \\
  \wedge\; \res\rtag = {\sf isort}\ x_2
\end{array}
\end{equation}

We pick the following $\phi'$:

\[
\phi' \defeq \pi_2\ \res\ltag \leq \pi_2\ \res\rtag \wedge \res\ltag =
      {\sf insert}\ h_1\ (\pi_1\ s_1) \wedge \res\rtag = {\sf
        insert}\ h_2\ (\pi_1\ s_2)
\]

\noindent
\underline{Proof of (\ref{goal:isort:let2:1})}: We start by applying
Theorem~\ref{thm:equivhol}. This yields three subgoals in HOL,
corresponding to the three conjuncts in $\phi'$:

\begin{equation}\label{goal:isort:let2:1:1}
\pi_2({\sf insert}\ h_1\ (\pi_1\ s_1)) \leq \pi_2({\sf
  insert}\ h_2\ (\pi_1\ s_2))
\end{equation}

\begin{equation}\label{goal:isort:let2:1:2}
{\sf insert}\ h_1\ (\pi_1\ s_1) = {\sf insert}\ h_1\ (\pi_1\ s_1)
\end{equation}

\begin{equation}\label{goal:isort:let2:1:3}
{\sf insert}\ h_2\ (\pi_1\ s_2) = {\sf insert}\ h_2\ (\pi_1\ s_2)
\end{equation}

(\ref{goal:isort:let2:1:2}) and (\ref{goal:isort:let2:1:3}) are
trivial, so we only have to prove (\ref{goal:isort:let2:1:1}). Since
$s_1 = {\sf isort}\ t_1$ and $s_2 = {\sf isort}\ t_2$ are conjuncts in
the assumption $\phi\defsubst{s_1}{s_2}$, (\ref{goal:isort:let2:1:1})
is equivalent to:

\begin{equation}\label{goal:isort:let2:1:1.1}
\pi_2({\sf insert}\ h_1\ (\pi_1({\sf isort}\ t_1)))
\leq \pi_2({\sf insert}\ h_2\ (\pi_1({\sf isort}\ t_2)))
\end{equation}

To prove this, we split cases on the shapes of $\pi_1({\sf
  isort}\ t_1)$ and $\pi_1({\sf isort}\ t_2)$. From the conjuncts in
$\phi\defsubst{s_1}{s_2}$, it follows immediately that $|\pi_1({\sf
  isort}\ t_1)| = |\pi_1({\sf isort}\ t_2)|$. Hence, only two cases
apply:

\noindent
Case: $\pi_1({\sf isort}\ t_1) = \pi_1({\sf isort}\ t_2) = []$. In
this case, by direct computation, $\pi_2({\sf
  insert}\ h_1\ (\pi_1({\sf isort}\ t_1))) = \pi_2({\sf
  insert}\ h_1\ []) = \pi_2([h_1], 0) = 0$. Similarly, and $\pi_2({\sf
  insert}\ h_2\ (\pi_1({\sf isort}\ t_2))) = 0$. So, the result
follows trivially.

\noindent Case: $\pi_1({\sf isort}\ t_1) = \cons{h_1'}{t_1'}$ and
$\pi_1({\sf isort}\ t_2)=\cons{h_2'}{t_2'}$. We first argue that $h_1
\leq h_1'$. Note that from the second and fifth conjuncts in
$\phi\defsubst{s_1}{s_2}$, it follows that ${\sf lmin}(t_1) = {\sf
  lmin}(\pi_1({\sf isort}\ t_1))$. Since $\pi_1({\sf isort}\ t_1) =
\cons{h_1'}{t_1'}$, we further get ${\sf lmin}(t_1) = {\sf
  lmin}(\pi_1({\sf isort}\ t_1)) = {\sf lmin}(\cons{h_1'}{t_1'}) =
     {\sf min}(h_1', {\sf lmin}(t_1')) \leq h_1'$. Finally, from the
     axiomatic definition of ${\sf sorted}(x_1)$ and $x_1 =
     \cons{h_1}{t_1}$, we derive $h_1 \leq {\sf
       lmin}(t_1)$. Combining, we get $h_1 \leq {\sf lmin}(t_1) \leq
     h_1'$.

Next, $\pi_2({\sf insert}\ h_1\ (\pi_1({\sf isort}\ t_1))) =
\pi_2({\sf insert}\ h_1\ (\cons{h_1'}{t_1'}))$. Expanding the
definition of ${\sf insert}$ and using $h_1 \leq h_1'$, we immediately
get $\pi_2({\sf insert}\ h_1\ (\pi_1({\sf isort}\ t_1))) = \pi_2({\sf
  insert}\ h_1\ (\cons{h_1'}{t_1'})) =
\pi_2(\cons{h_1}{\cons{h_1'}{t_1'}}, 1) = 1$. On the other hand, it is
fairly easy to prove (by case analyzing the result of $h_2 \leq h_2'$)
that $\pi_2({\sf insert}\ h_2\ (\pi_1({\sf isort}\ t_2))) = \pi_2({\sf
  insert}\ h_2\ (\cons{h_2'}{t_2'})) \geq 1$. Hence, $\pi_2({\sf
  insert}\ h_1\ (\pi_1({\sf isort}\ t_1))) = 1 \leq \pi_2({\sf
  insert}\ h_2\ (\pi_1({\sf isort}\ t_2)))$. 

\noindent
This proves (\ref{goal:isort:let2:1:1.1}) and, hence,
(\ref{goal:isort:let2:1:1}) and (\ref{goal:isort:let2:1}).\\

\noindent
\underline{Proof of (\ref{goal:isort:let2:2})}: By
Theorem~\ref{thm:equivhol}, it suffices to show the following in HOL,
under the assumptions of (\ref{goal:isort:let2:2}): \\

\begin{equation}\label{goal:isort:let2:2:1}
\pi_2(\pi_1 \ s_1', (\pi_2\ s_1) + (\pi_2\ s_1')) \leq \pi_2(\pi_1
\ s_2', (\pi_2\ s_2) + (\pi_2\ s_2'))
\end{equation}

\begin{equation}\label{goal:isort:let2:2:2}
(\pi_1 \ s_1', (\pi_2\ s_1) + (\pi_2\ s_1')) = {\sf isort}\ x_1
\end{equation}

\begin{equation}\label{goal:isort:let2:2:3}
(\pi_1 \ s_2', (\pi_2\ s_2) + (\pi_2\ s_2')) = {\sf isort}\ x_2
\end{equation}

By computation, (\ref{goal:isort:let2:2:1}) is equivalent to
$(\pi_2\ s_1) + (\pi_2\ s_1') \leq (\pi_2\ s_2) +
(\pi_2\ s_2')$. Using the definition of $\phi$, it is easy to see that
$\pi_2\ s_1 \leq \pi_2\ s_2$ is a conjunct in the assumption
$\phi\defsubst{s_1}{s_2}$. Similarly, using the definition of $\phi'$,
$\pi_2\ s_1' \leq \pi_2\ s_2'$ is a conjunct in the assumption
$\phi'\defsubst{s_1'}{s_2'}$. (\ref{goal:isort:let2:2:1}) follows
immediately from these.

To prove (\ref{goal:isort:let2:2:2}), note that since $x_1 =
\cons{h_1}{t_1}$, expanding the definition of ${\sf isort}$, we get
\[{\sf isort}\ x_1 =
(\pi_1({\sf insert}\ h_1\ (\pi_1 ({\sf isort}\ t_1))),
\pi_2({\sf isort}\ t_1) + \pi_2({\sf insert}\ h_1\ (\pi_1({\sf isort}\ t_1))))
\]
Matching with the left side of (\ref{goal:isort:let2:2:2}), it
suffices to show that $s_1' = {\sf insert}\ h_1\ (\pi_1 ({\sf
  isort}\ t_1))$ and $s_1 = {\sf isort}\ t_1$. These are immediate:
$s_1 = {\sf isort}\ t_1$ is a conjunct in the assumption
$\phi\defsubst{s_1}{s_2}$, while $s_1' = {\sf insert}\ h_1\ (\pi_1
({\sf isort}\ t_1))$ follows trivially from this and the conjunct
$s_1' = {\sf insert}\ h_1\ (\pi_1\ s_1)$ in
$\phi'\defsubst{s_1'}{s_2'}$. This proves (\ref{goal:isort:let2:2:2}).

The proof of (\ref{goal:isort:let2:2:3}) is similar to that of
(\ref{goal:isort:let2:2:2}). 

\noindent
This proves (\ref{goal:isort:let2:2}) and, hence,
(\ref{goal:isort:cons:2}).
\end{proof}



\section{Full RHOL rules}

The full set of RHOL rules is in the following figures:

\begin{figure*}[h]
\begin{mdframed}
\small
\centering
$
\infer[\sf ABS]
      {\jrhol{\Gamma}{\Psi}{\lambda x_1. t_1}{\tau_1 \to \sigma_1}{\lambda x_2. t_2}{\tau_2\to \sigma_2}{\forall x_1,x_2. \phi' \Rightarrow \phi\subst{\res\ltag}{\res\ltag\ x_1}\subst{\res\rtag}{\res\rtag\ x_2}}}
      {\jrhol{\Gamma,x_1:\tau_1,x_2:\tau_2}{\Psi,\phi'}{t_1}{\sigma_1}{t_2}{\sigma_2}{\phi}}
$
\\ \vspace{5mm}
$
\infer[\sf APP]{\jrhol{\Gamma}{\Psi}{t_1 u_1}{\sigma_1}{t_2 u_2}{\sigma_2}{\phi\subst{x_1}{u_1}\subst{x_2}{u_2}}}
      {\begin{array}{c}
\jrhol{\Gamma}{\Psi}{t_1}{\tau_1\to \sigma_1}{t_2}{\tau_2\to \sigma_2}{
          \forall x_1,x_2. \phi'\subst{\res\ltag}{x_1}\subst{\res\rtag}{x_2}\Rightarrow \phi\subst{\res\ltag}{\res\ltag\ x_1}\subst{\res\rtag}{\res\rtag\ x_2}}\\
\jrhol{\Gamma}{\Psi}{u_1}{\tau_1}{u_2}{\tau_2}{
          \phi'}
\end{array}}
$
\\ \vspace{5mm}
 $
 \hfill
 \infer[\sf ZERO]
        {\jrhol{\Gamma}{\Psi}{0}{\nat}{0}{\nat}{\phi}}
        {
         \jholn{\Gamma}{\Psi}{\phi\defsubst{0}{0}}
         }
 $
 \hfill
 $
 \infer[\sf SUCC]
        {\jrhol{\Gamma}{\Psi}{S t_1}{\nat}{S t_2}{\nat}{\phi}}
        {\begin{array}{c}
         \jrhol{\Gamma}{\Psi}{t_1}{\nat}{t_2}{\nat}{\phi'} \\
         \jholn{\Gamma}{\Psi}{\forall x_1 x_2 \phi'\defsubst{x_1}{x_2}
              \Rightarrow \phi\defsubst{S x_1}{S x_2}}
         \end{array}}
 $
 \\ \vspace{5mm}
$
\infer[\sf VAR]{\jrhol{\Gamma}{\Psi}{x_1}{\sigma_1}{x_2}{\sigma_2}{\phi}}
      {\jhol{\Gamma}{\Psi}{\phi\subst{\res\ltag}{x_1}\subst{\res\rtag}{x_2}}
       &
       \jlc{\Gamma}{x_1}{\sigma_1}
       &
       \jlc{\Gamma}{x_1}{\sigma_1}}
$
\hspace{2cm}
      $
\infer[\sf TRUE]
       {\jrhol{\Gamma}{\Psi}{\tbool}{\bool}{\tbool}{\bool}{\phi}}
       {
        \jholn{\Gamma}{\Psi}{\phi\defsubst{\tbool}{\tbool}}
        }
$
\\ \vspace{5mm}
$
\infer[\sf FALSE]
       {\jrhol{\Gamma}{\Psi}{\fbool}{\bool}{\fbool}{\bool}{\phi}}
       {
        \jholn{\Gamma}{\Psi}{\phi\defsubst{\fbool}{\fbool}}
        }
$
\hspace{4cm}
$
\infer[\sf NIL]
       {\jrhol{\Gamma}{\Psi}{[]}{\listt{\sigma_1}}{[]}{\listt{\sigma_2}}{\phi}}
       {
        \jholn{\Gamma}{\Psi}{\phi\defsubst{[]}{[]}}
        }
$
\\ \vspace{5mm}
$
\infer[\sf CONS]
       {\jrhol{\Gamma}{\Psi}{\cons{h_1}{t_1}}{\listt{\sigma_1}}{\cons{h_2}{t_2}}{\listt{\sigma_2}}{\phi}}
       {\begin{array}{c}
        \jrhol{\Gamma}{\Psi}{h_1}{\sigma_1}{h_2}{\sigma_2}{\phi'} \hspace{2cm}
        \jrhol{\Gamma}{\Psi}{t_1}{\listt{\sigma_1}}{t_2}{\listt{\sigma_2}}{\phi''} \\
        \jholn{\Gamma}{\Psi}{\forall x_1 x_2 y_1 y_2. \phi'\defsubst{x_1}{x_2} \Rightarrow \phi''\defsubst{y_1}{y_2}
             \Rightarrow \phi\defsubst{\cons{x_1}{y_1}}{\cons{x_2}{y_2}}}
       \end{array}}
$
\\ \vspace{5mm}
$
\infer[\sf PAIR]
       {\jrhol{\Gamma}{\Psi}{\pair{t_1}{u_1}}{\sigma_1\times\tau_1}{\pair{t_2}{u_2}}{\sigma_2\times\tau_2}{\phi}}
       {\begin{array}{c}
        \jrhol{\Gamma}{\Psi}{t_1}{\sigma_1}{t_2}{\sigma_2}{\phi'} \hspace{2cm}
        \jrhol{\Gamma}{\Psi}{u_1}{\tau_1}{u_2}{\tau_2}{\phi''} \\
        \jholn{\Gamma}{\Psi}{\forall x_1 x_2 y_1 y_2. \phi'\defsubst{x_1}{x_2} \Rightarrow \phi''\defsubst{y_1}{y_2}
             \Rightarrow \phi\defsubst{\pair{x_1}{y_1}}{\pair{x_2}{y_2}}}
        \end{array}}
$
\\ \vspace{5mm}
$
\infer[\sf PROJ_i]
       {\jrhol{\Gamma}{\Psi}{\pi_i(t_1)}{\sigma_1}{\pi_i(t_2)}{\sigma_2}{\phi}}
       {\jrhol{\Gamma}{\Psi}{t_1}{\sigma_1\times\tau_1}{t_2}{\sigma_2\times\tau_2}
       {\phi\defsubst{\pi_i(\res\ltag)}{\pi_i(\res\rtag)}}} 
$
\end{mdframed}
\caption{Core two-sided rules}
\end{figure*}

\begin{figure*}
\begin{mdframed}
\[
\infer[\sf SUB]{\jrhol{\Gamma}{\Psi}{t_1}{\sigma_1}{t_2}{\sigma_2}{\phi}}
      {\jrhol{\Gamma}{\Psi}{t_1}{\sigma_1}{t_2}{\sigma_2}{\phi'} &
       \Gamma \mid \Psi \vdash_{\sf HOL} \phi'\defsubst{t_1}{t_2} \Rightarrow \phi\defsubst{t_1}{t_2}}
\]
\\
\[
\infer[\sf \wedge_I]
      {\jrhol{\Gamma}{\Psi'}{t_1}{\sigma_2}{t_2}{\sigma_2}{\phi \wedge \phi'}}
      {\jrhol{\Gamma}{\Psi'}{t_1}{\sigma_2}{t_2}{\sigma_2}{\phi}
       & \jrhol{\Gamma}{\Psi'}{t_1}{\sigma_2}{t_2}{\sigma_2}{\phi'}}
\]
\\
\[
\infer[\sf \Rightarrow_I]
      {\jrhol{\Gamma}{\Psi'}{t_1}{\sigma_2}{t_2}{\sigma_2}{\phi' \Rightarrow \phi}}
      {\jrhol{\Gamma}{\Psi',\phi'\defsubst{t_1}{t_2}}{t_1}{\sigma_2}{t_2}{\sigma_2}{\phi}}
\]
\\
\[
\infer[\sf UHOL-L]
      {\jrhol{\Gamma}{\Psi}{t_1}{\sigma_1}{t_2}{\sigma_1}{\phi}}
      {\juhol{\Gamma}{\Psi}{t_1}{\sigma_1}{\phi\defsubst{\res}{t_2}}}
\]
\end{mdframed}                
\caption{Structural rules}
\end{figure*}

\begin{figure*}
\begin{mdframed}
\small
\centering
$
\infer[\sf ABS{-}L]
      {\jrhol{\Gamma}{\Psi}{\lambda x_1. t_1}{\tau_1 \to \sigma_1}{t_2}{\sigma_2}{\forall x_1. \phi' \Rightarrow \phi \subst{\res\ltag}{\res\ltag\ x_1}}}
      {\jrhol{\Gamma,x_1:\tau_1}{\Psi, \phi'}{t_1}{\sigma_1}{t_2}{\sigma_2}{\phi}}
$
\\ \vspace{5mm}
$
\infer[\sf APP{-}L]
      {\jrhol{\Gamma}{\Psi}{t_1 u_1}{\sigma_1}{u_2}{\sigma_2}{\phi}}
      {\begin{array}{c}
          \jrhol{\Gamma}{\Psi}{t_1}{\tau_1\to \sigma_1}{u_2}{\sigma_2}{\forall x_1. \phi'\subst{\res\ltag}{x_1} \Rightarrow \phi\subst{\res\ltag}{\res\ltag\ x_1}}\\
          \juhol{\Gamma}{\Psi}{u_1}{\sigma_1}{\phi'\subst{x_1}{u_1}}
\end{array}}
$
\\ \vspace{5mm}
 $
 \hfill
 \infer[\sf ZERO{-}L]
        {\jrhol{\Gamma}{\Psi}{0}{\nat}{t_2}{\sigma_2}{\phi}}
        {\begin{array}{c}
           \jlc{\Gamma}{t_2}{\sigma_2} \\
           \jholn{\Gamma}{\Psi}{\phi\defsubst{0}{t_2}}
         \end{array}
        }
 $
 \hfill
 $
 \infer[\sf SUCC{-}L]
        {\jrhol{\Gamma}{\Psi}{S t_1}{\nat}{t_2}{\sigma_2}{\phi}}
        {\begin{array}{c}
        \jrhol{\Gamma}{\Psi}{t_1}{\nat}{t_2}{\sigma_2}{\phi'} \\
         \jholn{\Gamma}{\Psi}{\forall x_1 x_2 \phi'\defsubst{x_1}{x_2}
              \Rightarrow \phi\defsubst{S x_1}{x_2}}
         \end{array}}
 $
 \\ \vspace{5mm}
$
\infer[\sf TRUE-L]
       {\jrhol{\Gamma}{\Psi}{\tbool}{\bool}{t_2}{\sigma_2}{\phi}}
       {
        \jholn{\Gamma}{\Psi}{\phi\defsubst{\tbool}{t_2}} &
        \jlc{\Gamma}{t_2}{\sigma_2}
        }
$
\hspace{1cm}
$
\infer[\sf FALSE-L]
       {\jrhol{\Gamma}{\Psi}{\fbool}{\bool}{t_2}{\sigma_2}{\phi}}
       {
        \jholn{\Gamma}{\Psi}{\phi\defsubst{\fbool}{t_2}} &
        \jlc{\Gamma}{t_2}{\sigma_2}
        }
$
\\ \vspace{5mm}
$
\infer[\sf VAR{-}L]
      {\jrhol{\Gamma}{\Psi}{x_1}{\sigma_1}{t_2}{\sigma_2}{\phi}}
      {\phi\subst{\res\ltag}{x_1} \in \Psi &  \res\rtag\not\in\ FV(\phi) &
       \jlc{\Gamma}{t_2}{\sigma_2}}
$
\hspace{1cm}
$
\infer[\sf NIL{-}L]
       {\jrhol{\Gamma}{\Psi}{[]}{\listt{\sigma_1}}{t_2}{\sigma_2}{\phi}}
       {
        \jhol{\Gamma}{\Psi}{\phi\defsubst{[]}{t_2}} &
        \jlc{\Gamma}{t_2}{\sigma_2}
        }
$
\\ \vspace{5mm}
$
\infer[\sf CONS{-}L]
       {\jrhol{\Gamma}{\Psi}{\cons{h_1}{t_1}}{\listt{\sigma_1}}{t_2}{\sigma_2}{\phi}}
       {\begin{array}{c}
        \jrhol{\Gamma}{\Psi}{h_1}{\sigma_1}{t_2}{\sigma_2}{\phi'}\hspace{2cm}
        \jrhol{\Gamma}{\Psi}{t_1}{\listt{\sigma_1}}{t_2}{\sigma_2}{\phi''} \\
        \jholn{\Gamma}{\Psi}{\forall x_1 x_2 y_1. \phi'\defsubst{x_1}{x_2} \Rightarrow \phi''\defsubst{y_1}{x_2}
             \Rightarrow \phi\defsubst{\cons{x_1}{y_1}}{x_2}}
        \end{array}}
$
\\ \vspace{5mm}
$
\infer[\sf PAIR{-}L]
       {\jrhol{\Gamma}{\Psi}{\pair{t_1}{u_1}}{\sigma_1\times\tau_1}{t_2}{\sigma_2}{\phi}}
       {\begin{array}{c}
        \jrhol{\Gamma}{\Psi}{t_1}{\sigma_1}{t_2}{\sigma_2}{\phi'}\hspace{2cm}
        \jrhol{\Gamma}{\Psi}{u_1}{\tau_1}{t_2}{\sigma_2}{\phi''} \\
        \jholn{\Gamma}{\Psi}{\forall x_1 x_2 y_1. \phi'\defsubst{x_1}{x_2} \Rightarrow \phi''\defsubst{y_1}{x_2}
             \Rightarrow \phi\defsubst{\pair{x_1}{y_1}}{x_2}}
        \end{array}}
$
\\ \vspace{5mm}
$
\infer[\sf PROJ_1{-}L]
       {\jrhol{\Gamma}{\Psi}{\pi_1(t_1)}{\sigma_1}{t_2}{\sigma_2}{\phi}}
       {\jrhol{\Gamma}{\Psi}{t_1}{\sigma_1\times\tau_1}{t_2}{\sigma_2}
          {\phi\subst{\res\ltag}{\pi_1(\res\ltag)}}}
$
\end{mdframed}
\caption{Core one-sided rules}
\end{figure*}

\begin{figure*}
  \begin{mdframed}
\[
\infer[\sf BOOLCASE]
  {\jrhol{\Gamma}{\Psi}{\casebool{t_1}{u_1}{v_1}}{\sigma_1}{\casebool{t_2}{u_2}{v_2}}{\sigma_2}{\phi}}
  {\begin{array}{c}
      \jrhol{\Gamma}{\Psi}{t_1}{\bool}{t_2}{\bool}{(\res\ltag = \tbool \wedge \res\rtag = \tbool) \vee (\res\ltag = \fbool \wedge \res\rtag = \fbool)} \\
      \jrhol{\Gamma}{\Psi, t_1=\tbool, t_2=\tbool}{u_1}{\sigma_1}{u_2}{\sigma_2}{\phi} \\
      \jrhol{\Gamma}{\Psi,t_1 = \fbool, t_2 = \fbool}{v_1}{\sigma_1}{v_2}{\sigma_2}{\phi}
   \end{array}}
\]

\[
\infer[\sf NATCASE]
  {\jrhol{\Gamma}{\Psi}{\casenat{t_1}{u_1}{v_1}}{\sigma_1}{\casenat{t_2}{u_2}{v_2}}{\sigma_2}{\phi}}
  {\begin{array}{c}
      \jrhol{\Gamma}{\Psi}{t_1}{\nat}{t_2}{\nat}{\res\ltag = 0 \Leftrightarrow \res\rtag = 0} \\
      \jrhol{\Gamma}{\Psi, t_1=0, t_2=0}{u_1}{\sigma_1}{u_2}{\sigma_2}{\phi} \\
      \jrhol{\Gamma}{\Psi}{v_1}{\nat\to\sigma_1}{v_2}{\nat\to\sigma_2}{\forall x_1 x_2. t_1 = S x_1 \Rightarrow t_2 = S x_2 \Rightarrow \phi\defsubst{\res\ltag\ x_1}{\res\rtag\ x_2}}
   \end{array}}
\]

\[
\infer[\sf LISTCASE]
  {\jrhol{\Gamma}{\Psi}{\caselist{t_1}{u_1}{v_1}}{\sigma_1}{\caselist{t_2}{u_2}{v_2}}{\sigma_2}{\phi}}
  {\begin{array}{c}
      \jrhol{\Gamma}{\Psi}{t_1}{\listt{\tau_1}}{t_2}{\listt{\tau_2}}{\res\ltag = \nil \Leftrightarrow \res\rtag = \nil} \\
      \jrhol{\Gamma}{\Psi, t_1= \nil, t_2= \nil}{u_1}{\sigma_1}{u_2}{\sigma_2}{\phi} \\
      \jrhol{\Gamma}{\Psi}{v_1}{\tau_1\to\listt{\tau_1}\to\sigma_1}{v_2}{\tau_2\to\listt{\tau_2}\to\sigma_2}\\
            {\forall h_1 h_2 l_1 l_2. t_1 = \cons{h_1}{l_1} \Rightarrow t_2 = \cons{h_2}{l_2} \Rightarrow \phi\defsubst{\res\ltag\ h_1\ l_1}{\res\rtag\ h_2\ l_2}}
   \end{array}}
\]

  \end{mdframed}
\caption{Synchronous case rules}
\end{figure*}

\begin{figure*}
  \begin{mdframed}

\[
\infer[\sf BOOLCASE-L]
  {\jrhol{\Gamma}{\Psi}{\casebool{t_1}{u_1}{v_1}}{\sigma_1}{t_2}{\sigma_2}{\phi}}
  {\begin{array}{c}
      \jlc{\Gamma}{t_1}{\bool} \\
      \jrhol{\Gamma}{\Psi, t_1=\tbool}{u_1}{\sigma_1}{t_2}{\sigma_2}{\phi} \\
      \jrhol{\Gamma}{\Psi, t_1=\fbool}{v_1}{\sigma_1}{t_2}{\sigma_2}{\phi}
   \end{array}}
\]

\[
\infer[\sf NATCASE-L]
  {\jrhol{\Gamma}{\Psi}{\casenat{t_1}{u_1}{v_1}}{\sigma_1}{t_2}{\sigma_2}{\phi}}
  {\begin{array}{c}
      \jlc{\Gamma}{t_1}{\nat} \\
      \jrhol{\Gamma}{\Psi, t_1=0}{u_1}{\sigma_1}{t_2}{\sigma_2}{\phi} \\
      \jrhol{\Gamma}{\Psi}{v_1}{\nat\to\sigma_1}{t_2}{\sigma_2}{\forall x_1. t_1 = S x_1 \Rightarrow \phi\subst{\res\ltag}{\res\ltag\ x_1}}
   \end{array}}
\]

\[
\infer[\sf LISTCASE-L]
  {\jrhol{\Gamma}{\Psi}{\caselist{t_1}{u_1}{v_1}}{\sigma_1}{t_2}{\sigma_2}{\phi}}
  {\begin{array}{c}
      \jlc{\Gamma}{t_1}{\listt{\tau}} \\
      \jrhol{\Gamma}{\Psi, t_1=\nil}{u_1}{\sigma_1}{t_2}{\sigma_2}{\phi} \\
      \jrhol{\Gamma}{\Psi}{v_1}{\tau\to\listt{\tau}\to\sigma_1}{t_2}{\sigma_2}{\forall h_1 l_1. t_1 = \cons{h_1}{l_1} \Rightarrow \phi\subst{\res\ltag}{\res\ltag\ h_1\ l_1}}
   \end{array}}
\]

  \end{mdframed}
\caption{One-sided case rules}
\end{figure*}

\begin{figure*}
  \begin{mdframed}

\[
\infer[\sf BBCASE-A]
  {\jrhol{\Gamma}{\Psi}{\casebool{t_1}{u_1}{v_1}}{\sigma_1}{\casebool{t_2}{u_2}{v_2}}{\sigma_2}{\phi}}
  {\begin{array}{c}
      \jrhol{\Gamma}{\Psi}{t_1}{\bool}{t_2}{\bool}{\top} \\
      \jrhol{\Gamma}{\Psi, t_1=\tbool, t_2=\tbool}{u_1}{\sigma_1}{u_2}{\sigma_2}{\phi} \\
      \jrhol{\Gamma}{\Psi,t_1\neq \tbool, t_2 = \tbool}{v_1}{\sigma_1}{u_2}{\sigma_2}{\phi} \\
      \jrhol{\Gamma}{\Psi,t_1 = \tbool, t_2\neq \tbool}{u_1}{\sigma_1}{v_2}{\sigma_2}{\phi} \\
      \jrhol{\Gamma}{\Psi,t_1\neq \tbool, t_2\neq \tbool}{v_1}{\sigma_1}{v_2}{\sigma_2}{\phi}
   \end{array}}
\]

\[
\infer[\sf BNCASE-A]
  {\jrhol{\Gamma}{\Psi}{\casebool{t_1}{u_1}{v_1}}{\sigma_1}{\casenat{t_2}{u_2}{v_2}}{\sigma_2}{\phi}}
  {\begin{array}{c}
      \jrhol{\Gamma}{\Psi}{t_1}{\bool}{t_2}{\nat}{\top} \\
      \jrhol{\Gamma}{\Psi, t_1=\tbool, t_2=0}{u_1}{\sigma_1}{u_2}{\sigma_2}{\phi} \\
      \jrhol{\Gamma}{\Psi,t_1\neq \tbool, t_2 = 0}{v_1}{\sigma_1}{u_2}{\sigma_2}{\phi} \\
      \jrhol{\Gamma}{\Psi,t_1 = \tbool}{u_1}{\sigma_1}{v_2}{\nat\to\sigma_2}{\forall x_2. t_2 = S x_2 \Rightarrow \phi\subst{\res\rtag}{\res\rtag\ x_2}} \\
      \jrhol{\Gamma}{\Psi,t_1\neq \tbool}{v_1}{\sigma_1}{v_2}{\nat\to\sigma_2}{\forall x_2. t_2 = S x_2 \Rightarrow \phi\subst{\rtag}{\res\rtag\ x_2}}
   \end{array}}
\]

\[
\infer[\sf BLCASE-A]
  {\jrhol{\Gamma}{\Psi}{\casebool{t_1}{u_1}{v_1}}{\sigma_1}{\caselist{t_2}{u_2}{v_2}}{\sigma_2}{\phi}}
  {\begin{array}{c}
      \jrhol{\Gamma}{\Psi}{t_1}{\bool}{t_2}{\listt{\tau_2}}{\top} \\
      \jrhol{\Gamma}{\Psi, t_1=\tbool, t_2=\nil}{u_1}{\sigma_1}{u_2}{\sigma_2}{\phi} \\
      \jrhol{\Gamma}{\Psi,t_1\neq \tbool, t_2 = \nil}{v_1}{\sigma_1}{u_2}{\tau_2\to\listt{\tau_2}\to\sigma_2}{\phi} \\
      \jrhol{\Gamma}{\Psi,t_1 = \tbool}{u_1}{\sigma_1}{v_2}{\tau_2\to\listt{\tau_2}\to\sigma_2}{\forall h_2 l_2. t_2 = \cons{h_2}{l_2} \Rightarrow \phi\subst{\res\rtag}{\res\rtag\ h_2\ l_2}} \\
      \jrhol{\Gamma}{\Psi,t_1\neq \tbool}{v_1}{\sigma_1}{v_2}{\tau_2\to\listt{\tau_2}\to\sigma_2}{\forall h_2 l_2. t_2 = \cons{h_2}{l_2} \Rightarrow \phi\subst{\res\rtag}{\res\rtag\ h_2\ l_2}}
   \end{array}}
\]

\[
\infer[\sf NNCASE-A]
  {\jrhol{\Gamma}{\Psi}{\casenat{t_1}{u_1}{v_1}}{\sigma_1}{\casenat{t_2}{u_2}{v_2}}{\sigma_2}{\phi}}
  {\begin{array}{c}
      \jrhol{\Gamma}{\Psi}{t_1}{\nat}{t_2}{\nat}{\top} \\
      \jrhol{\Gamma}{\Psi, t_1=0, t_2=0}{u_1}{\sigma_1}{u_2}{\sigma_2}{\phi} \\
      \jrhol{\Gamma}{\Psi, t_2 = 0}{v_1}{\nat\to\sigma_1}{u_2}{\sigma_2}{\forall x_1. t_1 = S x_1 \Rightarrow \phi\subst{\res\ltag}{\res\ltag\ x_1}} \\
      \jrhol{\Gamma}{\Psi,t_1 = 0}{u_1}{\sigma_1}{v_2}{\nat\to\sigma_2}{\forall x_2. t_2 = S x_2 \Rightarrow \phi\subst{\res\rtag}{\res\rtag\ x_2}} \\
      \jrhol{\Gamma}{\Psi}{v_1}{\nat\to\sigma_1}{v_2}{\nat\to\sigma_2}{\forall x_1 x_2. t_1 = S x_1 \Rightarrow t_2 = S x_2 \Rightarrow \phi\defsubst{\res\ltag\ x_1}{\res\rtag\ x_2}}
   \end{array}}
\]

\[
\infer[\sf LLCASE-A]
  {\jrhol{\Gamma}{\Psi}{\caselist{t_1}{u_1}{v_1}}{\sigma_1}{\caselist{t_2}{u_2}{v_2}}{\sigma_2}{\phi}}
  {\begin{array}{c}
      \jrhol{\Gamma}{\Psi}{t_1}{\listt{\tau_1}}{t_2}{\listt{\tau_2}}{\top} \\
      \jrhol{\Gamma}{\Psi, t_1=\nil, t_2=\nil}{u_1}{\sigma_1}{u_2}{\sigma_2}{\phi} \\
      \jrhol{\Gamma}{\Psi, t_2 = \nil}{v_1}{\tau_1\to\listt{\tau_1}\to\sigma_1}{u_2}{\sigma_2}{\forall h_1 l_1. t_1 = \cons{h_1}{l_1} \Rightarrow \phi\subst{\res\ltag}{\res\ltag\ h_1\ l_1}} \\
      \jrhol{\Gamma}{\Psi,t_1 = \nil}{u_1}{\tau_1\to\listt{\tau_1}\to\sigma_1}{v_2}{\tau_2\to\listt{\tau_2}\to\sigma_2}
            {\\ \forall h_2. t_2 = \cons{h_2}{l_2} \Rightarrow \phi\subst{\res\rtag}{\res\rtag\ h_2\ l_2}} \\
      \jrhol{\Gamma}{\Psi}{v_1}{\tau_1\to\listt{\tau_1}\to\sigma_1}{v_2}{\tau_2\to\listt{\tau_2}\to\sigma_2}
            \\{\forall h_1 h_2 l_1 l_2. t_1 = \cons{h_1}{l_1} \Rightarrow t_2 = \cons{h_2}{l_2} \Rightarrow \phi\defsubst{\res\ltag\ h_1\ l_1}{\res\rtag\ h_2\ l_2}}
   \end{array}}
\]
  \end{mdframed}
\caption{Asynchronous case rules (selected)}
\end{figure*}

\begin{figure*}
  \begin{mdframed}

\[
\infer[\sf NATCASE*]
  {\jrhol{\Gamma}{\Psi}{\casenat{t_1}{u_1}{v_1}}{\sigma_1}{\casenat{t_2}{u_2}{v_2}}{\sigma_2}{\phi}}
  {\begin{array}{c}
      \jrhol{\Gamma}{\Psi}{t_1}{\nat}{t_2}{\nat}{\phi' \wedge \res\ltag = 0 \Leftrightarrow \res\rtag = 0} \\
      \jrhol{\Gamma}{\Psi, \phi'\defsubst{0}{0}}{u_1}{\sigma_1}{u_2}{\sigma_2}{\phi} \\
      \jrhol{\Gamma}{\Psi}{v_1}{\nat\to\sigma_1}{v_2}{\nat\to\sigma_2}{\forall x_1 x_2. \phi'\defsubst{S x_1}{S x_2} \Rightarrow \phi\defsubst{\res\ltag\ x_1}{\res\rtag\ x_2}}
   \end{array}}
\]

\[
\infer[\sf LISTCASE*]
  {\jrhol{\Gamma}{\Psi}{\caselist{t_1}{u_1}{v_1}}{\sigma_1}{\caselist{t_2}{u_2}{v_2}}{\sigma_2}{\phi}}
  {\begin{array}{c}
      \jrhol{\Gamma}{\Psi}{t_1}{\listt{\tau_1}}{t_2}{\listt{\tau_2}}{\phi' \wedge \res\ltag = \nil \Leftrightarrow \res\rtag = \nil} \\
      \jrhol{\Gamma}{\Psi, \phi'\defsubst{\nil}{\nil}}{u_1}{\sigma_1}{u_2}{\sigma_2}{\phi} \\
      \jrhol{\Gamma}{\Psi}{v_1}{\tau_1\to\listt{\tau_1}\to\sigma_1}{v_2}{\tau_2\to\listt{\tau_2}\to\sigma_2}\\
            {\forall h_1 h_2 l_1 l_2. \phi'\defsubst{\cons{h_1}{l_1}}{\cons{h_2}{l_2}} \Rightarrow \phi\defsubst{\res\ltag\ h_1\ l_1}{\res\rtag\ h_2\ l_2}}
   \end{array}}
\]
  \end{mdframed}
  \caption{Alternative case rules}
\end{figure*}

\begin{figure*}
  \begin{mdframed}

\[
\infer[\sf LETREC]
  {\jrhol{\Gamma}{\Psi}{{\rm letrec}\ f_1\ x_1\ = e_1}{I_1\to{\sigma_2}}{{\rm letrec}\ f_2\ x_2\ = e_2}{I_2\to{\sigma_2}}{\forall x_1 x_2. \phi' \Rightarrow \phi\defsubst{\res\ltag\ x_1}{\res\rtag\ x_2}}}
  {\begin{array}{c}
      \wdef{f_1}{x_1}{e_1} \:\: \wdef{f_2}{x_2}{e_2} \\
      \jrhol{\Gamma, x_1:I_1, x_2:I_2, f_1:I_1\to{\sigma}, f_2:I_2\to{\sigma_2}}
         {\Psi, \phi',\\ \forall m_1 m_2. (|m_1|, |m_2|) < (|x_1|, |x_2|) \Rightarrow \phi'\subst{x_1}{m_1}\subst{x_2}{m_2}\Rightarrow
                     \phi\subst{x_1}{m_1}\subst{x_2}{m_2}\defsubst{f_1\ m_1}{f_2\ m_2}}{\\ e_1}{\sigma_1}{e_2}{\sigma_2}{\phi}
   \end{array}}
\]

\[
\infer[\sf LETREC-L]
  {\jrhol{\Gamma}{\Psi}{{\rm letrec}\ f_1\ x_1\ = e_1}{I_1\to{\sigma_2}}{t_2}{\sigma_2}{\forall x_1 . \phi' \Rightarrow \phi\subst{\res\ltag}{\res\ltag\ x_1}}}
  {\begin{array}{c}
      \wdef{f_1}{x_1}{e_1} \\
      \jrhol{\Gamma, x_1:I_1, f_1:I_1\to{\sigma}}
         {\Psi, \phi', \\ \forall m_1. |m_1| < |x_1| \Rightarrow \phi'\subst{x_1}{m_1} \Rightarrow
               \phi\subst{x_1}{m_1}\subst{x_2}{m_2}\defsubst{f_1\ m_1}{t_2}}{e_1}{\sigma_1}{t_2}{\sigma_2}{\phi}
   \end{array}}
\]

\[\text{ where }I_1,I_2 \in \{\nat,\listt{\tau}\}\]

  \end{mdframed}
\caption{Recursion rules}
\end{figure*}

\end{document}